\def\complexNumbers{\mathbb{C}}

\def\integers{\mathbb{Z}}
\def\integersPositive{\mathbb{Z}^{+}}
\def\integersNonnegative{\mathbb{Z}_{0}^{+}}

\def\constante{{\rm e}}
\def\constantj{{\rm j}}
\def\exponentialBase{\xi}
\def\lengthPillar{L}

\def\peakCrossCorrelatation[#1]{\rho_{\text{peak}}#1}

\def\constante{{\rm e}}
\def\constantj{{\rm j}}
\def\exponentialBase{\xi}
\def\constanti{{\rm i}}
\def\constantMinusi{{\rm j}}
\def\constantOne{{\rm +}}
\def\constantMinusOne{{\rm -}}

\def\upsampleVarA{k}
\def\upsampleVarB{l}
\def\paddingVar{d}
\def\lagForCorrelation{\tau}

\def\indexEleOfSeq{\iota}
\def\indexIteration{n}

\def\indexIterationANF{l}

\def\indexFirstOrderMonomial{j}

\def\orderMonomial[#1]{k_{#1}}
\def\coeffientsANF[#1]{c_{#1}}
\def\polyVariable{z}
\def\polyVariableIt{w}
\def\numberOfIterations{m}
\def\numberOfPointsForPSK{H}
\def\modulationSymbolF[#1]{m_{#1}}
\def\cardinalitySetOfOperators[#1]{{H}_{#1}}

\def\numberSequencesDifferentThanLEngthOne{A}
\def\numberSequencesNotColinear{B}

\def\varMonomial{x}
\def\monomial[#1]{x_{#1}}
\def\varUpsample{U}

\def\lengthGaGb{N}
\def\lengthGcGd{M}
\def\lengthGcGdIterative[#1]{M_{#1}}
\def\lengthGcGdContIterative[#1]{\dot{M}_{#1}}

\def\scaleAexp[#1]{a_{#1}}
\def\scaleBexp[#1]{b_{#1}}
\def\scaleEexp[#1]{e_{#1}}
\def\angleexp[#1]{c_{#1}}
\def\angleexpAll[#1]{k_{#1}}
\def\angleScaleAexp[#1]{\dot{c}_{#1}}
\def\angleScaleBexp[#1]{\ddot{c}_{#1}}

\def\arbitraryPhaseA{c'}
\def\arbitraryPhaseB{c''}

\def\separationGolay[#1]{d_{#1}}

\def\eleGa[#1]{{a}_{#1}}
\def\eleGb[#1]{{b}_{#1}}
\def\elex[#1]{{x}_{#1}}
\def\eles[#1]{{s}_{#1}}
\def\eley[#1]{{y}_{#1}}

\def\apac[#1][#2]{\rho_{#1}(#2)}
\def\apacPositive[#1][#2]{\rho^{+}_{#1}(#2)}
\def\binaryAsignment[#1][#2]{b_{#1}^{(#2)}}

\def\eleSeqf[#1]{{f}_{#1}}
\def\eleSeqg[#1]{{g}_{#1}}
\def\eleSeqcf[#1]{{c}_{f,#1}}
\def\eleSeqcg[#1]{{c}_{g,#1}}

\def\scaleA[#1]{\alpha_{#1}}
\def\scaleB[#1]{\beta_{#1}}
\def\angleGolay[#1]{\omega_{#1}}
\def\angleScaleA[#1][#2]{\dot{\omega}_{#1}^{#2}}
\def\angleScaleB[#1][#2]{\ddot{\omega}_{#1}^{#2}}
\def\permutationShift[#1]{{\psi_{#1}}}
\def\permutationMono[#1]{{\pi_{#1}}}
\def\permutationUpsample[#1]{{\kappa_{#1}}}
\def\permutationOrderEle[#1]{{\phi_{#1}}}

\def\symbolDuration{T_{\rm s}}
\def\timeVar{t}

\def\cubicMetric{CM}

\def\vNorm{v_{\rm norm}(\timeVar)}
\def\vNormCubic{v^3_{\rm norm}(\timeVar)}

\def\thershold{\beta}
\def\upsampleValueU{u}
\def\seedA{\mathcal{S}_{\rm c}''}
\def\seedB{\mathcal{S}_{\rm d}''}
\def\seedAp{\mathcal{S}_{\rm c}'}
\def\seedBp{{\mathcal{S}_{\rm d}'}}
\def\seedSize{V}
\def\seedSizeEq{W}

\def\indexForFirst{\dot{u}}
\def\indexForSecond{\dot{v}}

\def\setC{\mathcal{C}}
\def\setD{\mathcal{D}}
\def\setQ{\mathcal{Q}_1}

\def\numberOfSequences{K}
\def\indexSequence{i}

\def\funczArbitrary[#1]{z(#1)}
\def\funcaArbitrary[#1]{a(#1)}
\def\funcbArbitrary[#1]{b(#1)}
\def\coefficientArbitrary[#1]{k_{#1}}

\def\indexRB{k}

\def\distanceToPoint[#1]{d_{#1}}
\def\ratioBetweenDistanceAndInner[#1]{l_{#1}}
\def\angleBetweenPointAndXaxis[#1]{\theta_{#1}}
\def\angleBetweenPointAndXYDiagonal[#1]{\psi_{#1}}
\def\angleBetweenPointAndSymPoint[#1]{\xi_{#1}}

\def\numberOfClusters{N_{\rm rb}}
\def\RBsize{N_{\rm sc}}
\def\numberOfNulls{N_{\rm null}}

\def\dftSize{N_{\rm IDFT}}

\def\vecArrangement[#1]{\textbf{b}_{#1}}

\def\seqGa{\textit{\textbf{a}}}
\def\seqGb{\textit{\textbf{b}}}
\def\seqGaIt[#1]{\textit{\textbf{a}}^{(#1)}}
\def\seqGbIt[#1]{\textit{\textbf{b}}^{(#1)}}
\def\seqGc{\textit{\textbf{c}}}
\def\seqGd{\textit{\textbf{d}}}
\def\seedSeqGcCandidate[#1]{{\seqGc_{#1}''}}
\def\seedSeqGdCandidate[#1]{{\seqGd_{#1}''}}
\def\seqGcCandidate[#1]{{\seqGc_{#1}'}}
\def\seqGdCandidate[#1]{{\seqGd_{#1}'}}
\def\seqGcRecursion[#1]{{\seqGc_{#1}}}
\def\seqGdRecursion[#1]{{\seqGd_{#1}}}
\def\seqGcContRecursion[#1]{{\dot{\seqGc}_{#1}}}
\def\seqGdContRecursion[#1]{{\dot{\seqGd}_{#1}}}

\def\seqGcTildeRecursion[#1]{{\tilde{\seqGc}_{#1}}}
\def\seqGdTildeRecursion[#1]{{\tilde{\seqGd}_{#1}}}

\def\seqPermutationShift{\bm{\psi}}
\def\seqPermutationCompShift{\bm{\pi}}
\def\seqPermutationOrder{\bm{\phi}}
\def\seqPermutationUpsample{\bm{\kappa}}

\def\seqGf{\textit{\textbf{t}}}
\def\seqGg{\textit{\textbf{r}}}
\def\seqGx{\textit{\textbf{x}}}

\def\seqGs{\textit{\textbf{s}}}
\def\seqGsShift[#1]{{\textit{\textbf{s}}}_{#1}}

\def\seqf{\textit{\textbf{f}}}

\def\seqSub[#1]{\textit{\textbf{h}}_{#1}}

\def\seqFirstOrderMonomial[#1]{\textit{\textbf{m}}_{#1}}
\def\seqx{\textit{\textbf{x}}}

\def\seqToBeModulated[#1]{\textit{\textbf{s}}_{#1}}
\def\seqCelep[#1][#2]{\phi_{#1#2}}
\def\seqDelep[#1][#2]{\theta_{#1#2}}

\def\flipConjugate[#1]{{{\tilde{#1}}}}
\def\expectationOperator[#1]{{\mathbb{E}}[#1]}
\def\operator[#1][#2]{\mathcal{O}_{#1}^{(#2)}}
\def\operatordot[#1][#2]{\bar{\mathcal{O}}_{#1}^{(#2)}}

\def\compositeOperatorF[#1][#2]{{F}_{#1}{(#2)}}
\def\compositeOperatorG[#1][#2]{{G}_{#1}{(#2)}}
\def\compositeOperatorFdot[#1][#2]{\bar{F}_{#1}{(#2)}}
\def\compositeOperatorGdot[#1][#2]{\bar{G}_{#1}{(#2)}}

\def\compositeOperatorFcd[#1][#2]{{F}_{{\rm cd\tilde{c}\tilde{d}},#1}{(#2)}}
\def\compositeOperatorGcd[#1][#2]{{G}_{{\rm cd\tilde{c}\tilde{d}},#1}{(#2)}}

\def\compositeOperatorFangle[#1][#2]{{F}_{{\rm comp},#1}{(#2)}}
\def\compositeOperatorGangle[#1][#2]{{G}_{{\rm comp},#1}{(#2)}}

\def\setOfOperators[#1]{{\mathfrak{J}}_{#1}}

\def\operatorBinary[#1][#2]{O_{#1}^{(#2)}}
\def\operatorSign[#1][#2]{{\rm S}_{#1}^{(#2)}}
\def\operatorScaleA[#1][#2]{{\rm{A}}_{#1}^{(#2)}}
\def\operatorScaleB[#1][#2]{{\rm{B}}_{#1}^{(#2)}}
\def\operatorAngle[#1][#2]{\Omega_{#1}^{(#2)}}
\def\operatorSeparation[#1][#2]{\Delta_{#1}^{(#2)}}
\def\operatorOrderA[#1][#2]{\dot{\rm O}_{#1}^{(#2)}}
\def\operatorOrderB[#1][#2]{\ddot{\rm O}_{#1}^{(#2)}}
\def\operatorAngleScaleA[#1][#2]{\dot{\Omega}_{#1}^{(#2)}}
\def\operatorAngleScaleB[#1][#2]{{\Omega}_{#1}^{(#2)}}
\def\operatorAngleConjScaleA[#1][#2]{\dot{\Upsilon}_{#1}^{(#2)}}
\def\operatorAngleConjScaleB[#1][#2]{\ddot{\Upsilon}_{#1}^{(#2)}}

\def\operatorOrderC[#1][#2]{{\rm C}_{#1}^{(#2)}}
\def\operatorOrderD[#1][#2]{{\rm D}_{#1}^{(#2)}}
\def\operatorOrderCtilde[#1][#2]{{\rm \tilde{C}}_{#1}^{(#2)}}
\def\operatorOrderDtilde[#1][#2]{{\rm \tilde{D}}_{#1}^{(#2)}}

\def\upsampleOp[#1][#2]{{\uparrow_{#1}\{#2\}}}
\def\shift[#1]{r_{#1}}
\def\operatorRMS[#1]{{\rm rms}\{#1\}}
\def\operationIDFT[#1][#2]{{\rm IDFT}\{#1,#2\}}
\def\operationMAX[#1]{{\rm max}\{#1\}}

\def\functionSpace{\mathcal{F}}
\def\functionf[#1]{p^{(#1)}}
\def\functiong[#1]{q^{(#1)}}
\def\functionh{r}
\def\functionfdot[#1]{\bar{p}_{\indexIterationANF}^{(#1)}}
\def\functiongdot[#1]{\bar{q}_{\indexIterationANF}^{(#1)}}

\def\funcfForCommonPhase{c_{\rm i}}
\def\funcfForCommonPhaseA{f_{\rm i}}
\def\funcfForCommonPhaseB{g_{\rm i}}
\def\funcfForCommonShift{f_{\rm s}}
\def\funcForCommonOrder{p_{\rm o}}
\def\funcfForCommonOrder{f_{\rm o}}
\def\funcgForCommonOrder{g_{\rm o}}

\def\funcfForTurynPart[#1]{f_{{\rm f},#1}}
\def\funcgForTurynPart[#1]{g_{{\rm f},#1}}

\def\funcfForFinalPhase{f_{\rm i}}

\def\funcfForANF{f}
\def\funcGfForANF[#1]{f_{#1}}
\def\funcGgForANF[#1]{g_{#1}}

\def\funcGfForC[#1]{f_{{\rm c}, #1}}
\def\funcGgForC[#1]{g_{{\rm c}, #1}}
\def\funcGfForD[#1]{f_{{\rm d}, #1}}
\def\funcGgForD[#1]{g_{{\rm d}, #1}}
\def\funcGfForCtilde[#1]{f_{{\rm \tilde{c}}, #1}}
\def\funcGgForCtilde[#1]{g_{{\rm \tilde{c}}, #1}}
\def\funcGfForDtilde[#1]{f_{{\rm \tilde{d}}, #1}}
\def\funcGgForDtilde[#1]{g_{{\rm \tilde{d}}, #1}}

\def\polySeq[#1][#2]{p_{#1}(#2)}

\newcommand\mydots{\hbox to 1em{.\hss.\hss.}}

\def\IEEEsubmission{0}
\def\reviewColor{black}
\if\IEEEsubmission1
	\documentclass[journal,12pt,onecolumn,draftclsnofoot,]{IEEEtran}
\else
	\documentclass[journal]{IEEEtran}
\fi

\usepackage{multicol}
\usepackage{lipsum}
\usepackage{mathtools}
\usepackage{amsthm}
\usepackage{acronym}
\usepackage{stfloats}
\usepackage{xcolor}
\usepackage{amsfonts}
\usepackage{acronym}
\usepackage{cite}
\usepackage{bm}
\usepackage{amsmath,amssymb}
\usepackage{graphicx}
\usepackage[english]{babel}
\usepackage[caption=false,font=footnotesize,skip=0pt]{subfig}
\usepackage[geometry]{ifsym}
\usepackage{array}
\usepackage[utf8]{inputenc}
\usepackage[T1]{fontenc}
\usepackage{array}
\usepackage{algorithm}
\usepackage{algpseudocode}

\makeatletter
\def\BState{\State\hskip-\ALG@thistlm}
\makeatother

\newcolumntype{L}[1]{>{\raggedright\let\newline\\\arraybackslash\hspace{0pt}}m{#1}}
\newcolumntype{C}[1]{>{\centering\let\newline\\\arraybackslash\hspace{0pt}}m{#1}}
\newcolumntype{R}[1]{>{\raggedleft\let\newline\\\arraybackslash\hspace{0pt}}m{#1}}

\usepackage{cuted}
\makeatletter
\newif\ifAC@uppercase@first%
\def\Aclp#1{\AC@uppercase@firsttrue\aclp{#1}\AC@uppercase@firstfalse}%
\def\AC@aclp#1{%
	\ifcsname fn@#1@PL\endcsname%
	\ifAC@uppercase@first%
	\expandafter\expandafter\expandafter\MakeUppercase\csname fn@#1@PL\endcsname%
	\else%
	\csname fn@#1@PL\endcsname%
	\fi%
	\else%
	\AC@acl{#1}s%
	\fi%
}%
\def\Acp#1{\AC@uppercase@firsttrue\acp{#1}\AC@uppercase@firstfalse}%
\def\AC@acp#1{%
	\ifcsname fn@#1@PL\endcsname%
	\ifAC@uppercase@first%
	\expandafter\expandafter\expandafter\MakeUppercase\csname fn@#1@PL\endcsname%
	\else%
	\csname fn@#1@PL\endcsname%
	\fi%
	\else%
	\AC@ac{#1}s%
	\fi%
}%
\def\Acfp#1{\AC@uppercase@firsttrue\acfp{#1}\AC@uppercase@firstfalse}%
\def\AC@acfp#1{%
	\ifcsname fn@#1@PL\endcsname%
	\ifAC@uppercase@first%
	\expandafter\expandafter\expandafter\MakeUppercase\csname fn@#1@PL\endcsname%
	\else%
	\csname fn@#1@PL\endcsname%
	\fi%
	\else%
	\AC@acf{#1}s%
	\fi%
}%
\def\Acsp#1{\AC@uppercase@firsttrue\acsp{#1}\AC@uppercase@firstfalse}%
\def\AC@acsp#1{%
	\ifcsname fn@#1@PL\endcsname%
	\ifAC@uppercase@first%
	\expandafter\expandafter\expandafter\MakeUppercase\csname fn@#1@PL\endcsname%
	\else%
	\csname fn@#1@PL\endcsname%
	\fi%
	\else%
	\AC@acs{#1}s%
	\fi%
}%
\edef\AC@uppercase@write{\string\ifAC@uppercase@first\string\expandafter\string\MakeUppercase\string\fi\space}%
\def\AC@acrodef#1[#2]#3{%
	\@bsphack%
	\protected@write\@auxout{}{%
		\string\newacro{#1}[#2]{\AC@uppercase@write #3}%
	}\@esphack%
}%
\def\Acl#1{\AC@uppercase@firsttrue\acl{#1}\AC@uppercase@firstfalse}
\def\Acf#1{\AC@uppercase@firsttrue\acf{#1}\AC@uppercase@firstfalse}
\def\Ac#1{\AC@uppercase@firsttrue\ac{#1}\AC@uppercase@firstfalse}
\def\Acs#1{\AC@uppercase@firsttrue\acs{#1}\AC@uppercase@firstfalse}

\newtheorem{theorem}{Theorem}

\acrodef{AWGN}{additive white Gaussian noise}
\acrodef{PAPR}{peak-to-average-power ratio}
\acrodef{APAC}{aperiodic auto correlation}
\acrodef{OFDM}{orthogonal frequency division multiplexing}
\acrodef{OFDMA}{orthogonal frequency division multiple access}
\acrodef{DFT}{discrete Fourier transform}
\acrodef{IDFT}{inverse discrete Fourier transform}
\acrodef{DC}{direct current}
\acrodef{CS}{complementary sequence}
\acrodef{GCP}{Golay complementary pair}
\acrodef{ANF}{algebraic normal form}
\acrodef{PSK}{phase-shift keying}
\acrodef{QAM}{quadrature-amplitude modulation}
\acrodef{QPSK}{quadrature phase-shift keying}
\acrodef{BPSK}{binary phase-shift keying}
\acrodef{GDJ}{Golay-Davis-Jedwab}
\acrodef{PMEPR}{peak-to-mean envelope power ratios}
\acrodef{FFT}{fast Fourier transform}
\acrodef{BER}{bit-error rate}
\acrodef{SNR}{signal-to-noise ratio}
\acrodef{4G}{Fourth Generation}
\acrodef{5G}{Fifth Generation}
\acrodef{NR}{New Radio}
\acrodef{LTE}{Long-Term Evolution}
\acrodef{PTS}{partial transmit sequences}
\acrodef{PSD}{power spectral density}
\acrodef{LDPC}{low-density parity check}
\acrodef{OCB}{occupied channel bandwidth}
\acrodef{CP}{cyclic prefix}
\acrodef{CM}{cubic metric}
\acrodef{DAC}{digital-to-analog converter}
\acrodef{RS}{reference symbol}
\acrodef{LAA}{license-assisted access}
\acrodef{eLAA}{enhanced licensed-assisted access}
\acrodef{PRB}{physical resource block}
\acrodef{OCB}{occupied channel bandwidth}
\acrodef{NR-U}{NR in unlicensed spectrum}
\acrodef{IMD}{inter-modulation distortion}
\acrodef{ZC}{Zadoff-Chu}
\acrodef{SR}{scheduling request}
\acrodef{i.i.d.}{independent-and-identically distributed}
\acrodef{NC}{non-contiguous}
\acrodef{RM}{Reed-Muller}
\acrodef{OCC}{orthogonal cover code}
\acrodef{MMSE}{minimum mean square error}
\acrodef{ML}{maximum-likelihood}
\acrodef{BLER}{block-error rate}
\acrodef{UCI}{uplink control information}
\acrodef{CSI}{channel state information}
\acrodef{PUCCH}{physical uplink control channel}
\acrodef{UL}{uplink}
\acrodef{IFDMA}{IFDMA}
\acrodef{DTX}{discontinuous transmission}
\acrodef{ACK}{acknowlegment}
\acrodef{NACK}{negative acknowlegment}
\acrodef{MAI}{multiple-access interference}
\acrodef{CCI}{co-channel interference}
\acrodef{MRC}{maximum-ratio combining}
\acrodef{FDE}{frequency-domain equalization}
\begin{document}
\title{ 
An Uplink Control Channel Design with Complementary Sequences for Unlicensed Bands
}
\author{Alphan~\c{S}ahin,~\IEEEmembership{Member,~IEEE,}
	and~Rui~Yang,~\IEEEmembership{Member,~IEEE}
	\thanks{{\color{\reviewColor}This paper  was presented in part at the IEEE International Conference on Communications (ICC) 2019 \cite{sahin_2019icc}.}
		
		 Alphan~\c{S}ahin and Rui~Yang are affiliated with University of South Carolina, Columbia, SC and InterDigital, Huntington Quadrangle, Melville, NY, respectively. E-mail: asahin@mailbox.sc.edu, rui.yang@interdigital.com
}}
\maketitle

\begin{abstract}
\color{\reviewColor}
In this paper, two modulation schemes based on \acp{CS} are proposed for uplink control channels in unlicensed bands. These schemes address high \ac{PAPR} under non-contiguous resource allocation in the frequency domain and reduce the maximum \ac{PAPR} to 3 dB. The first scheme allows the users to transmit a small amount of \ac{UCI} such as acknowledgment signals and does not introduce a trade-off between \ac{PAPR} and \ac{CCI}. The second scheme, which enables up to 21 \ac{UCI} bits for a single user or 11 \ac{UCI} bits for three users in an interlace, is based on a new theorem introduced in this paper. This theorem leads distinct \acp{CS} compatible with a wide variety of resource allocations while capturing the inherent relationship between \acp{CS} and \ac{RM}  codes, which makes \acp{CS} more useful for practical systems. The numerical results show that the proposed schemes maintain the low-\ac{PAPR} benefits without increasing the error rate for non-contiguous resource allocations in the frequency domain.

\end{abstract}


\acresetall

\section{Introduction}

{\color{\reviewColor}
To improve overall network efficiency and address the rapid increase in data demand, the wireless industry has started to extend 3GPP \ac{LTE} and \ac{5G} \ac{NR}  for the operation in unlicensed bands \cite{Lagen_2020, Labib_2017, Aijaz_2013}. However, the communication protocols designed for licensed bands need major changes as coexistence assurance is required in the unlicensed bands. To ensure fairness of channel access and usage among different radio access technologies},  stringent regulatory requirements are imposed on unlicensed bands. For example, according to the ETSI regulations \cite{etsi_2017},  in the 5 GHz band, the \ac{PSD} of the transmitted signal should be less than $10$~dBm/MHz and the \ac{OCB} should be larger than  $80\%$ of the nominal channel bandwidth. Therefore, a narrow bandwidth transmission (e.g., a single \ac{PRB} is $180$~kHz in \ac{LTE}) in $20$~MHz channel in a $5$~GHz unlicensed band does not meet the \ac{OCB} requirement and limits the coverage range due to the \ac{PSD} requirement. To be able to increase the transmit power under the \ac{PSD} constraint while complying with the \ac{OCB} requirement, 3GPP \ac{LTE} \ac{eLAA} and \ac{NR-U} have adopted {\em interlaced transmission} which allocates  disperse  and non-contiguous \acp{PRB} {\color{\reviewColor} as shown in \figurename~\ref{fig:interlace}, called {\em interlace}, in the \ac{UL} \cite{nr_phy_2020}. This major change on the baseline resource allocation prohibits the use of single-carrier waveform (e.g., \ac{DFT}-spread \ac{OFDM}) and the corresponding physical channels that benefit from low \ac{PAPR} in the licensed bands. In addition, the number of utilized \acp{PRB} in an interlace in \ac{NR} is a function of subcarrier spacing and  bandwidth, which makes the problem more challenging. 

\begin{figure}[t]
	\centering
	{\includegraphics[width =2.8in]{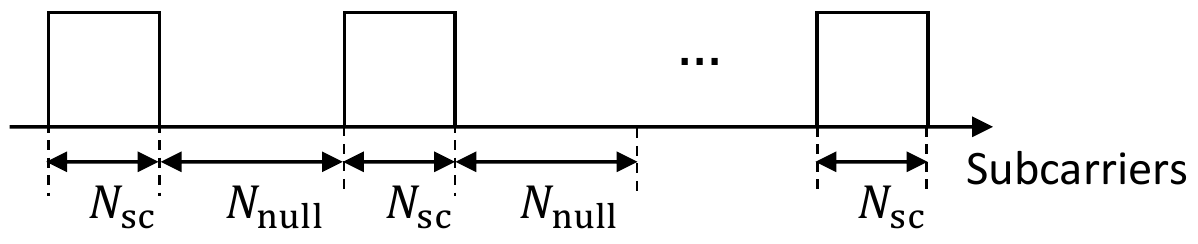}
	}
	\caption{Interlace model.}
	\label{fig:interlace}
\end{figure}

In 3GPP  5G \ac{NR} R15 \cite{nr_phy_2017}, the \ac{PUCCH} is meticulously designed to ensure the link reliability while handling multiple users with very limited resources for licensed bands. It consists of five different formats. Format~0 is based on sequence selection and designed for 1 or 2 \ac{UCI} bits such as acknowledgment (i.e., \ac{ACK} and \ac{NACK} signals) or \ac{SR}. It does not include reference symbols and shares the same structural properties of non-coherent orthogonal signaling \cite{Kundu_2018}. It occupies a single \ac{PRB} while allowing 6 users to share the same \ac{PRB}.  In \cite{Kundu_2018}, an alternative design to Format 0 with reference symbols is discussed. However, no benefit of using reference symbols is observed. Format~1 extends Format~0 to 4-14 \ac{OFDM} symbols  with an \ac{OCC}. It is based on sequence modulation where a \ac{BPSK} or \ac{QPSK} symbol is multiplied with a sequence and includes reference symbols. It supports up to 84 users and provides an enhanced coverage range. For Format~0-1,  the reliability is primarily ensured by a set of low-cross correlation seed sequences, in which each sequence results in an \ac{OFDM} symbol with low \ac{PAPR}. Orthogonal sequences are generated through cyclic shifts in time by exploiting the properties of unimodular sequences \cite{Benedetto_2009}. Format~2-4 support moderate and large \ac{UCI} payloads. In Format~2, the data symbols and reference symbols are directly mapped to the subcarriers. It occupies 1-2 \ac{OFDM} symbols and 1-16 \acp{PRB}. For Format~3-4, the waveform is based on \ac{DFT}-spread~\ac{OFDM} to reduce \ac{PAPR}. Format~3 supports a large payload with 1-16 \acp{PRB} with no user multiplexing capacity in the same \ac{PRB}. On the other hand,  Format~4 is limited to a single \ac{PRB}, but it supports up to 4 users in the same \ac{PRB} with pre-DFT \ac{OCC} (See Figure 11 in \cite{qualcommNRUproposal}). In 3GPP 5G \ac{NR} R15, which format is used is determined by the number of assigned symbols and the number of \ac{UCI} bits to be transmitted. The channel coding is also determined based on the number of \ac{UCI} bits. While a polar code is adopted for more than 11 bits, a  (32,11) linear block code defined in Table~5.3.3.3-1 in \cite{nr_coding_2020} is utilized for 3-11 \ac{UCI} bits.

The  \ac{PUCCH} formats are extended to the interlaced transmission for unlicensed bands in 3GPP  5G \ac{NR} R16 \cite{nr_phy_2020}. To avoid major modifications in the standard, \ac{PAPR} reduction methods relying on randomization are adopted. For Format 0 and 1, a resource-block dependent sequence generation is utilized, called {\em cycle-shift hopping}. The cyclic shift used for each \ac{PRB} in the interlace is determined as a function of the \ac{PRB} index \cite{ericssonNRUneedPf0}. Therefore, the coherent additions of the peak samples in time for the same signal component on different \acp{PRB} are avoided. Similarly, Format 2 is extended with \ac{OCC}-cycling across \acp{PRB} of an interlace, i.e., a user uses different spreading coefficients for different \acp{PRB}. By capturing the user multiplexing feature of Format 4,  Format 3 is extended with a pre-\ac{DFT} \ac{OCC} with block-wise repetition followed by mapping over the whole interlace in the frequency domain \cite{qualcommNRUproposal}. We refer the reader to \cite{ericssonNRUproposal} and \cite{featureLead} for several other non-standard solutions and discussions for \ac{NR} \ac{PUCCH} in the unlicensed band.

The literature is rich with \ac{PAPR} reduction methods for \ac{OFDM} \cite{Rahmatallah_2013}, \cite{Wunder_2013}. However, low-complexity methods which do not require optimization for each \ac{OFDM} symbol and tailored for non-contiguous allocation are limited. With \ac{DFT} precoding and an interleaved subcarrier mapping, an \ac{OFDM} symbol is converted to a low-\ac{PAPR} single-carrier waveform with repetitions in time \cite{Myung_2006}. However, an interleaved subcarrier mapping is not compatible with the interlaces in \ac{NR} and \ac{LTE}. In \cite{davis_1999}, Davis and Jedwab showed that there exists a joint coding and modulation scheme guaranteeing a maximum $3$ dB \ac{PAPR} for \ac{OFDM} symbols by exploiting \acp{CS} \cite{Golay_1961}, which utilizes $2^\numberOfIterations$ subcarriers in a contiguous manner, where $\numberOfIterations$ is an integer. In \cite{Hori_2018}, a multiple-access scheme based on super-orthogonal convolutional codes utilizing \acp{CS} is proposed. By using an interleaved subcarrier mapping, the low-\ac{PAPR} property of \acp{CS} is kept and frequency diversity is achieved. In \cite{Sahin_2018}, a theoretical framework is proposed to synthesize \acp{CS} with null symbols, i.e., non-contiguous \acp{CS} by extending Davis and Jedwab's framework. It can introduce zero elements in \acp{CS} of length $2^\numberOfIterations\cdot\lengthGaGb$  for non-negative integer $\numberOfIterations$ and $\lengthGaGb$ and the number of non-zero clusters is $2^\numberOfIterations$. Thus, it does not address \ac{NR-U} interlaces as  the number of \acp{PRB} in \ac{NR-U} interlace (e.g., $10$ \ac{PRB} for $15$~kHz) cannot be factorized as $2^\numberOfIterations$. To the best our knowledge, a systematic design of low-\ac{PAPR} communication schemes for flexible interlaced transmission is not available in the literature.

In this study, we propose two modulation schemes for uplink control channel based on non-contiguous \acp{CS}. 
We focus on reliable low data rate communication schemes with resources shared by multiple users for a single \ac{OFDM} symbol
The first modulation scheme is for 1 or 2 \ac{UCI} bits and an alternative to \ac{NR} \ac{PUCCH} Format 0 for the interlaced transmission. The second one is a joint coding-and-modulation scheme allows users to transmit moderate payloads in an interlace.  

Our main contributions are as follows:
\begin{itemize}
\item {\bf Theoretical framework:} To derive the proposed methods, we introduce Theorem \ref{th:golayIterative} and Theorem \ref{th:reduced} which allows synthesizing a large number of distinct non-contiguous \acp{CS} by  permuting the multiple seed \acp{GCP} systematically. 
\item {\bf High reliability:} We propose schemes that lead to \ac{OFDM} symbols with a maximum $3$~dB \acp{PAPR} while exploiting the frequency diversity. Approximately $3$~dB and $4$~dB PAPR gains are obtained as compared to the approaches used for \ac{NR} \ac{PUCCH} Format 0 and Format~3 without sacrificing the error rate, respectively. A \ac{GCP} set with low peak cross-correlation is also proposed.
\item {\bf Flexible interlace:} We show that there exist low-\ac{PAPR} modulation schemes for a flexible interlaced transmission. Even if the number of \acp{PRB} or the number of  zeros between the \acp{PRB} in an interlace is changed, the  \ac{PAPR} does not exceed  $3$~dB with the proposed schemes.
\item {\bf Low-complexity design:}  The introduced modulation schemes do not rely on symbol-based optimization. Hence, it is suitable for practical systems. 	

\item {\bf Multi-user support:} While the first scheme supports up to $6$ users, the second scheme enables 21 \ac{UCI} bits for a single user or 11 \ac{UCI} bits for three users on the same interlace of a single \ac{OFDM} symbol.

\end{itemize}
}

The rest of the paper is organized as follows. In Section~\ref{sec:prelim}, we provide the notation and preliminary discussions. In Section~\ref{sec:theorems}, we obtain Theorem \ref{th:golayIterative}  and Theorem~\ref{th:reduced}. In Section~\ref{sec:csPucch}, we derive the proposed schemes. In Section~\ref{sec:numerical}, we  compare the proposed schemes with the other state-of-the-art approaches, numerically. We conclude the paper in Section~\ref{sec:conclusion}.

{\em Notation:} The field of complex numbers,  the set of integers, the set of positive integers, and the set of non-negative integers are denoted by $\complexNumbers$, $\integers$, $\integersPositive$, and $\integersNonnegative$ respectively. 
The symbols $\constanti$, $\constantMinusi$, $\constantOne$, and $\constantMinusOne$ denote $\sqrt{-1}$, $-\sqrt{-1}$, $1$, and $-1$, respectively. 
A sequence of length $\lengthGaGb$ is represented by $\seqGa=(\eleGa[\indexIteration])_{\indexIteration=0}^{\lengthGaGb-1}= (\eleGa[0],\eleGa[1],\dots, \eleGa[\lengthGaGb-1])$. The element-wise complex conjugation and the element-wise absolute operation are denoted by  $(\cdot)^*$ and $|\cdot|$, respectively. 
The sequence $\flipConjugate[\seqGa]$ is the conjugate of the element-wise reversed sequence $\seqGa$.
The operation $\upsampleOp[\upsampleVarA][\seqGa]$ introduces $\upsampleVarA-1$ zero symbols between the elements of $\seqGa$.
The operations $\seqGa+\seqGb$, $\seqGa-\seqGb$, $\seqGa \odot \seqGb$, $\seqGa*\seqGb$, and $\langle\seqGa,\seqGb\rangle$ are the element-wise summation, the element-wise subtraction, the element-wise multiplication, linear convolution, and the inner product of $\seqGa$ and $\seqGb$, respectively.

\section{Preliminaries and Further Notation}
\label{sec:prelim}
We model an interlace as a non-contiguous resource allocation which consists of $\numberOfClusters$ \acp{PRB} where the \acp{PRB} are separated by $\numberOfNulls$ tones in the frequency domain as illustrated in \figurename~\ref{fig:interlace}. We assume that each \ac{PRB}  is composed of $\RBsize$ subcarriers. For example, an interlace in \ac{NR-U} can be expressed as $\RBsize=12$ subcarriers, $\numberOfClusters=10$ \acp{PRB}, and  $\numberOfNulls = 9\times12=108$ subcarriers for 15~kHz subcarrier spacing. The interlace structure  in \ac{NR-U}  varies based on the subcarrier spacing and bandwidth \cite{nr_phy_2020}. 

\subsection{Polynomial Representation of a Sequence}
\label{subsec:poly}
The polynomial representation of the sequence $\seqGa$ can be given by
\begin{align}
\polySeq[\seqGa][\polyVariable] \triangleq \eleGa[\lengthGaGb-1]\polyVariable^{\lengthGaGb-1} + \eleGa[\lengthGaGb-2]\polyVariable^{\lengthGaGb-2}+ \dots + \eleGa[0]~,
\label{eq:polySeq}
\end{align} 
where $\polyVariable\in \complexNumbers$ is a complex number.  One can show that the following identities hold true:
\begin{align}
\polySeq[\seqGa][\polyVariable^\upsampleVarA]&=\polySeq[{\upsampleOp[{\upsampleVarA}][{\seqGa}]}][\polyVariable]\nonumber~,\\
\polySeq[\seqGa][\polyVariable^\upsampleVarA]\polySeq[\seqGb][\polyVariable^\upsampleVarB]&=\polySeq[{\upsampleOp[{\upsampleVarA}][{\seqGa}]}*{\upsampleOp[{\upsampleVarB}][{\seqGb}]}][\polyVariable]\nonumber~,\\
\polySeq[\seqGa][\polyVariable]\polyVariable^\paddingVar&=\polySeq[(\underbrace{0,0,\mydots,0}_{\paddingVar},\seqGa)][\polyVariable]\nonumber~,
\end{align} 
for $\upsampleVarB,\upsampleVarA\in\integersPositive$ and $\paddingVar\in\integersNonnegative$.
If a sequence consists of zero elements between two non-zero elements, it is a non-contiguous sequence. Otherwise, it is a contiguous sequence. 
{\color{\reviewColor}
The support of $\seqGa$ is $\{\varMonomial\in\integers_\lengthGaGb|\eleGa[\varMonomial]\neq 0\}$. The set $\{\eleGa[\varMonomial]|\eleGa[\varMonomial]\neq 0, \eleGa[i]=\eleGa[j]=0, \varMonomial\in\{i+1,i+2,\mydots,j-1\}\}$ is denoted as a non-zero cluster in $\seqGa$.}

The polynomial representation given in \eqref{eq:polySeq} corresponds to an \ac{OFDM} symbol in continuous for $\polyVariable\in\{\constante^{\constanti\frac{2\pi\timeVar}{\symbolDuration}}| 0\le\timeVar <\symbolDuration
\}$,  where  $\symbolDuration$ denotes the \ac{OFDM} symbol duration. The instantaneous envelope power can be expressed as
\begin{align}
|\polySeq[\seqGa][\polyVariable]|^2 = \apacPositive[\seqGa][0] + 
2\sum_{\lagForCorrelation=1}^{\lengthGaGb-1}|\apacPositive[\seqGa][\lagForCorrelation]|\cos\left(\frac{2\pi\timeVar}{\symbolDuration}\lagForCorrelation+\angle \apac[\seqGa][\lagForCorrelation]\right)~,
\label{eq:instantaneousPower}
\end{align}
where  $\apacPositive[\seqGa][\lagForCorrelation]=\sum_{\indexEleOfSeq=0}^{\lengthGaGb-\lagForCorrelation-1} \eleGa[\indexEleOfSeq]^*\eleGa[\indexEleOfSeq+\lagForCorrelation]$ is the \ac{APAC} of the sequence $\seqGa$ \cite{Sahin_2018}.

The minimization of the instantaneous envelope power of an \ac{OFDM} symbol generated through a non-contiguous sequence in the frequency domain is more constrained as compared to the one with a contiguous sequence. For example, consider the interlace model given in  \figurename~\ref{fig:interlace}. 
If the same number of non-zero elements in an interlace is utilized contiguously in the frequency domain, the number of constraints that need to be met for $0$ dB \ac{PAPR} (i.e., $\apacPositive[\seqGa][\lagForCorrelation]=0$  for $\lagForCorrelation\neq0$) is $\numberOfClusters\RBsize-1$ based on \eqref{eq:instantaneousPower}. On the other hand, for $\numberOfNulls\ge\RBsize$, the number of  constraints  increases to  $2\numberOfClusters\RBsize-\RBsize-\numberOfClusters$. As a numerical example,  while the number of constraints for a contiguous resource allocation with $120$ subcarriers is $119$, it increases to $218$ for an interlace in \ac{NR} for $15$~kHz subcarrier spacing, which can be more challenging to satisfy for a low-\ac{PAPR} design.

\subsection{Complementary Sequences}
The sequence pair  $(\seqGa,\seqGb)$ of length  $\lengthGaGb$ is called a \ac{GCP} if
\begin{align}
\apac[\seqGa][\lagForCorrelation]+\apac[\seqGb][\lagForCorrelation] = 0,~~ \text{for}~ \lagForCorrelation~\neq0~,~
\label{eq:GCP}
\end{align}
where the sequences $\seqGa$ and $\seqGb$ are \acp{CS}. By using the definition, one can show that the GCP $(\seqGa,\seqGb)$ satisfies 
\begin{align}
|\polySeq[\seqGa][\polyVariable]|^2+|\polySeq[\seqGb][\polyVariable]|^2\bigg\rvert_{\polyVariable=\constante^{\constanti\frac{2\pi\timeVar}{\symbolDuration}}} =\underbrace{\apac[\seqGa][0]+\apac[\seqGb][0]}_{\text{constant}}~.
\label{eq:timeDomainGCP}
\end{align}
The main property that we inherited from \acp{GCP} in this study is that the instantaneous peak power of the corresponding \ac{OFDM} signal generated through a \ac{CS} $\seqGa$ is bounded, i.e.,
$\max_{\timeVar}|\polySeq[\seqGa][{
	\constante^{\constantj\frac{2\pi\timeVar}{\symbolDuration}}
}]|^2 \le \apac[\seqGa][0]+\apac[\seqGb][0]$. Therefore, based on \eqref{eq:timeDomainGCP}, the \ac{PAPR} of the \ac{OFDM} signal is less than or equal to $10\log_{10}(2)\approx3$~dB if $\apac[\seqGa][0]=\apac[\seqGb][0]$. For the other properties of \acp{GCP}, we refer the reader to the survey given in \cite{parker_2003}.

\subsection{Unimodular Sequences}
\label{sec:unimodular}
Let $\seqGx=(\elex[0],\elex[1],\dots, \elex[\lengthGaGb-1])\in \complexNumbers^\lengthGaGb$ be a sequence of length $\lengthGaGb$. If $|\elex[i]|=c$ for  $i=0,1,\mydots,\lengthGaGb-1$, $\seqGx$ is referred to as a unimodular or constant-amplitude sequence of length $\lengthGaGb$.   Without loss of generality, we assume $c=1$ in this study.
For a unimodular sequence $\seqGx$, one can show that
$
\langle
\seqGx \odot \seqGsShift[i],\seqGx \odot \seqGsShift[j]\rangle = 0 ~ \text{if}~ i \neq j
$, 
where $\seqGsShift[{\shift[]}] = {(\constante^{{\shift[]}\frac{2\pi\constanti}{\lengthGaGb}\times0}, \constante^{{\shift[]}\frac{2\pi\constanti}{\lengthGaGb}\times1}, \dots, \constante^{{\shift[]}\frac{2\pi\constanti}{\lengthGaGb}\times(\lengthGaGb-1)} )}$ for ${\shift[]}=0,1,\mydots,\lengthGaGb-1$ \cite{Benedetto_2009}.
Thus, $\{\seqGx \odot \seqGsShift[{\shift[]}]| {\shift[]}=0,1,\mydots,\lengthGaGb-1\}$ is an orthogonal basis where each sequence can be synthesized in an \ac{OFDM} transmitter with low-complexity operations, i.e.,  shifting the useful duration of \ac{OFDM} signal generated through $\seqGx$ in time cyclically. The unimodular sequences are suitable for \ac{OCC} design, which have been used for increasing the number of users or transmitting more information bits on the same \acp{PRB} in both \ac{NR} \cite{nr_phy_2020} and \ac{LTE} \cite{Sesia2009}.

\subsection{Algebraic Representation of a Sequence}
\label{subsec:algebraic_sequence}
A generalized Boolean function is a function $\funcfForANF$ that maps from $\integers^\numberOfIterations_2=\{(\monomial[1],\monomial[2],\dots, \monomial[\numberOfIterations])| \monomial[\indexFirstOrderMonomial]\in\integers_2\}$ to $\integers_\numberOfPointsForPSK$ as $\funcfForANF:\integers^\numberOfIterations_2\rightarrow\integers_\numberOfPointsForPSK$ where $\numberOfPointsForPSK$ is an integer. We associate a sequence $\seqf$ of length $2^\numberOfIterations$ with the function $\funcfForANF(\monomial[1],\monomial[2],\dots, \monomial[\numberOfIterations])$ by listing its values as $(\monomial[1],\monomial[2],\dots, \monomial[\numberOfIterations])$ ranges over its $2^\numberOfIterations$ values in lexicographic order (i.e., the most significant bit is $\monomial[1]$). In other words, the $(\varMonomial +1)$th element of the sequence $\seqf$ is equal to $\funcfForANF(\monomial[1],\monomial[2],\dots, \monomial[\numberOfIterations])$ where $\varMonomial = \sum_{\indexFirstOrderMonomial=1}^{\numberOfIterations}\monomial[\indexFirstOrderMonomial]2^{\numberOfIterations-\indexFirstOrderMonomial}$. 
Note that different generalized Boolean functions yield  different sequences as each generalized Boolean function can be uniquely expressed as a linear combination of the monomials over $\integers_\numberOfPointsForPSK$ \cite{davis_1999}.
For the sake of simplifying the notation, the sequence $(\monomial[1],\monomial[2],\dots, \monomial[\numberOfIterations])$ and the function $\funcfForANF(\monomial[1],\monomial[2],\dots, \monomial[\numberOfIterations])$ are denoted by $\seqx$ and  $\funcfForANF(\seqx)$, respectively.

{\color{\reviewColor}
\section{Theoretical Framework}
\label{sec:theorems}
In this section, we introduce  Theorem~\ref{th:golayIterative} and Theorem~\ref{th:reduced} to generate \acp{CS} with flexible support and explain the origin of the proposed schemes in Section~\ref{sec:csPucch}. Our first theorem generalizes Golay's concatenation and interleaving methods \cite{Golay_1961} as follows:
}
\begin{theorem}
	\label{th:golayIterative}
	Let $(\seqGa,\seqGb)$ and $(\seqGc,\seqGd)$ be \acp{GCP} of length $\lengthGaGb$ and $\lengthGcGd$, respectively,  $\angleGolay[1],\angleGolay[2]\in \{u:u\in\complexNumbers, |u|=1\}$,  and $\upsampleVarA, \upsampleVarB, \paddingVar\in\integers$. Then, the sequences $\seqGf$ and $\seqGg$ where their polynomial representations are given by
	\begin{align}
	\polySeq[{\seqGf}][\polyVariable] =& \angleGolay[1]\polySeq[{\seqGa}][\polyVariable^\upsampleVarA]\polySeq[{\seqGc}][\polyVariable^\upsampleVarB]   
	+ \angleGolay[2]\polySeq[{\seqGb}][\polyVariable^\upsampleVarA]\polySeq[{\seqGd}][\polyVariable^\upsampleVarB]  \polyVariable^{\paddingVar} ~, 
	\label{eq:gcpGf}
	\\
	\polySeq[{\seqGg}][\polyVariable] =& \angleGolay[1]\polySeq[{\seqGa}][\polyVariable^\upsampleVarA]\polySeq[{\flipConjugate[\seqGd]}][\polyVariable^\upsampleVarB] 
	- \angleGolay[2]\polySeq[{\seqGb}][\polyVariable^\upsampleVarA]\polySeq[{\flipConjugate[\seqGc]}][\polyVariable^\upsampleVarB]  \polyVariable^{\paddingVar} ~, 
	\label{eq:gcpGg}
	\end{align}
	construct a \ac{GCP}.
\end{theorem}

\begin{proof}
	Since the sequence pairs $(\seqGa,\seqGb)$ and $(\seqGc,\seqGd)$ are \acp{GCP}, by the definition, $|\polySeq[{\seqGa}][\polyVariable]|^2 + |\polySeq[{\seqGb}][\polyVariable]|^2=C_1$ and $|\polySeq[{\seqGc}][\polyVariable]|^2 + |\polySeq[{\seqGd}][\polyVariable]|^2=C_2$, where $C_1$ and $C_2$ are some constants.
	Thus, we  need to show that $|\polySeq[{\seqGf}][\polyVariable]|^2+|\polySeq[{\seqGg}][\polyVariable]|^2$ is also a constant. 
	Since $\polySeq[{\flipConjugate[\seqGa]}][\polyVariable^\upsampleVarA]=\polySeq[{\seqGa^*}][\polyVariable^{-\upsampleVarA}]\polyVariable^{\upsampleVarA\lengthGaGb-\upsampleVarA}$,  $|\polySeq[{\seqGf}][\polyVariable]|^2+|\polySeq[{\seqGg}][\polyVariable]|^2$ can be calculated as
	\begin{align}
	|\polySeq[{\seqGf}][\polyVariable]&|^2+|\polySeq[{\seqGg}][\polyVariable]|^2 \nonumber
	\\=
	&~~~
	(\angleGolay[1]\polySeq[{\seqGa}][\polyVariable^\upsampleVarA]\polySeq[{\seqGc}][\polyVariable^\upsampleVarB] 
	+ \angleGolay[2]\polySeq[{\seqGb}][\polyVariable^\upsampleVarA]\polySeq[{\seqGd}][\polyVariable^\upsampleVarB]  \polyVariable^{\paddingVar})
	\nonumber\\
	&\times 
	(\angleGolay[1]^*\polySeq[{\seqGa^*}][\polyVariable^{-\upsampleVarA}]\polySeq[{\seqGc^*}][\polyVariable^{-\upsampleVarB}]   
	+ \angleGolay[2]^*\polySeq[{\seqGb^*}][\polyVariable^{-\upsampleVarA}]\polySeq[{\seqGd^*}][\polyVariable^{-\upsampleVarB}]  \polyVariable^{-\paddingVar})\nonumber\\
	&+
	(\angleGolay[1]\polySeq[{\seqGa}][\polyVariable^\upsampleVarA]\polySeq[{\flipConjugate[\seqGd]}][\polyVariable^\upsampleVarB] 
	- \angleGolay[2]\polySeq[{\seqGb}][\polyVariable^\upsampleVarA]\polySeq[{\flipConjugate[\seqGc]}][\polyVariable^\upsampleVarB]  \polyVariable^{\paddingVar} )
	\nonumber\\&\times
	(\angleGolay[1]^*\polySeq[{\seqGa^*}][\polyVariable^{-\upsampleVarA}]\polySeq[{\flipConjugate[\seqGd^*]}][\polyVariable^{-\upsampleVarB}] 
	- \angleGolay[2]^*\polySeq[{\seqGb^*}][\polyVariable^\upsampleVarA]\polySeq[{\flipConjugate[\seqGc]^*}][\polyVariable^{-\upsampleVarB}]  \polyVariable^{-\paddingVar} )
	\nonumber \\
	\stackrel{(a)}{=}&~~~ 
	\polySeq[{\seqGa}][\polyVariable^\upsampleVarA]\polySeq[{\seqGa^*}][\polyVariable^{-\upsampleVarA}]
	\polySeq[{\seqGc}][\polyVariable^\upsampleVarB]\polySeq[{\seqGc^*}][\polyVariable^{-\upsampleVarB}]
	\nonumber \\
	&+\polySeq[{\seqGa}][\polyVariable^\upsampleVarA]\polySeq[{\seqGa^*}][\polyVariable^{-\upsampleVarA}]
	\polySeq[\flipConjugate[{\seqGd}]][\polyVariable^\upsampleVarB]\polySeq[\flipConjugate[{\seqGd}]^*][\polyVariable^{-\upsampleVarB}]
	\nonumber \\
	&+\polySeq[{\seqGb}][\polyVariable^\upsampleVarA]\polySeq[{\seqGb^*}][\polyVariable^{-\upsampleVarA}]
	\polySeq[\flipConjugate[{\seqGc}]][\polyVariable^\upsampleVarB]\polySeq[\flipConjugate[{\seqGc}]^*][\polyVariable^{-\upsampleVarB}]
	\nonumber \\
	&+\polySeq[{\seqGb}][\polyVariable^\upsampleVarA]\polySeq[{\seqGb^*}][\polyVariable^{-\upsampleVarA}]
	\polySeq[{\seqGd}][\polyVariable^\upsampleVarB]\polySeq[{\seqGd^*}][\polyVariable^{-\upsampleVarB}]
	\nonumber\\
	\stackrel{(b)}{=}&~~~ 
	(\polySeq[{\seqGa}][\polyVariable^\upsampleVarA]\polySeq[{\seqGa^*}][\polyVariable^{-\upsampleVarA}]
	+\polySeq[{\seqGb}][\polyVariable^\upsampleVarA]\polySeq[{\seqGb^*}][\polyVariable^{-\upsampleVarA}])
	\nonumber \\
	&\times (\polySeq[{\seqGc}][\polyVariable^\upsampleVarB]\polySeq[{\seqGc^*}][\polyVariable^{-\upsampleVarB}]+
	\polySeq[{\seqGd}][\polyVariable^\upsampleVarB]\polySeq[{\seqGd^*}][\polyVariable^{-\upsampleVarB}])
	= C_1C_2~,\nonumber
	\end{align}
	where (a) follows from $\polySeq[\flipConjugate[{\seqGc}]^*][\polyVariable^{-\upsampleVarB}]\polySeq[{\flipConjugate[\seqGd]}][\polyVariable^{\upsampleVarB}]=\polySeq[{\seqGc}][\polyVariable^{\upsampleVarB}]\polySeq[{\seqGd^*}][\polyVariable^{-\upsampleVarB}]$ and (b) is because  $\polySeq[\flipConjugate[{\seqGc}]][\polyVariable^\upsampleVarB]\polySeq[\flipConjugate[{\seqGc}]^*][\polyVariable^{-\upsampleVarB}]=\polySeq[{\seqGc^*}][\polyVariable^{-\upsampleVarB}]\polySeq[{\seqGc}][\polyVariable^\upsampleVarB]$ and $\polySeq[\flipConjugate[{\seqGd}]][\polyVariable^\upsampleVarB]\polySeq[\flipConjugate[{\seqGd}]^*][\polyVariable^{-\upsampleVarB}] = \polySeq[{\seqGd^*}][\polyVariable^{-\upsampleVarB}]\polySeq[{\seqGd}][\polyVariable^\upsampleVarB]
	$.
\end{proof}
Note that the special cases of Theorem~\ref{th:golayIterative} are available in earlier work. For example, binary contiguous \acp{CS}  or multi-level contiguous \acp{CS} are discussed when $\lengthGcGd=1$ \cite{Turyn_1974,parker_2003,Garcia_2010_ml}. {\color{\reviewColor}However, Theorem~\ref{th:golayIterative} also plays a central role for generating non-contiguous \acp{CS} which is not widely discussed in the literature. For example, based on the identities given in Section \ref{subsec:poly}, the factor $\polyVariable^{\paddingVar}$ increases the degree of the polynomial $\angleGolay[2]\polySeq[{\seqGb}][\polyVariable^\upsampleVarA]\polySeq[{\seqGd}][\polyVariable^\upsampleVarB] $ by $\paddingVar$, which yields two clusters in the sequence $\seqGf$ where the number of zeroes between them can be chosen arbitrarily. 
 This is one of key observations that we exploit in this study to limit the \ac{PAPR} of \ac{OFDM} symbol for flexible non-contiguous allocations. Similarly, $\upsampleVarA>\lengthGcGd$ or $\upsampleVarB>\lengthGaGb$ can generate non-contiguous \acp{CS} due to the convolutions of up-sampled sequences.}

To support more information bits, it is important to generate distinct \acp{CS}. However, Theorem~\ref{th:golayIterative} does not show how to generate distinct \acp{CS}, systematically. To address this issue, we introduce a new theorem  as follows:
\begin{theorem}
	\label{th:reduced}
	Let $\seqPermutationCompShift=(\permutationMono[\indexIteration])_{\indexIteration=1}^{\numberOfIterations}$ and $\seqPermutationOrder=(\permutationOrderEle[\indexIteration])_{\indexIteration=1}^{\numberOfIterations}$ be two sequences defined by permutations of $\{1,2,\dots,\numberOfIterations\}$. For any  GCP $(\seqGcRecursion[\indexIteration],\seqGdRecursion[\indexIteration])$ of length $\lengthGcGdIterative[\indexIteration]\in\integersPositive$, $\varUpsample\in\integersNonnegative$, $\separationGolay[\indexIteration]\in \integersNonnegative$, and $\angleexp[\indexIteration],\arbitraryPhaseA,\arbitraryPhaseB \in[0,\numberOfPointsForPSK)$ for $\indexIteration=1,2,\mydots,\numberOfIterations$, let 
	\begin{align}
	\funcfForCommonPhase(\seqx)
	=& {\frac{\numberOfPointsForPSK}{2}\sum_{\indexIteration=1}^{\numberOfIterations-1}\monomial[{\permutationMono[{\indexIteration}]}]\monomial[{\permutationMono[{\indexIteration+1}]}]}+\sum_{\indexIteration=1}^\numberOfIterations \angleexp[\indexIteration]\monomial[{\permutationMono[{\indexIteration}]}]\label{eq:imagPartReduced}~,
	\\
	\funcForCommonOrder(\seqx,\polyVariable)=
	&\prod_{\indexIteration=1}^{\numberOfIterations-1}	\polySeq[{\seqGcRecursion[{\permutationOrderEle[\indexIteration]}]}][\polyVariable]((1- \monomial[{\permutationMono[{\indexIteration}]}])(1 - \monomial[{\permutationMono[{\indexIteration+1}]}]))_2 \nonumber\\&~~~~+ 
	\polySeq[{\seqGdRecursion[{\permutationOrderEle[\indexIteration]}]}][\polyVariable]({\monomial[{\permutationMono[{\indexIteration}]}]}(1 - \monomial[{\permutationMono[{\indexIteration+1}]}]))_2 \nonumber\\&~~~~+
	\polySeq[{\seqGdTildeRecursion[{\permutationOrderEle[\indexIteration]}]}][\polyVariable]((1 - \monomial[{\permutationMono[{\indexIteration}]}])\monomial[{\permutationMono[{\indexIteration+1}]}])_2 \nonumber\\&~~~~+
	{ \polySeq[{\seqGcTildeRecursion[{\permutationOrderEle[\indexIteration]}]}][\polyVariable]  (\monomial[{\permutationMono[{\indexIteration}]}]}\monomial[{\permutationMono[{\indexIteration+1}]}])_2
~,
	\end{align}
	and
	\begin{align}
	\funcfForCommonPhaseA(\seqx)&= \funcfForCommonPhase(\seqx)+\arbitraryPhaseA\nonumber~,\\
	\funcfForCommonPhaseB(\seqx)&= \funcfForCommonPhase(\seqx)+\arbitraryPhaseB \nonumber~, \\	
	\funcfForCommonOrder(\seqx,\polyVariable) &=  	\funcForCommonOrder(\seqx,\polyVariable) (\polySeq[{\seqGcRecursion[{\permutationOrderEle[\numberOfIterations]}]}][\polyVariable](1 - \monomial[{\permutationMono[{\numberOfIterations}]}])_2 + 
	\polySeq[{\seqGdRecursion[{\permutationOrderEle[\numberOfIterations]}]}][\polyVariable]\monomial[{\permutationMono[{\numberOfIterations}]}])~,\nonumber\\
	\funcgForCommonOrder(\seqx,\polyVariable) &=  	\funcForCommonOrder(\seqx,\polyVariable) (\polySeq[{\seqGdTildeRecursion[{\permutationOrderEle[\numberOfIterations]}]}][\polyVariable](1 - \monomial[{\permutationMono[{\numberOfIterations}]}])_2 + 
	\polySeq[{\seqGcTildeRecursion[{\permutationOrderEle[\numberOfIterations]}]}][\polyVariable]\monomial[{\permutationMono[{\numberOfIterations}]}])~,\nonumber
	\\
	\funcfForCommonShift(\seqx)&=\sum_{\indexIteration=1}^\numberOfIterations\separationGolay[\indexIteration]\monomial[{\permutationMono[{\indexIteration}]}]
	\nonumber~.
	\end{align}
	Then, the sequences $\seqGf$ and $\seqGg$ where their polynomial representations are given by
	\begin{align}
	\polySeq[{\seqGf}][\polyVariable] &= 
	\sum_{\varMonomial=0}^{2^\numberOfIterations-1} 
	\funcfForCommonOrder(\seqx,\polyVariable)\times 
	\exponentialBase^{ \constantj \funcfForFinalPhase(\seqx)}\times
	\polyVariable^{\funcfForCommonShift(\seqx) + \varMonomial\varUpsample}\label{eq:encodedFOFDM}~,
	\\
	\polySeq[{\seqGg}][\polyVariable] &= 
	\sum_{\varMonomial=0}^{2^\numberOfIterations-1} 
	\funcgForCommonOrder(\seqx,\polyVariable) \times
	\exponentialBase^{ \constantj \funcfForCommonPhaseB(\seqx)}\times
	\polyVariable^{\funcfForCommonShift(\seqx) +\varMonomial\varUpsample}\label{eq:encodedGOFDM}
	\end{align}
	construct a \ac{GCP}, where $\exponentialBase=\constante^{\frac{2\pi}{\numberOfPointsForPSK}}$.
\end{theorem}
The proof of Theorem~\ref{th:reduced} is given in Appendix~\ref{app:GCPrecursion}. 

{\color{\reviewColor}
Theorem~\ref{th:reduced} contains the results in \cite{davis_1999, Sahin_2018}, and \cite{paterson_2000}:
\begin{itemize}
\item The functions  $\funcfForCommonPhaseA(\seqx)$ and $\funcfForCommonPhaseB(\seqx)$ in Theorem~\ref{th:reduced} are identical to the ones in \cite{davis_1999} for $\angleexp[\indexIteration],\arbitraryPhaseA,\arbitraryPhaseB \in \integers_{2^h}$ and  \cite{paterson_2000} for $\angleexp[\indexIteration],\arbitraryPhaseA,\arbitraryPhaseB \in\integers_\numberOfPointsForPSK$, where $h\ge1$ is an integer. It was shown that  $\funcfForCommonPhaseA(\seqx)$ and $\funcfForCommonPhaseB(\seqx)$ yield the codewords in the cosets of the first-order \ac{RM} code within the second-order \ac{RM} code where the Hamming distance between two codewords is at least $2^{\numberOfIterations-2}$.
\item 
The function $\funcfForCommonShift(\seqx)$ in Theorem~\ref{th:reduced} appears in \cite{Sahin_2018} to generate non-contiguous \acp{CS} by increasing the degrees of the polynomials in the summands as in \eqref{eq:encodedFOFDM} and \eqref{eq:encodedGOFDM}. The number of non-zero clusters in the \acp{CS} can reach up to $2^\numberOfIterations$ with $\funcfForCommonShift(\seqx)$.
\end{itemize}
On the other hand, Theorem~\ref{th:reduced} introduces a new term which can be utilized for obtaining the number of non-zero clusters different than $2^\numberOfIterations$ through multiple seed \acp{GCP}:
\begin{itemize}
	\item  In Theorem~\ref{th:reduced}, $\funcfForCommonOrder(\seqx,\polyVariable)$ and $\funcgForCommonOrder(\seqx,\polyVariable)$ are the products of  $\numberOfIterations$ polynomials determined systematically based on the permutations of $\seqPermutationOrder$ and $\seqPermutationCompShift$ for $\indexIteration=1,2,\mydots,\numberOfIterations$	whereas they are generated through a single \ac{GCP} of length $\lengthGaGb$ and are not functions of $\indexIteration$ in \cite{Sahin_2018}. While $\seqPermutationOrder$ determines the sequences, $\seqPermutationCompShift$ defines the order of the sequences in $(\funcfForCommonOrder(\seqx,\polyVariable))_{\varMonomial=0}^{2^3-1}$. 
\end{itemize}
For example, let $\numberOfIterations=3$, $\seqPermutationCompShift=(3,2,1)$, and $\seqPermutationOrder=(1,2,3)$. The values of function $\funcfForCommonOrder(\seqx,\polyVariable)$ can be enumerated as
\begin{align}
(\funcfForCommonOrder(\seqx,\polyVariable))_{\varMonomial=0}^{2^3-1} =& (
\polySeq[{\seqGcRecursion[3]}][\polyVariable]\polySeq[{\seqGcRecursion[2]}][\polyVariable]\polySeq[{\seqGcRecursion[1]}][\polyVariable],
\polySeq[{\seqGcRecursion[3]}][\polyVariable]\polySeq[{\seqGcRecursion[2]}][\polyVariable]\polySeq[{\seqGdRecursion[1]}][\polyVariable],
\nonumber\\&
\polySeq[{\seqGcRecursion[3]}][\polyVariable]\polySeq[{\seqGdRecursion[2]}][\polyVariable]\polySeq[{\seqGdTildeRecursion[1]}][\polyVariable],
\polySeq[{\seqGcRecursion[3]}][\polyVariable]\polySeq[{\seqGdRecursion[2]}][\polyVariable]\polySeq[{\seqGcTildeRecursion[1]}][\polyVariable],
\nonumber\\&
\polySeq[{\seqGdRecursion[3]}][\polyVariable]\polySeq[{\seqGdTildeRecursion[2]}][\polyVariable]\polySeq[{\seqGcRecursion[1]}][\polyVariable],
\polySeq[{\seqGdRecursion[3]}][\polyVariable]\polySeq[{\seqGdTildeRecursion[2]}][\polyVariable]\polySeq[{\seqGdRecursion[1]}][\polyVariable],
\nonumber\\&
\polySeq[{\seqGdRecursion[3]}][\polyVariable]\polySeq[{\seqGcTildeRecursion[2]}][\polyVariable]\polySeq[{\seqGdTildeRecursion[1]}][\polyVariable],
\polySeq[{\seqGdRecursion[3]}][\polyVariable]\polySeq[{\seqGcTildeRecursion[2]}][\polyVariable]\polySeq[{\seqGcTildeRecursion[1]}][\polyVariable]
).
\label{eq:examplefo}
\end{align}
If  $\seqPermutationOrder$ is changed to $(3,2,1)$, the enumeration leads to 
\begin{align}
(\funcfForCommonOrder(\seqx,\polyVariable))_{\varMonomial=0}^{2^3-1} =& (
\polySeq[{\seqGcRecursion[1]}][\polyVariable]\polySeq[{\seqGcRecursion[2]}][\polyVariable]\polySeq[{\seqGcRecursion[3]}][\polyVariable],
\polySeq[{\seqGcRecursion[1]}][\polyVariable]\polySeq[{\seqGcRecursion[2]}][\polyVariable]\polySeq[{\seqGdRecursion[3]}][\polyVariable],
\nonumber\\&
\polySeq[{\seqGcRecursion[1]}][\polyVariable]\polySeq[{\seqGdRecursion[2]}][\polyVariable]\polySeq[{\seqGdTildeRecursion[3]}][\polyVariable],
\polySeq[{\seqGcRecursion[1]}][\polyVariable]\polySeq[{\seqGdRecursion[2]}][\polyVariable]\polySeq[{\seqGcTildeRecursion[3]}][\polyVariable],
\nonumber\\&
\polySeq[{\seqGdRecursion[1]}][\polyVariable]\polySeq[{\seqGdTildeRecursion[2]}][\polyVariable]\polySeq[{\seqGcRecursion[3]}][\polyVariable],
\polySeq[{\seqGdRecursion[1]}][\polyVariable]\polySeq[{\seqGdTildeRecursion[2]}][\polyVariable]\polySeq[{\seqGdRecursion[3]}][\polyVariable],
\nonumber\\&
\polySeq[{\seqGdRecursion[1]}][\polyVariable]\polySeq[{\seqGcTildeRecursion[2]}][\polyVariable]\polySeq[{\seqGdTildeRecursion[3]}][\polyVariable],
\polySeq[{\seqGdRecursion[1]}][\polyVariable]\polySeq[{\seqGcTildeRecursion[2]}][\polyVariable]\polySeq[{\seqGcTildeRecursion[3]}][\polyVariable]
)~,
\label{eq:examplefo2}
\end{align}
where the different sequences are chosen, but their distribution in $(\funcfForCommonOrder(\seqx,\polyVariable))_{\varMonomial=0}^{2^3-1}$ remains the same as compared to the one in \eqref{eq:examplefo}. Since $\funcfForCommonOrder(\seqx,\polyVariable)$ is the product of the $\numberOfIterations$ polynomials generated through the seed sequences, it is also equal to the polynomial representation of the convolutions of the corresponding sequences  based on the identities given in Section \ref{subsec:poly}. The length of the  $\varMonomial$th composite sequence after the convolutions can be calculated as $\lengthPillar = (\sum_{\indexIteration=1}^{\numberOfIterations}\lengthGcGdIterative[\indexIteration])-\numberOfIterations+1$. 

For $\angleexp[\indexIteration],\arbitraryPhaseA,\arbitraryPhaseB \in\integers_\numberOfPointsForPSK$, $\funcfForCommonOrder(\seqx,\polyVariable)$  is multiplied with $\varMonomial$th \ac{PSK} symbol determined by $\funcfForCommonPhaseA(\seqx)$, i.e., a \ac{RM} code over $\integers_\numberOfPointsForPSK$ in \eqref{eq:encodedFOFDM}.
In addition,  the degree of the  polynomial composed by $\funcfForCommonOrder(\seqx,\polyVariable)$ is also increased by $\polyVariable^{\funcfForCommonShift(\seqx) +\varMonomial\varUpsample}$. Therefore,  the overall operation can be represented as a shift of the $\varMonomial$th phase-rotated composite sequence in the $\polyVariable$-domain,  where the amount of shift is $\funcfForCommonShift(\seqx) +\varMonomial\varUpsample$.  Hence, the final sequence $\seqGf$ is then obtained by summing the $2^\numberOfIterations$ shifted and phase-rotated composite sequences. The length of the final sequence can be calculated as $\lengthPillar+\varUpsample(2^\numberOfIterations-1)+\sum_{\indexIteration=1}^{\numberOfIterations}\separationGolay[\indexIteration]$. 
\begin{figure}[t]
	\centering
	{\includegraphics[width =3.0in]{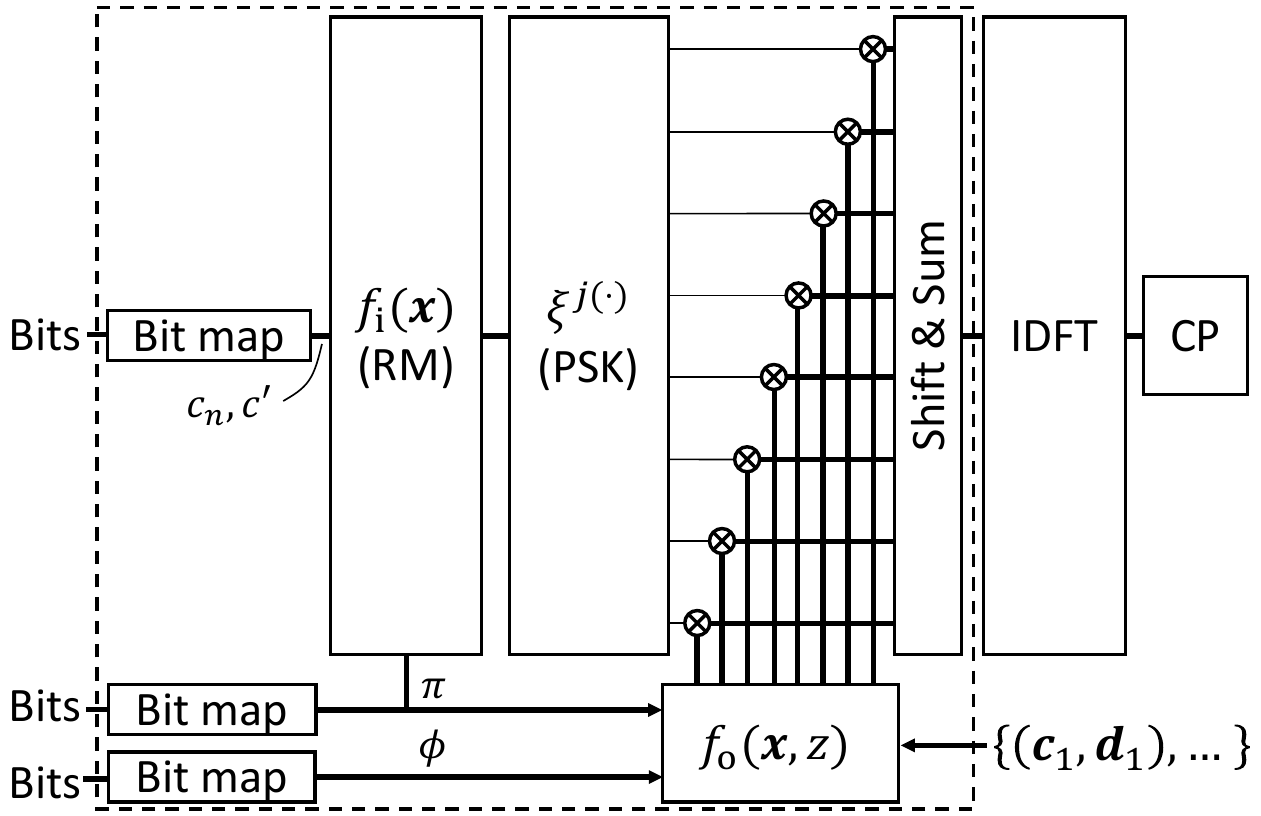}
	}
	\caption{{\color{\reviewColor}Interpretation of \eqref{eq:encodedFOFDM} as an \ac{OFDM} transmitter for $\numberOfIterations=3$.}}
	\label{fig:txCoherent}
\end{figure}
In \figurename~\ref{fig:txCoherent}, we illustrate these steps as an \ac{OFDM} transmitter with \ac{CP} where we configure the parameters  $\seqPermutationCompShift$, $\seqPermutationOrder$, $\angleexp[\indexIteration]$, and $\arbitraryPhaseA$ based on information bits. 

The number of non-zero elements and the number of non-zero clusters in the \acp{CS} are limited to $2^\numberOfIterations\cdot\lengthGaGb$ and $2^\numberOfIterations$ in \cite{Sahin_2018}, respectively. However, they can be chosen  more flexibly as a result of polynomial multiplications in $\funcfForCommonOrder(\seqx,\polyVariable)$. To see this, let 
$\seqGcRecursion[\indexIteration]=\upsampleOp[{\upsampleVarA_{\indexIteration}}][{\seqGcContRecursion[\indexIteration]}]$ and 
$\seqGdRecursion[\indexIteration]=\upsampleOp[{\upsampleVarA_{\indexIteration}}][{\seqGdContRecursion[\indexIteration]}]$ where $(\seqGcContRecursion[\indexIteration],\seqGdContRecursion[\indexIteration])$ is a contiguous \ac{GCP} of length $\lengthGcGdContIterative[\indexIteration]$. 
Therefore, $\funcfForCommonOrder(\seqx,\polyVariable)$ is equal to the product of the polynomial representation of upsampled \acp{CS}, e.g.,  $\polySeq[{\upsampleOp[{\upsampleVarA_1}][{\seqGcRecursion[1]}]}][\polyVariable]
\polySeq[{\upsampleOp[{\upsampleVarA_2}][{\seqGcRecursion[2]}]}][\polyVariable]
\polySeq[{\upsampleOp[{\upsampleVarA_3}][{\seqGcRecursion[3]}]}][\polyVariable]$. Let
\begin{align}
	\upsampleVarA_{{\permutationUpsample[{\indexIteration+1}]}} &\ge 
		\lengthGcGdContIterative[{\permutationUpsample[{\indexIteration}]}]\upsampleVarA_{{\permutationUpsample[{\indexIteration}]}}
		\label{eq:condition}
\end{align}
for $\indexIteration=1,2\mydots,\numberOfIterations-1$ and $\upsampleVarA_{\permutationUpsample[1]}\ge1$, where 
$\seqPermutationUpsample=(\permutationUpsample[\indexIteration])_{\indexIteration=1}^{\numberOfIterations}$ is a sequence defined by a permutation of $\{1,2,\dots,\numberOfIterations\}$. The seed \acp{CS} then spread each other (as Kronecker products) and the composite sequences can be non-contiguous. When \eqref{eq:condition} is tight, the length of the composite sequence is $\lengthPillar = \prod_{\indexIteration=1}^{\numberOfIterations}\lengthGcGdContIterative[\indexIteration]$.

Let $\numberSequencesDifferentThanLEngthOne$ and $\numberSequencesNotColinear$ be the number of \acp{CS} of length larger than $1$ and the number of seed \acp{CS} that are not co-linear with each other, respectively. 
Assuming that the supports of the composite sequences do not overlap in the $\polyVariable$-domain, the number of distinct \acp{CS} generated with Theorem~\ref{th:reduced} is  $\numberSequencesDifferentThanLEngthOne!\frac{(\numberOfIterations!)^2}{(\numberOfIterations-\numberSequencesNotColinear+1)!}\numberOfPointsForPSK^{\numberOfIterations+1}$  under the condition \eqref{eq:condition} and the presence of at least one seed \ac{CS} of length larger than $1$. This result is substantially different from the numbers provided in \cite{davis_1999} and \cite{Sahin_2018}.  It can be obtained from 
  $\numberSequencesDifferentThanLEngthOne!$  permutations of $\seqPermutationUpsample$ (i.e., results in different spreading configurations), 
 ${\numberOfIterations!}/{(\numberOfIterations-\numberSequencesNotColinear+1)!}$ permutations of $\seqPermutationOrder$ (i.e., gives different composite sequences),
 $\numberOfIterations!$ permutations of $\seqPermutationCompShift$ (i.e., alters the order of the composite sequence in $(\funcfForCommonOrder(\seqx,\polyVariable))_{\varMonomial=0}^{2^\numberOfIterations-1}$) in the presence of at least one seed \ac{GCP}, and $\numberOfPointsForPSK$ different values for $\angleexp[\indexIteration=1,2,\mydots,\numberOfIterations]$ and $\arbitraryPhaseA$. 
  Note that the minimum Euclidean distance for the sequences $\{(\exponentialBase^{ \constantj \funcfForFinalPhase(\seqx)})_{\varMonomial=0}^{2^\numberOfIterations-1}\}$ is equal to $2^{\numberOfIterations/2}\sin{(\pi/\numberOfPointsForPSK)}$ as the codewords are in the second-order \ac{RM} code. The Boolean functions that determine seed sequences and their positions are also from the second-order \ac{RM} code. However, the minimum Euclidean distance between \acp{CS} is a function of the elements of seed \acp{GCP} in general.
}

\section{CS-Based UL Control Channel}
\label{sec:csPucch}
In this section, we derive two modulation schemes for \ac{UCI} transmission by relying Theorem~\ref{th:golayIterative} and Theorem~\ref{th:reduced} discussed in Section~\ref{sec:theorems}. 
We generate non-contiguous \acp{CS} compatible with an interlace through the parameters $\paddingVar$, $\upsampleVarA$, and $\upsampleVarB$ in  Theorem~\ref{th:golayIterative} for the first scheme supporting up to 2 \ac{UCI} bits. We exploit the permutations of  $\seqPermutationCompShift$ and $\seqPermutationOrder$, $\angleexp[\indexIteration=1,2,\mydots,\numberOfIterations]$, and $\arbitraryPhaseA$  in Theorem~\ref{th:reduced} to  transmit more than 2 \ac{UCI} bits. We also discuss the user multiplexing with these schemes.

\subsection{Transmission for up to 2 \ac{UCI} bits}
\label{subsec:twobitsAndusermultiplexing}
\begin{figure}[t]
	\centering
	{\includegraphics[width =2.5in]{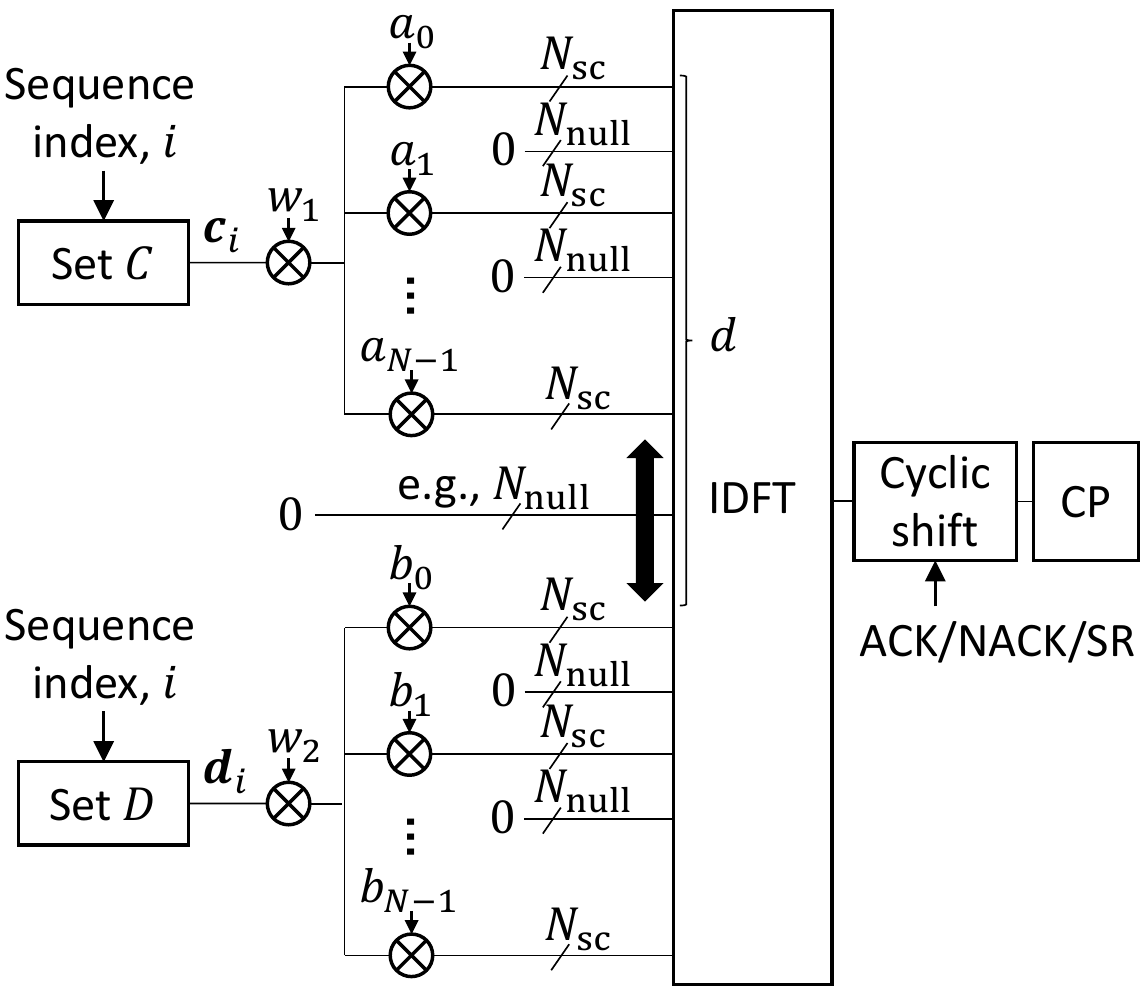}
	}
	\caption{Transmitter for up to 2 \ac{UCI} bits.}
	\label{fig:tx}
\end{figure}

Consider the interlace model in \figurename~\ref{fig:interlace}. Let  $\setC\triangleq\{\seqGc_1,\seqGc_2,\mydots,\seqGc_\numberOfSequences\}$ and $\setD\triangleq\{\seqGd_1,\seqGd_2,\mydots,\seqGd_\numberOfSequences\}$ where  $(\seqGc_\indexSequence,\seqGd_\indexSequence)$ is a \ac{GCP} of length $\RBsize$ for $\indexSequence=1,2,\mydots,\numberOfSequences$.
We first  choose a \ac{GCP} $(\seqGa,\seqGb)$ of length $\numberOfClusters/2$  where the elements of $\seqGa$, $\seqGb$, $\seqGc_\indexSequence$ and $\seqGd_\indexSequence$ are in the set $\setQ\triangleq\{\constantOne,\constantMinusOne,\constanti,\constantMinusi\}$.  We then generate the interlace through \eqref{eq:gcpGf} in Theorem \ref{th:golayIterative} by setting $\seqGc=\seqGc_\indexSequence$, $\seqGd=\seqGd_\indexSequence$, $\angleGolay[1]=\angleGolay[2]=\constante^{\frac{\constanti\pi}{4}}$, $\upsampleVarA = \RBsize+\numberOfNulls$, $\upsampleVarB =1$, and $\paddingVar = (\RBsize+\numberOfNulls)\times\numberOfClusters/2$. 
With this choice, $\seqGa$ and $\seqGb$ act as sequences spreading $\seqGc_\indexSequence$ and $\seqGd_\indexSequence$. This can be seen from the identities given in Section \ref{subsec:poly} as
\begin{align}
\polySeq[{\seqGa}][\polyVariable^{\RBsize+\numberOfNulls}]\polySeq[{\seqGc}][\polyVariable]=\polySeq[{\upsampleOp[{\RBsize+\numberOfNulls}][\seqGa]*\seqGc}][\polyVariable],
\end{align}
and
\begin{align}
\polySeq[{\seqGb}][\polyVariable^{\RBsize+\numberOfNulls}]\polySeq[{\seqGd}][\polyVariable]=\polySeq[{\upsampleOp[{\RBsize+\numberOfNulls}][\seqGb]*\seqGd}][\polyVariable].
\end{align}
In other words, the \acp{PRB} are constructed with the phased-rotated versions of $\seqGc_\indexSequence$ and $\seqGd_\indexSequence$ and the phase rotations are determined by the elements of $\seqGa$ and $\seqGb$ as shown in \figurename~\ref{fig:tx}. Based on the second part of \eqref{eq:gcpGf} in Theorem \ref{th:golayIterative}, $\paddingVar$ can be chosen as $(\RBsize+\numberOfNulls)\times\numberOfClusters/2$ to pad the sequence ${\upsampleOp[{\RBsize+\numberOfNulls}][\seqGb]*\seqGd_\indexSequence}$ with $(\RBsize+\numberOfNulls)\times\numberOfClusters/2$ zeros. Hence, while the first half of the interlace is a function of $\seqGa$ and $\seqGc_\indexSequence$, the second part is generated through $\seqGb$ and $\seqGd_\indexSequence$ as illustrated in  \figurename~\ref{fig:tx}. For instance, by considering the interlace parameters in \ac{NR-U} for $15$~kHz subcarrier spacing, the interlace can be constructed when $\upsampleVarA =120$, $\upsampleVarB =1$, and $\paddingVar = 600$ and the sequences $\seqGa$, $\seqGb$, $\seqGc_\indexSequence$, $\seqGd_\indexSequence$ can be arbitrarily chosen such as $\seqGa = (\constantOne,\constantOne,\constantOne,\constantMinusi,\constanti)$, $\seqGb = (\constantOne,\constanti,\constantMinusOne,\constantOne,\constantMinusi)$, $\seqGc_\indexSequence =  (\constantOne,\constantOne,\constantOne,\constantOne,\constantMinusOne,\constantMinusOne,\constantMinusOne,\constantOne,\constanti,\constantMinusi,\constantMinusOne,\constantOne)$, and $\seqGd_\indexSequence = (\constantOne,\constantOne,\constanti,\constanti,\constantOne,\constantOne,\constantMinusOne,\constantOne,\constantOne,\constantMinusOne,\constantOne,\constantMinusOne)$ \cite{holzmann_1991}.

{\color{\reviewColor}
One of the main benefits of the proposed approach is that $\setC$ and $\setD$ can be reused by  changing $\numberOfNulls$, $\paddingVar$, or defining  a new single \ac{GCP} $(\seqGa,\seqGb)$, which remarkably simplifies the design involving different interlace configurations.
For example, the \ac{GCP} $(\seqGa,\seqGb)$ can be configured based on subcarrier spacing to maintain signal bandwidth. In another example, a larger $\paddingVar$  can generate a gap in the frequency domain, which can be utilized for contiguous random access signals. In both examples, 
 \ac{PAPR} is still maintained to be less than or equal to $3$~dB without modifying the sequences in $\setC$ and $\setD$. Note that we design $\setC$ and $\setD$ in Section~\ref{subsec:cciset} to address the \ac{CCI} (i.e., inter-cell interference) minimization problem independently of the spreading sequences $\seqGa$ and $\seqGb$ without losing the low-\ac{PAPR} benefit.
}

{\color{\reviewColor}
\subsubsection{User multiplexing and Transmitter}
\label{subsec:formatzeroTX}
For this scheme, we consider orthogonal sequence selection for user multiplexing and \ac{UCI} transmission. We assume that all scheduled users utilize the $\indexSequence$th \ac{GCP} $(\seqGc_\indexSequence,\seqGd_\indexSequence)$ and \ac{GCP} $(\seqGa,\seqGb)$. No overhead due to the reference symbols is introduced as in \ac{NR} \ac{PUCCH} Format 0.
The interlace has $\RBsize\numberOfClusters$ non-zero elements. Hence, it is possible to generate $\RBsize\numberOfClusters$  orthogonal resources which can be shared by up to $\RBsize\numberOfClusters$ users. 
By exploiting the property of unimodular sequences as discussed in Section~\ref{sec:unimodular}, the $\shift[]$th orthogonal resource can be generated by multiplying the $\indexEleOfSeq$th non-zero element of the \ac{CS}  with  $\constante^{{\shift[]}\frac{2\pi\constanti}{\RBsize\numberOfClusters}\times\indexEleOfSeq}$ for ${\shift[]},\indexEleOfSeq\in\{0,1,\mydots,\RBsize\numberOfClusters-1\}$. 
 However, the orthogonality between the sequences cannot be kept under dispersive channels. 
 To circumvent this issue,  the phase rotation is restricted to be periodic with the period of $\RBsize$, i.e.,  ${\shift[]}=\{0,\numberOfClusters,2\numberOfClusters,\mydots,(\RBsize-1)\numberOfClusters\}$. Therefore, the orthogonality  {\em within} the \ac{PRB} can still be maintained if the channel between  each user and the base station is assumed to be flat  within the bandwidth of the \ac{PRB}.
 Under this restriction, the maximum number of users that can be supported reduces to $\RBsize$, but a low-complexity receiver can be employed. Note that corresponding modulation operation in the frequency domain can be effectively implemented through uniformly separated cyclic shifts in time as shown in \figurename~\ref{fig:tx}.  If each \ac{PRB} consists of $\RBsize=12$ subcarriers, 12 orthogonal resources in the interlace can be shared by 6 users to transmit 1-bit information (e.g., \ac{ACK}/\ac{NACK}) or 3 users to transmit 2-bit information (e.g., \ac{ACK}/\ac{NACK} and scheduling request) with $M$-ary orthogonal signaling, where $M=2$ for 1 \ac{UCI} bit or $M=4$ for 2 \ac{UCI} bits. Each user selects one of the orthogonal resources to indicate \ac{UCI}.  

\subsubsection{Receiver Design}
The receiver exploits the orthogonality of the sequences in each \ac{PRB} to decode the information. Since there is no reference symbol, the receiver first calculates the absolute square of the matched filter output for the $\shift[]$th orthogonal sequence in each \ac{PRB}. It then combines the results for both \ac{ACK} and \ac{NACK} at different \acp{PRB} to obtain the test statistic. If the test statistic is lower than a threshold determined by Neyman-Pearson criterion, it is considered as \ac{DTX}, i.e., the user could not decode the signal in the downlink. Otherwise, by comparing the matched filter results, it determines if it is \ac{ACK} or \ac{NACK}. Note that the detected sequence can indicate \ac{ACK}/\ac{NACK} and/or  \ac{SR}, depending on the network configuration. 
}

\subsubsection{Co-channel interference mitigation}
\label{subsec:cciset}
For the \ac{CCI} minimization, the cross-correlation between any two sequences in $\setC$ and $\setD$ should be as low as possible.
Due to the multipath channel, the signals may be exposed to additional shift in time within the \ac{CP}. Therefore, the cross-correlation analysis should consider the largest possible correlation in {\em time}. 
{\color{\reviewColor}
In the frequency domain, the peak cross-correlation between $\seqGc_\indexSequence$ and $\seqGc_j$ can be defined as
\begin{align}
\peakCrossCorrelatation[(\seqGc_\indexSequence,\seqGc_j)]\triangleq\max_{{\shift[]{}\in[0,\RBsize-1]}}\langle\seqGc_\indexSequence,\seqGc_j\odot \seqGsShift[{\shift[]}] \rangle
\end{align}
where $\seqGsShift[{\shift[]}] = {(\constante^{\shift[]\frac{2\pi{\constanti}}{\RBsize}\times0}, \constante^{\shift[]\frac{2\pi{\constanti}}{\RBsize}\times1}, \dots, \constante^{\shift[]\frac{2\pi{\constanti}}{\RBsize}\times(\RBsize-1)} )}$. Therefore, the maximum peak cross-correlation for both sequences in $\setC$ and $\setD$ should be minimized, i.e., $\peakCrossCorrelatation[(\seqGc_\indexSequence,\seqGc_j)]\le\thershold$ and $\peakCrossCorrelatation[(\seqGd_\indexSequence,\seqGd_j)]\le\thershold$
for $\indexSequence\neq j$ and $\indexSequence,j\in \{1,2,\mydots,\numberOfSequences\}$, where $\thershold$ is a threshold.}
In \ac{NR}, the number of available base sequences  is set to $\numberOfSequences=30$ for $\RBsize = 12$ \cite{nr_phy_2020} and the maximum peak cross-correlation is $0.8$. Hence, we also target the same number of sequences in  $\setC$ and $\setD$ and a similar or better maximum peak cross-correlation with \acp{CS}. This naturally leads to the following question for the proposed scheme: Do there exist $\setC$ and $\setD$ for $\numberOfSequences=30$ and $\RBsize=12$ such that the maximum peak cross-correlation between any two \acp{CS} is less than $0.8$?

To answer this question, we propose a simple search algorithm which exploits the exhaustively generated \acp{GCP}  in \cite{holzmann_1991} for length 12 to obtain $\setC$ and $\setD$. We initialize the algorithm with $\seedSize=52$ \acp{GCP} of length 12 listed in  \cite{holzmann_1991} and populate as $\seedA=\{\seedSeqGcCandidate[1],\mydots,\seedSeqGcCandidate[\seedSize]\}$ and $\seedB=\{\seedSeqGdCandidate[1],\mydots,\seedSeqGdCandidate[\seedSize]\}$. For the $\indexForFirst$th seed \ac{GCP} $({\seedSeqGcCandidate[\indexForFirst]},{\seedSeqGdCandidate[\indexForFirst]})$, we first enumerate $\seedSizeEq=8$ equivalent \acp{GCP} by interchanging, reflecting both (i.e., reversing the order of the elements of the sequences), and conjugate reflecting original sequences in the seed \ac{GCP}, which lead to the sets $\seedAp=\{\seqGcCandidate[1],\mydots,\seqGcCandidate[\seedSizeEq]\}$ and  $\seedBp=\{\seqGdCandidate[1],\mydots,\seqGdCandidate[\seedSizeEq]\})$. Because of the properties of \ac{GCP}, the $(\seqGcCandidate[\indexForSecond],\seqGdCandidate[\indexForSecond])$ still constructs \acp{GCP} for $\indexForSecond=1,\mydots,\seedSizeEq$. For a given candidate \ac{GCP} $(\seqGcCandidate[\indexForSecond],\seqGdCandidate[\indexForSecond])$, we calculate $\langle\seqGc_\indexSequence,\seqGcCandidate[\indexForSecond]\odot \seqGsShift[{\shift[]}] \rangle$ and  $\langle\seqGd_\indexSequence,\seqGdCandidate[\indexForSecond]\odot \seqGsShift[{\shift[]}] \rangle$  for $\seqGc_\indexSequence\in\setC$ and $\seqGd_\indexSequence\in\setD$ and  $\shift[]\in\{0,1/\lengthGcGd\upsampleValueU,\mydots,(\lengthGcGd\upsampleValueU-1)/\lengthGcGd\upsampleValueU\}$ and  $\upsampleValueU>1$. If the results are less than or equal to $\thershold$ for all $\shift[]$, we update $\setC$ and $\setD$ by including the sequences in the candidate \ac{GCP} to the sets. 

\begin{table}[t]
    \centering
	\caption{The sequences in $\setC$ and $\setD$}
	\label{table:seqC}
	\if \IEEEsubmission1
		\resizebox{2.6in}{!}
		{
			\begin{tabular}{r|cc}
				$\indexSequence$  & $\seqGc_\indexSequence$ & $\seqGd_\indexSequence$ \\
				\hline
				1  & ({+},{-},{i},{j},{+},{-},{-},{-},{+},{+},{+},{+}) & ({-},{+},{-},{+},{+},{-},{+},{+},{j},{j},{+},{+}) \\
				2  & ({+},{-},{j},{i},{+},{-},{-},{-},{+},{+},{+},{+}) & ({-},{+},{-},{+},{+},{-},{+},{+},{i},{i},{+},{+}) \\
				3  & ({+},{+},{+},{+},{i},{+},{-},{j},{+},{-},{-},{+}) & ({+},{+},{-},{-},{i},{+},{+},{i},{+},{-},{+},{-}) \\
				4  & ({+},{-},{-},{+},{i},{-},{+},{j},{+},{+},{+},{+}) & ({-},{+},{-},{+},{j},{+},{+},{j},{-},{-},{+},{+}) \\
				5  & ({+},{+},{+},{+},{j},{+},{-},{i},{+},{-},{-},{+}) & ({+},{+},{-},{-},{j},{+},{+},{j},{+},{-},{+},{-}) \\
				6  & ({+},{-},{-},{+},{j},{-},{+},{i},{+},{+},{+},{+}) & ({-},{+},{-},{+},{i},{+},{+},{i},{-},{-},{+},{+}) \\
				7  & ({+},{+},{+},{+},{i},{-},{+},{j},{+},{-},{-},{+}) & ({+},{+},{-},{-},{i},{-},{-},{i},{+},{-},{+},{-}) \\
				8  & ({+},{-},{-},{+},{i},{+},{-},{j},{+},{+},{+},{+}) & ({-},{+},{-},{+},{j},{-},{-},{j},{-},{-},{+},{+}) \\
				9  & ({+},{+},{+},{+},{j},{-},{+},{i},{+},{-},{-},{+}) & ({+},{+},{-},{-},{j},{-},{-},{j},{+},{-},{+},{-}) \\
				10 & ({+},{-},{-},{+},{j},{+},{-},{i},{+},{+},{+},{+}) & ({-},{+},{-},{+},{i},{-},{-},{i},{-},{-},{+},{+}) \\
				11 & ({+},{+},{-},{+},{-},{j},{+},{j},{-},{+},{+},{+}) & ({-},{-},{+},{-},{j},{-},{i},{-},{-},{+},{+},{+}) \\
				12 & ({+},{+},{-},{+},{-},{i},{+},{i},{-},{+},{+},{+}) & ({-},{-},{+},{-},{i},{-},{j},{-},{-},{+},{+},{+}) \\
				13 & ({+},{+},{+},{-},{-},{i},{-},{j},{-},{+},{-},{-}) & ({+},{+},{+},{-},{j},{+},{j},{-},{+},{-},{+},{+}) \\
				14 & ({+},{+},{+},{-},{-},{j},{-},{i},{-},{+},{-},{-}) & ({+},{+},{+},{-},{i},{+},{i},{-},{+},{-},{+},{+}) \\
				15 & ({+},{+},{-},{+},{+},{j},{-},{j},{-},{+},{+},{+}) & ({-},{-},{+},{-},{j},{+},{i},{+},{-},{+},{+},{+}) \\
				16 & ({+},{+},{-},{+},{+},{i},{-},{i},{-},{+},{+},{+}) & ({-},{-},{+},{-},{i},{+},{j},{+},{-},{+},{+},{+}) \\
				17 & ({+},{+},{+},{-},{+},{i},{+},{j},{-},{+},{-},{-}) & ({+},{+},{+},{-},{j},{-},{j},{+},{+},{-},{+},{+}) \\
				18 & ({+},{+},{+},{-},{+},{j},{+},{i},{-},{+},{-},{-}) & ({+},{+},{+},{-},{i},{-},{i},{+},{+},{-},{+},{+}) \\
				19 & ({+},{+},{+},{-},{i},{i},{+},{-},{+},{+},{-},{+}) & ({+},{+},{+},{-},{-},{-},{j},{i},{-},{-},{+},{-}) \\
				20 & ({+},{-},{+},{+},{-},{+},{j},{j},{-},{+},{+},{+}) & ({-},{+},{-},{-},{j},{i},{-},{-},{-},{+},{+},{+}) \\
				21 & ({+},{+},{+},{-},{j},{j},{-},{+},{+},{+},{-},{+}) & ({+},{+},{+},{-},{+},{+},{i},{j},{-},{-},{+},{-}) \\
				22 & ({+},{-},{+},{+},{+},{-},{i},{i},{-},{+},{+},{+}) & ({-},{+},{-},{-},{i},{j},{+},{+},{-},{+},{+},{+}) \\
				23 & ({+},{+},{+},{i},{-},{+},{-},{-},{i},{+},{-},{+}) & ({+},{+},{+},{i},{-},{+},{+},{+},{j},{-},{+},{-}) \\
				24 & ({+},{+},{+},{j},{-},{+},{-},{-},{j},{+},{-},{+}) & ({+},{+},{+},{j},{-},{+},{+},{+},{i},{-},{+},{-}) \\
				25 & ({+},{+},{-},{+},{+},{+},{j},{i},{-},{-},{-},{+}) & ({+},{+},{-},{+},{j},{j},{+},{-},{+},{+},{+},{-}) \\
				26 & ({+},{-},{-},{-},{j},{i},{+},{+},{+},{-},{+},{+}) & ({-},{+},{+},{+},{-},{+},{i},{i},{+},{-},{+},{+}) \\
				27 & ({+},{+},{-},{+},{-},{-},{i},{j},{-},{-},{-},{+}) & ({+},{+},{-},{+},{i},{i},{-},{+},{+},{+},{+},{-}) \\
				28 & ({+},{-},{-},{-},{i},{j},{-},{-},{+},{-},{+},{+}) & ({-},{+},{+},{+},{+},{-},{j},{j},{+},{-},{+},{+}) \\
				29 & ({+},{+},{-},{+},{i},{+},{-},{i},{-},{-},{+},{+}) & ({+},{+},{-},{+},{i},{+},{+},{j},{+},{+},{-},{-}) \\
				30 & ({+},{+},{-},{-},{j},{-},{+},{j},{+},{-},{+},{+}) & ({-},{-},{+},{+},{i},{+},{+},{j},{+},{-},{+},{+})
			\end{tabular}
		}
	\else
		\begin{tabular}{r|cc}
			$\indexSequence$  & $\seqGc_\indexSequence$ & $\seqGd_\indexSequence$ \\
			\hline
			1  & ({+},{-},{i},{j},{+},{-},{-},{-},{+},{+},{+},{+}) & ({-},{+},{-},{+},{+},{-},{+},{+},{j},{j},{+},{+}) \\
			2  & ({+},{-},{j},{i},{+},{-},{-},{-},{+},{+},{+},{+}) & ({-},{+},{-},{+},{+},{-},{+},{+},{i},{i},{+},{+}) \\
			3  & ({+},{+},{+},{+},{i},{+},{-},{j},{+},{-},{-},{+}) & ({+},{+},{-},{-},{i},{+},{+},{i},{+},{-},{+},{-}) \\
			4  & ({+},{-},{-},{+},{i},{-},{+},{j},{+},{+},{+},{+}) & ({-},{+},{-},{+},{j},{+},{+},{j},{-},{-},{+},{+}) \\
			5  & ({+},{+},{+},{+},{j},{+},{-},{i},{+},{-},{-},{+}) & ({+},{+},{-},{-},{j},{+},{+},{j},{+},{-},{+},{-}) \\
			6  & ({+},{-},{-},{+},{j},{-},{+},{i},{+},{+},{+},{+}) & ({-},{+},{-},{+},{i},{+},{+},{i},{-},{-},{+},{+}) \\
			7  & ({+},{+},{+},{+},{i},{-},{+},{j},{+},{-},{-},{+}) & ({+},{+},{-},{-},{i},{-},{-},{i},{+},{-},{+},{-}) \\
			8  & ({+},{-},{-},{+},{i},{+},{-},{j},{+},{+},{+},{+}) & ({-},{+},{-},{+},{j},{-},{-},{j},{-},{-},{+},{+}) \\
			9  & ({+},{+},{+},{+},{j},{-},{+},{i},{+},{-},{-},{+}) & ({+},{+},{-},{-},{j},{-},{-},{j},{+},{-},{+},{-}) \\
			10 & ({+},{-},{-},{+},{j},{+},{-},{i},{+},{+},{+},{+}) & ({-},{+},{-},{+},{i},{-},{-},{i},{-},{-},{+},{+}) \\
			11 & ({+},{+},{-},{+},{-},{j},{+},{j},{-},{+},{+},{+}) & ({-},{-},{+},{-},{j},{-},{i},{-},{-},{+},{+},{+}) \\
			12 & ({+},{+},{-},{+},{-},{i},{+},{i},{-},{+},{+},{+}) & ({-},{-},{+},{-},{i},{-},{j},{-},{-},{+},{+},{+}) \\
			13 & ({+},{+},{+},{-},{-},{i},{-},{j},{-},{+},{-},{-}) & ({+},{+},{+},{-},{j},{+},{j},{-},{+},{-},{+},{+}) \\
			14 & ({+},{+},{+},{-},{-},{j},{-},{i},{-},{+},{-},{-}) & ({+},{+},{+},{-},{i},{+},{i},{-},{+},{-},{+},{+}) \\
			15 & ({+},{+},{-},{+},{+},{j},{-},{j},{-},{+},{+},{+}) & ({-},{-},{+},{-},{j},{+},{i},{+},{-},{+},{+},{+}) \\
			16 & ({+},{+},{-},{+},{+},{i},{-},{i},{-},{+},{+},{+}) & ({-},{-},{+},{-},{i},{+},{j},{+},{-},{+},{+},{+}) \\
			17 & ({+},{+},{+},{-},{+},{i},{+},{j},{-},{+},{-},{-}) & ({+},{+},{+},{-},{j},{-},{j},{+},{+},{-},{+},{+}) \\
			18 & ({+},{+},{+},{-},{+},{j},{+},{i},{-},{+},{-},{-}) & ({+},{+},{+},{-},{i},{-},{i},{+},{+},{-},{+},{+}) \\
			19 & ({+},{+},{+},{-},{i},{i},{+},{-},{+},{+},{-},{+}) & ({+},{+},{+},{-},{-},{-},{j},{i},{-},{-},{+},{-}) \\
			20 & ({+},{-},{+},{+},{-},{+},{j},{j},{-},{+},{+},{+}) & ({-},{+},{-},{-},{j},{i},{-},{-},{-},{+},{+},{+}) \\
			21 & ({+},{+},{+},{-},{j},{j},{-},{+},{+},{+},{-},{+}) & ({+},{+},{+},{-},{+},{+},{i},{j},{-},{-},{+},{-}) \\
			22 & ({+},{-},{+},{+},{+},{-},{i},{i},{-},{+},{+},{+}) & ({-},{+},{-},{-},{i},{j},{+},{+},{-},{+},{+},{+}) \\
			23 & ({+},{+},{+},{i},{-},{+},{-},{-},{i},{+},{-},{+}) & ({+},{+},{+},{i},{-},{+},{+},{+},{j},{-},{+},{-}) \\
			24 & ({+},{+},{+},{j},{-},{+},{-},{-},{j},{+},{-},{+}) & ({+},{+},{+},{j},{-},{+},{+},{+},{i},{-},{+},{-}) \\
			25 & ({+},{+},{-},{+},{+},{+},{j},{i},{-},{-},{-},{+}) & ({+},{+},{-},{+},{j},{j},{+},{-},{+},{+},{+},{-}) \\
			26 & ({+},{-},{-},{-},{j},{i},{+},{+},{+},{-},{+},{+}) & ({-},{+},{+},{+},{-},{+},{i},{i},{+},{-},{+},{+}) \\
			27 & ({+},{+},{-},{+},{-},{-},{i},{j},{-},{-},{-},{+}) & ({+},{+},{-},{+},{i},{i},{-},{+},{+},{+},{+},{-}) \\
			28 & ({+},{-},{-},{-},{i},{j},{-},{-},{+},{-},{+},{+}) & ({-},{+},{+},{+},{+},{-},{j},{j},{+},{-},{+},{+}) \\
			29 & ({+},{+},{-},{+},{i},{+},{-},{i},{-},{-},{+},{+}) & ({+},{+},{-},{+},{i},{+},{+},{j},{+},{+},{-},{-}) \\
			30 & ({+},{+},{-},{-},{j},{-},{+},{j},{+},{-},{+},{+}) & ({-},{-},{+},{+},{i},{+},{+},{j},{+},{-},{+},{+})
		\end{tabular}
	\fi
	\vspace{-10pt}
\end{table}

We list the sets obtained for $\seqGc_\indexSequence$ and $\seqGd_\indexSequence$ in Table \ref{table:seqC} for $\thershold=0.715$ and $\upsampleValueU=128$. With the aforementioned procedure, 
we could not obtain $\setC$ and $\setD$ when $\thershold<0.715$ for $\numberOfSequences=30$ and $\RBsize=12$. However, the numerical results given in Section~\ref{sec:numerical} show that the maximum peak cross-correlation is still less than the ones for \ac{ZC} sequences and the sequences adopted in \ac{NR} \cite{nr_phy_2020}. 
{\color{\reviewColor}
In \cite{interDigitalGCSeval}, a comparison for the peak cross-correlation for different sequences sets is provided. The comparison shows that reducing maximum peak cross-correlation less than $0.7$ is challenging under \ac{PAPR} and \ac{QPSK} alphabet constraints.}
It is also worth noting that the sets obtained for $\seqGc_\indexSequence$ and $\seqGd_\indexSequence$ are not unique and are based on the initial seed sequences.

{\color{\reviewColor}
\subsection{Transmission for more than 2 \ac{UCI} bits}
\label{subsec:morebits}
In \ac{NR},  $\RBsize$ is fixed to $12=2^2\times3$ subcarriers. Assuming an even number of non-zero clusters, e.g., $\numberOfClusters=10$ clusters, we set $\numberOfIterations=3$.} 
Let  $(\seqGa,\seqGb)$ be a \ac{GCP} of length $3$, and $(\seqGc,\seqGd)$ be a \ac{GCP} of length $\numberOfClusters/2$. Based on Theorem~\ref{th:reduced}, the following configurations result in \acp{CS} compatible with the interlace in \figurename~\ref{fig:interlace}:
\begin{itemize}
\item Configuration~1: $(\seqGcRecursion[1],\seqGdRecursion[1])=((\constantOne),(\constantOne))$, $(\seqGcRecursion[2],\seqGdRecursion[2])=(\seqGa,\seqGb)$, $(\seqGcRecursion[3],\seqGdRecursion[3])=(\upsampleOp[{2(\numberOfNulls+\RBsize)}][\seqGc],\upsampleOp[{2(\numberOfNulls+\RBsize)}][\seqGd])$, $\varUpsample=3$, 
$\separationGolay[{\permutationMono[{\indexIteration}]=1}]=\numberOfNulls+\RBsize-4\varUpsample$, and $\separationGolay[{\permutationMono[{\indexIteration}]\neq1}]=0$
\item Configuration~2: $(\seqGcRecursion[1],\seqGdRecursion[1])=((\constantOne),(\constantOne))$, $(\seqGcRecursion[2],\seqGdRecursion[2])=(\seqGa,\seqGb)$, $(\seqGcRecursion[3],\seqGdRecursion[3])=(\upsampleOp[{\numberOfNulls+\RBsize}][\seqGc],\upsampleOp[{\numberOfNulls+\RBsize}][\seqGd])$, $\varUpsample=3$, $\separationGolay[{\permutationMono[{\indexIteration}]=1}]=(\numberOfNulls+\RBsize)\numberOfClusters/2-4\varUpsample$, and $\separationGolay[{\permutationMono[{\indexIteration}]\neq1}]=0$
\item Configuration~3: $(\seqGcRecursion[1],\seqGdRecursion[1])=((\constantOne),(\constantOne))$, $(\seqGcRecursion[2],\seqGdRecursion[2])=(\upsampleOp[{4}][\seqGa],\upsampleOp[{4}][\seqGb])$, $(\seqGcRecursion[3],\seqGdRecursion[3])=(\upsampleOp[{2(\numberOfNulls+\RBsize)}][\seqGc],\upsampleOp[{2(\numberOfNulls+\RBsize)}][\seqGd])$, $\varUpsample=1$, $\separationGolay[{\permutationMono[{\indexIteration}]=1}]=\numberOfNulls+\RBsize-4\varUpsample$, and $\separationGolay[{\permutationMono[{\indexIteration}]\neq1}]=0$
\item Configuration~4: $(\seqGcRecursion[1],\seqGdRecursion[1])=((\constantOne),(\constantOne))$, $(\seqGcRecursion[2],\seqGdRecursion[2])=(\upsampleOp[{4}][\seqGa],\upsampleOp[{4}][\seqGb])$, $(\seqGcRecursion[3],\seqGdRecursion[3])=(\upsampleOp[{\numberOfNulls+\RBsize}][\seqGc],\upsampleOp[{\numberOfNulls+\RBsize}][\seqGd])$, $\varUpsample=1$, $\separationGolay[{\permutationMono[{\indexIteration}]=1}]=(\numberOfNulls+\RBsize)\numberOfClusters/2-4\varUpsample$, and $\separationGolay[{\permutationMono[{\indexIteration}]\neq1}]=0$
\end{itemize}

While Configuration 1 and Configuration 2 cascade sequences $\seqGa$, $\seqGb$, $\tilde{\seqGa}$, and $\tilde{\seqGb}$ in each \ac{PRB}, Configuration 3 and Configuration 4 interleave the elements of these sequences. The difference between Configuration 1 and Configuration 2 is that they shuffle the sequences in \acp{PRB} in the interlace in a different order because of the choices of up-sampling factors and $\separationGolay[{\permutationMono[{\indexIteration}]=1}]$. Similarly, Configuration 3 and Configuration 4 yield different orders in the interlace. For these configurations, $\seqGa$, $\seqGb$, $\tilde{\seqGa}$, and $\tilde{\seqGb}$ are multiplied the elements of $\seqGc$, $\seqGd$, $\tilde{\seqGc}$, and $\tilde{\seqGd}$ and $2^3$ \ac{QPSK} symbols based on $\seqPermutationCompShift$ and $\seqPermutationOrder$ for $\numberOfPointsForPSK=4$. {\color{\reviewColor} The distribution of the \ac{QPSK} symbols to \acp{PRB} for different configurations are illustrated in \figurename~\ref{fig:configurations}.}
For all configurations, $\numberOfNulls$ and $\numberOfClusters$ can chosen flexibly without  concern of increasing the \ac{PAPR}.

\begin{figure}[t]
	\centering
	\subfloat[\color{\reviewColor}Configuration 1.]{\includegraphics[width =3.4in]{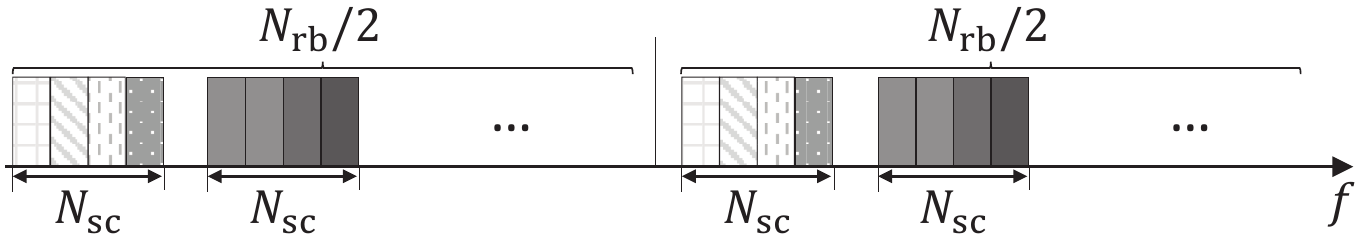}\label{subfig:c1}}\\
	\subfloat[\color{\reviewColor}Configuration 2.]{\includegraphics[width =3.4in]{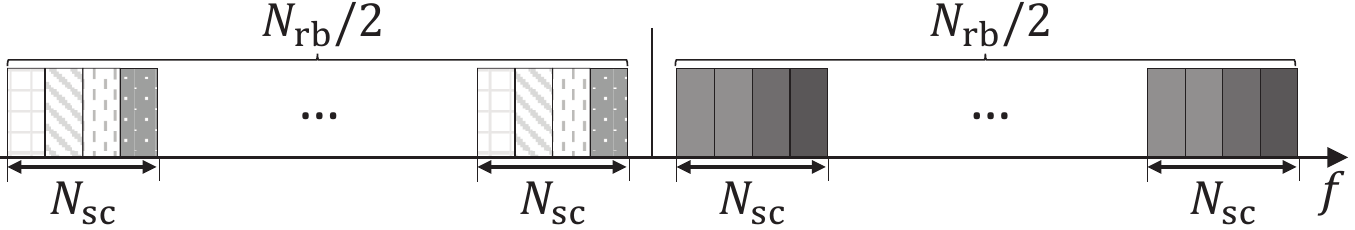}\label{subfig:c2}}\\	
	\subfloat[\color{\reviewColor}Configuration 3.]{\includegraphics[width =3.4in]{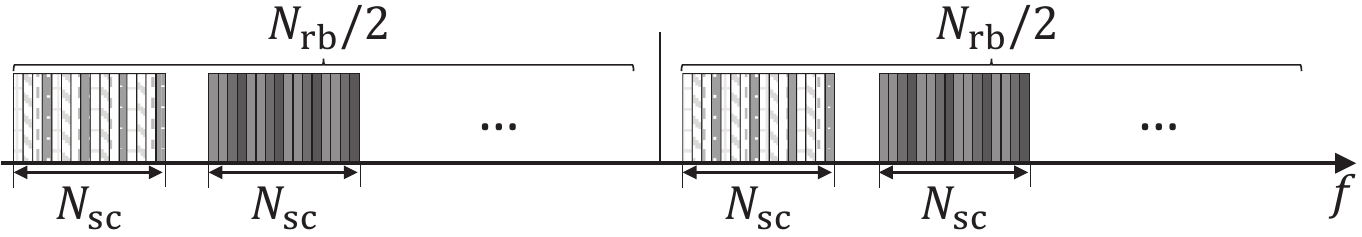}\label{subfig:c3}}\\
\subfloat[\color{\reviewColor}Configuration 4.]{\includegraphics[width =3.4in]{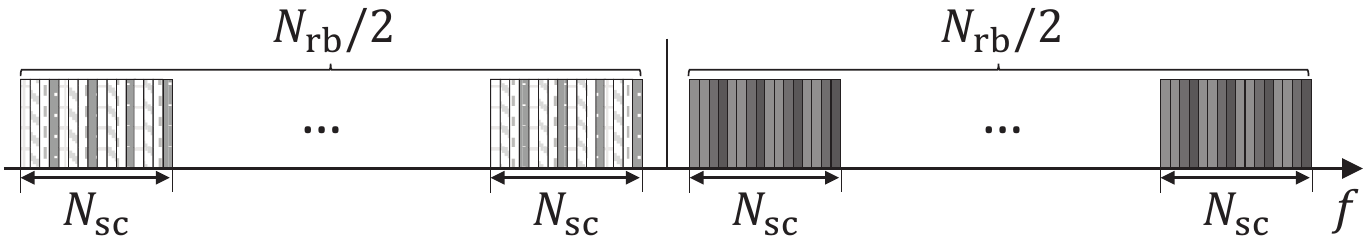}\label{subfig:c4}}	
	\caption{{\color{\reviewColor}Configurations. Each color tone represents one of the $2^3$ \ac{QPSK} symbols distributed to \acp{PRB} through $\seqGa$, $\seqGb$, $\tilde{\seqGa}$, and $\tilde{\seqGb}$ of length $3$.}}
	\label{fig:configurations}
\end{figure}

{\color{\reviewColor}
Each configuration   leads to $(3!)^2\numberOfPointsForPSK^{\numberOfIterations+1}$ \acp{CS} compatible with the interlace structure in \ac{NR-U} since the number of seed \acp{CS} that are not co-linear $\numberSequencesNotColinear$ is  $3$ and $\numberOfIterations=3$. For $\numberOfPointsForPSK=4$, it gives $9216$ \acp{CS} with a \ac{QPSK} alphabet. }
By including the distinct combinations of $(\seqGa,\seqGb)$ and $(\seqGc,\seqGd)$, i.e., interchanging the sequences in a \ac{GCP}, e.g.,  $(\seqGb,\seqGa)$, or conjugate reflecting one of the sequences in a \ac{GCP}, e.g., $(\tilde{\seqGc},\seqGd)$, the number of distinct \acp{CS} increases by a factor 64.
Therefore, $\lfloor\log_2 64\times9216 \rfloor =19$ information bits can be transmitted for each configuration. As a result, overall, there exist at least $4\times64\times9216=2359296$ \acp{CS}, which can carry a maximum of $21$ \ac{UCI} bits.



\subsubsection{User multiplexing and Transmitter}
\label{subsec:morebitsAndusermultiplexing}
To allow user multiplexing in the interlace while enabling a low-complexity receiver, we keep the orthogonality of the sequences from different users in each \ac{PRB}. To meet this condition, we exploit the property of the unimodular sequences and consider only one of the configurations, e.g., Configuration 1. We obtain three orthogonal sequences by modulating unimodular $\seqGa$ and $\seqGb$ as $\seqGa \odot \seqGsShift[{\shift[]}]$ and $\seqGb \odot \seqGsShift[{\shift[]}]$, respectively, where $\seqGsShift[{\shift[]}] = {(\constante^{\shift[]\frac{2\pi{\constanti}}{3}\times0}, \constante^{\shift[]\frac{2\pi{\constanti}}{3}\times1}, \constante^{\frac{\shift[]2\pi{\constanti}}{3}\times2} )}$ and $\shift[]\in\{0,1,2\}$.
In addition, we fix the locations of $\seqGa \odot \seqGsShift[{\shift[]}]$ and $\seqGb \odot \seqGsShift[{\shift[]}]$ on each \ac{PRB} by setting $\permutationOrderEle[\numberOfIterations=3]=2$ and $\permutationMono[{\numberOfIterations=3}]=1$, which result in  $\seqPermutationCompShift\in\{(3,2,1),(2,3,1)\}$ and  $\seqPermutationOrder\in\{(3,1,2),(1,3,2)\}$. The rationale behind this choice can be understood by expressing $\funcfForCommonOrder(\seqx,\polyVariable)$ for $\permutationOrderEle[\numberOfIterations=3]=2$ and $\permutationMono[{\numberOfIterations=3}]=1$ as
\begin{align}
\funcfForCommonOrder(\seqx,\polyVariable) &=  	\funcForCommonOrder(\seqx,\polyVariable) (\polySeq[{\seqGa \odot \seqGs}][\polyVariable](1 - \monomial[{1}])_2 + 
\polySeq[{\seqGb \odot \seqGs}][\polyVariable]\monomial[{1}])~. \label{eq:userMultiplex}
\end{align}
While $\funcForCommonOrder(\seqx,\polyVariable)$ takes different values depending on $\seqGc$, $\seqGd$, and the first two elements of $\seqPermutationCompShift$ and $\seqPermutationOrder$, the remaining term in \eqref{eq:userMultiplex} places $\seqGa \odot \seqGs$ and $\seqGb \odot \seqGsShift[{\shift[]}]$ in a fixed order. Therefore, the sequences $\seqGa \odot \seqGsShift[{\shift[]}]$ and $\seqGb \odot \seqGsShift[{\shift[]}]$  are multiplied with the elements of $\seqGc$, $\seqGd$,  $\tilde{\seqGc}$, $\tilde{\seqGd}$ and the outcome of $\funcfForCommonPhaseA(\seqx)$. 
The proposed scheme enables three users to transmit  $\log_2({2^2\numberOfPointsForPSK^{\numberOfIterations+1}}) = 10$ bits on the same interlace for a given \ac{GCP} $(\seqGc,\seqGd)$. 
The number of bits can be increased if the seed sequences $\seqGc$ and $\seqGd$ are modified. For example, if the sequences in $(\seqGc,\seqGd)$ are interchanged, the number of bits can be increased to 11 bits. Note that modifying $\seqGc$ and $\seqGd$ does not destroy the orthogonality between the sequences for different users as  $\seqGa \odot \seqGsShift[{\shift[]}]$ and $\seqGb \odot \seqGsShift[{\shift[]}]$  are multiplied with scalars depending on the elements of $\seqGc$, $\seqGd$,  $\tilde{\seqGc}$, and $\tilde{\seqGd}$. 

\begin{figure*}[t]
	\centering
	{\includegraphics[width =6in]{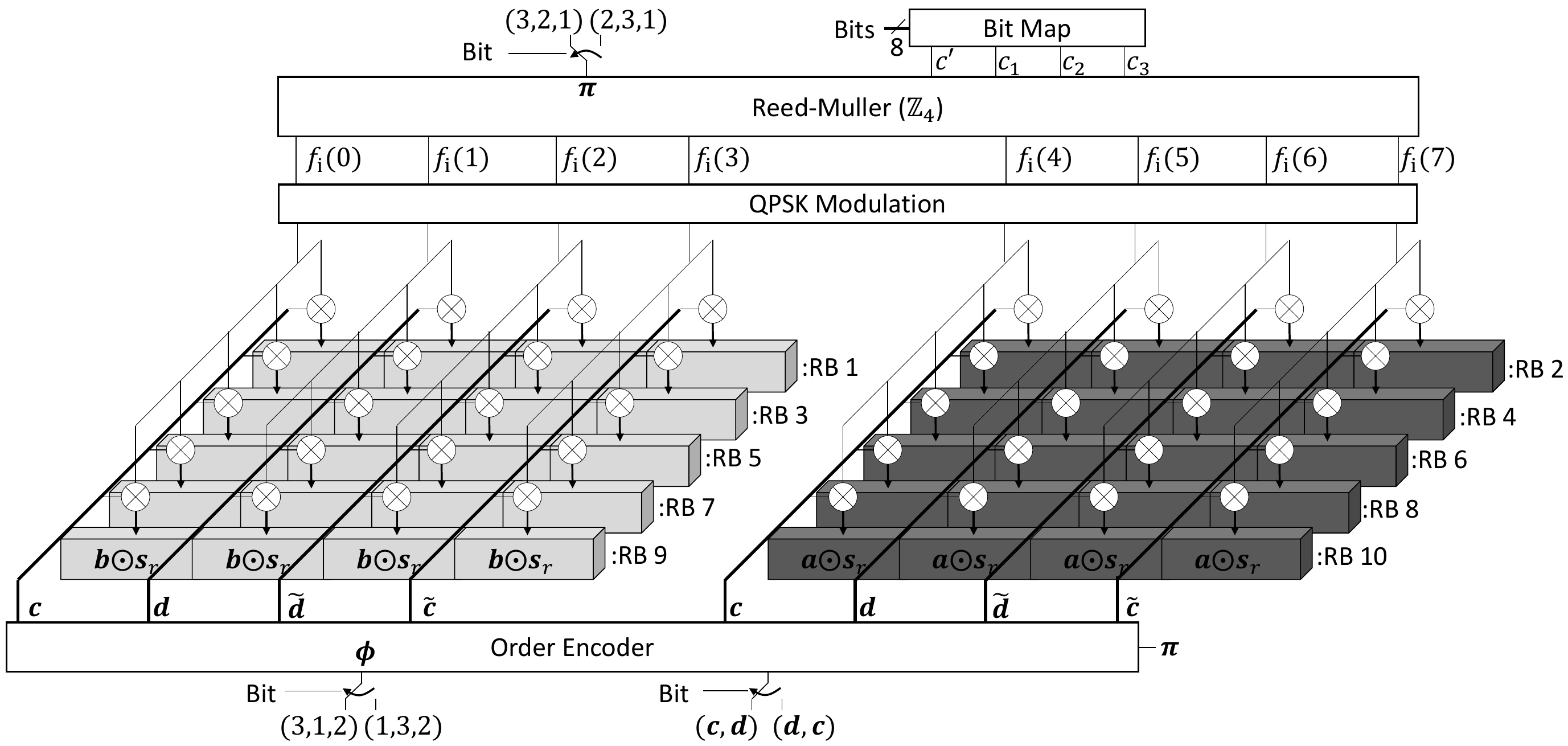}
	}
	\caption{Transmitter for more than 2 \ac{UCI} bits based on Configuration~1. The phases of $\seqGa$ and $\seqGb$ are modified based on the order encoder and the \ac{RM} code.}
	\label{fig:txCoherenta}
\end{figure*}

In \figurename~\ref{fig:txCoherenta}, we illustrate how $\seqGa$ and $\seqGb$ are placed in each \ac{PRB} and modified based on  $\funcfForCommonPhaseA(\seqx)$ and $\funcfForCommonOrder(\seqx,\polyVariable)$ for $\numberOfClusters=10$ with the proposed scheme. Interchanging $\seqGa$ and $\seqGb$, the values of $\seqPermutationCompShift$ and $\seqPermutationOrder$ are controlled with 3 \ac{UCI} bits. The output of the order encoder is exemplified for $\seqPermutationCompShift=(3,2,1)$ and $\seqPermutationOrder=(3,1,2)$, which gives $(\seqGc,\seqGd,\tilde{\seqGd},\tilde{\seqGc},\seqGc,\seqGd,\tilde{\seqGd},\tilde{\seqGc})$. The parameters $\angleexp[\indexIteration=1,2,3],\arbitraryPhaseA \in\integers_4$ are set based on 8 \ac{UCI} bits. The bit mapping is done based on a Gray mapping, e.g.,  $00\rightarrow0$, $01\rightarrow1$, $10\rightarrow3$, and $11\rightarrow2$, to improve the error rate performance. Since the output of the \ac{RM} code over $\integers_4$ is distributed to 5 different \acp{PRB} as in \figurename~\ref{fig:txCoherenta}, the proposed scheme inherently exploits the frequency diversity in the frequency selective channels.

\subsubsection{Receiver Design}
\label{subsec:receiverCoherent}
The proposed scheme is compatible with the \ac{RM} code in \cite{davis_1999}. Thus, a simple receiver can be developed by re-using the \ac{ML} decoder proposed in \cite{Schmidt2005} for first-order \ac{RM} codes. At the receiver, we first separate the users by applying the matched filters for $\seqGa \odot \seqGsShift[{\shift[]}]$ and $\seqGb \odot \seqGsShift[{\shift[]}]$ for different $\shift[]$, which lead to 4 complex values for each \ac{PRB} and user. We then coherently combine $\numberOfClusters/2$ complex values distributed to $\numberOfClusters/2$ different \acp{PRB} based on {\color{\reviewColor}\ac{MRC}}. Subsequently, we use the \ac{ML} decoder in \cite{Schmidt2005} to obtain $\angleexp[\indexIteration=1,2,3]$ and $\arbitraryPhaseA$. We perform this operation for 8 different hypotheses due to the  combinations of interchange of $\seqGa$ and $\seqGb$, $\seqPermutationCompShift$, and $\seqPermutationOrder$. We choose the best hypothesis based on  \ac{ML}. 

{\color{\reviewColor}
The receiver for this scheme requires the estimate of the channel between the receiver and each user. The channel estimation for multiple users can be achieved by using a dedicated \ac{OFDM} symbol constructed with the proposed scheme with a set of fixed parameters. Since $\{\seqGa \odot \seqGsShift[{0}],\seqGa \odot \seqGsShift[{1}],\seqGa \odot \seqGsShift[{2}]\}$ and $\{\seqGb \odot \seqGsShift[{0}],\seqGb \odot \seqGsShift[{1}],\seqGb \odot \seqGsShift[{2}]\}$ are orthogonal sets, the receiver can estimate the channel for each user with a set of  matched filters in the frequency domain for fixed $\seqPermutationCompShift$, $\seqPermutationOrder$, $\seqGc$, $\seqGd$, $\angleexp[\indexIteration]$ and $\arbitraryPhaseA$.
}

Our receiver introduces $\numberOfClusters\RBsize+8\times\numberOfClusters\RBsize/3$ complex multiplications and $2/3\numberOfClusters\RBsize+8\times8(\numberOfClusters/2-1)$ complex summations for the user separation  and the hypothesis testing in addition to the complexity of the \ac{ML} decoder which is low for $\numberOfIterations=3$ as reported in \cite{Schmidt2005}.

\section{Numerical Analysis}\label{sec:numerical}
In this section, we evaluate the proposed modulation schemes numerically. We consider the interlace parameters in  \ac{NR-U}  for $15$~kHz. For the first scheme, we employ the sequences given in Table~\ref{table:seqC} and the spreading sequences  $\seqGa = (\constantOne,\constantOne,\constantOne,\constantMinusi,\constanti)$ and $\seqGb = (\constantOne,\constanti,\constantMinusOne,\constantOne,\constantMinusi)$. 
For comparison, we consider the sequences adopted in  \cite{nr_phy_2020} for \ac{PUCCH} Format 0 and 1 and introduce three \ac{PAPR} minimization techniques for interlaced transmission. The first two methods rely on the optimal phase rotation (i.e., \ac{PTS}) with the \ac{QPSK} alphabet for each \ac{PRB} for a given sequence, which prioritize either \ac{CM} or \ac{PAPR}.
The third approach is {\em cyclic-shift hopping} adopted in \ac{NR} for interlaced transmission with \ac{PUCCH} Format 0 \cite{nr_phy_2020}. The sequence on $\indexRB$th occupied \ac{PRB} in the interlace is multiplied with the sequence
$ {(\constante^{\indexRB\frac{2\pi{\constanti}}{\RBsize}\times0}, \constante^{\indexRB\frac{2\pi{\constanti}}{\RBsize}\times1}, \dots, \constante^{\indexRB\frac{2\pi{\constanti}}{\RBsize}\times(\RBsize-1)})}$ 
for $\indexRB=0,1,\mydots,9$. For the fourth design, we generate all possible \ac{ZC} sequences of length 113 (cyclically padded to 120) and select the best $30$ sequences based on the \ac{PAPR} of the corresponding signals after they are mapped to the interlace. For all schemes, we assume that 6 users transmit orthogonal sequences to indicate \ac{ACK}  or \ac{NACK} on the same interlace as discussed in Section~\ref{subsec:formatzeroTX}.
  
  For a larger \ac{UCI} payload, we consider the modulation scheme introduced in Section~\ref{subsec:morebitsAndusermultiplexing} and set $\seqGa=(\constantOne,\constanti,\constantOne)$, $\seqGb=(\constantOne,\constantOne,\constantMinusOne)$, $\seqGc = (\constantOne,\constantOne,\constantOne,\constantMinusi,\constanti)$, and $\seqGd = (\constantOne,\constanti,\constantMinusOne,\constantOne,\constantMinusi)$. {\color{\reviewColor} We compare the proposed scheme with two other approaches. The first approach uses  \ac{OCC}  on each \ac{PRB} \cite{ericssonNRUproposal}. For this scheme, we consider 10 \ac{QPSK} symbols for each user. To reduce the \ac{PAPR}, each \ac{QPSK} symbol is multiplied with a distinct column of a \ac{DFT} matrix of size 12 (i.e., the \ac{OCC})  and the resulting vectors are mapped to the \acp{PRB} of interlace. For the second approach, we consider the approach used in Format 3 in \ac{NR} for interlaced transmission, i.e., pre-DFT OCC. We first generate 30 $\pi/2$-\ac{BPSK} symbols for each user and expand it with  an \ac{OCC} sequence of length 4. After we calculate the \ac{DFT} of the spread sequence, the output is mapped to the interlace \cite{qualcommNRUproposal}.  For the sake of fair comparison, we consider the same spectral efficiency for all schemes and transmit 11 \ac{UCI} bits per user.  For the competing schemes, we use the (32,11) linear block code with the rate matching in \cite{nr_coding_2020}.  At the receiver side, we assume 2 antennas and the received signals are combined with \ac{MRC} and processed with \ac{MMSE} equalizer and \ac{ML} decoder. The receiver for the proposed scheme is given in Section~\ref{subsec:receiverCoherent}. }

\subsection{PAPR/CM Distribution}
In \figurename~\ref{fig:papr}, the \ac{PAPR} distributions for all aforementioned approaches are provided. For \ac{ACK}/\ac{NACK} indication,  the optimal phase rotations prioritizing  \ac{PAPR} and \ac{CM} for \ac{NR} sequences result in a maximum \ac{PAPR} of $5.3$ dB and $5.7$ dB, respectively, while the \ac{ZC} sequences limit the \ac{PAPR} to $6$ dB. The cyclic-shift hopping also reduces the maximum \ac{PAPR} to $6$ dB. The \ac{PAPR} for the schemes in \cite{ericssonNRUproposal} and \cite{qualcommNRUproposal} for 11 \ac{UCI} bits reach to 8.1~dB and 7.3~dB, respectively. The proposed schemes offer limit the \ac{PAPR} to 3 dB as they exploit \acp{CS}. The  \ac{PAPR} gains with the proposed schemes for up-to 2 \ac{UCI} bits and 11 \ac{UCI} bits are in the range of $2.7$-$3$ dB and $4.3$-$5.1$~dB, respectively.
\begin{figure}[t]
	\centering
	{\includegraphics[width =3.4in]{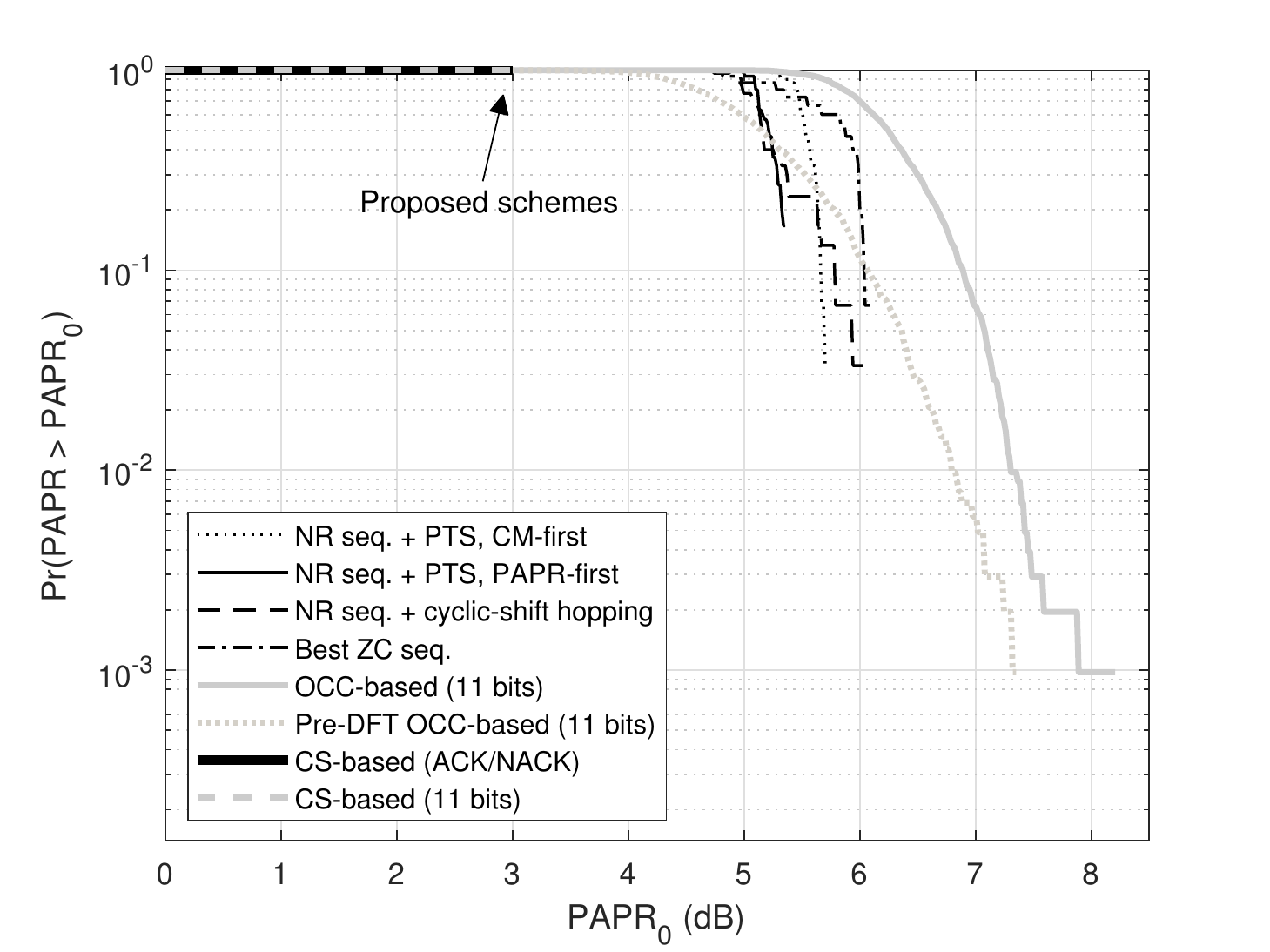}
	\vspace{-10pt}
	}
	\caption{PAPR distribution.}
	\label{fig:papr}
\end{figure} 

Another metric that characterizes the fluctuation of the resulting signal is the \ac{CM}. We calculate the \ac{CM} in dB  as
$\cubicMetric = {20 \log_{10}(\operatorRMS[{\vNormCubic}])} / {1.56}$,
where $\vNorm$ is the synthesized signal in time with the power of $1$ \cite{eval_NR}. In \figurename~\ref{fig:cm}, we compare the \ac{CM} distributions for the aforementioned schemes. Similar to the \ac{PAPR} results, the proposed schemes  improve the \ac{CM} within the range of $0.8$-$1.7$~dB over the schemes considered in this study.
\begin{figure}[t]
	\centering
	{\includegraphics[width =3.4in]{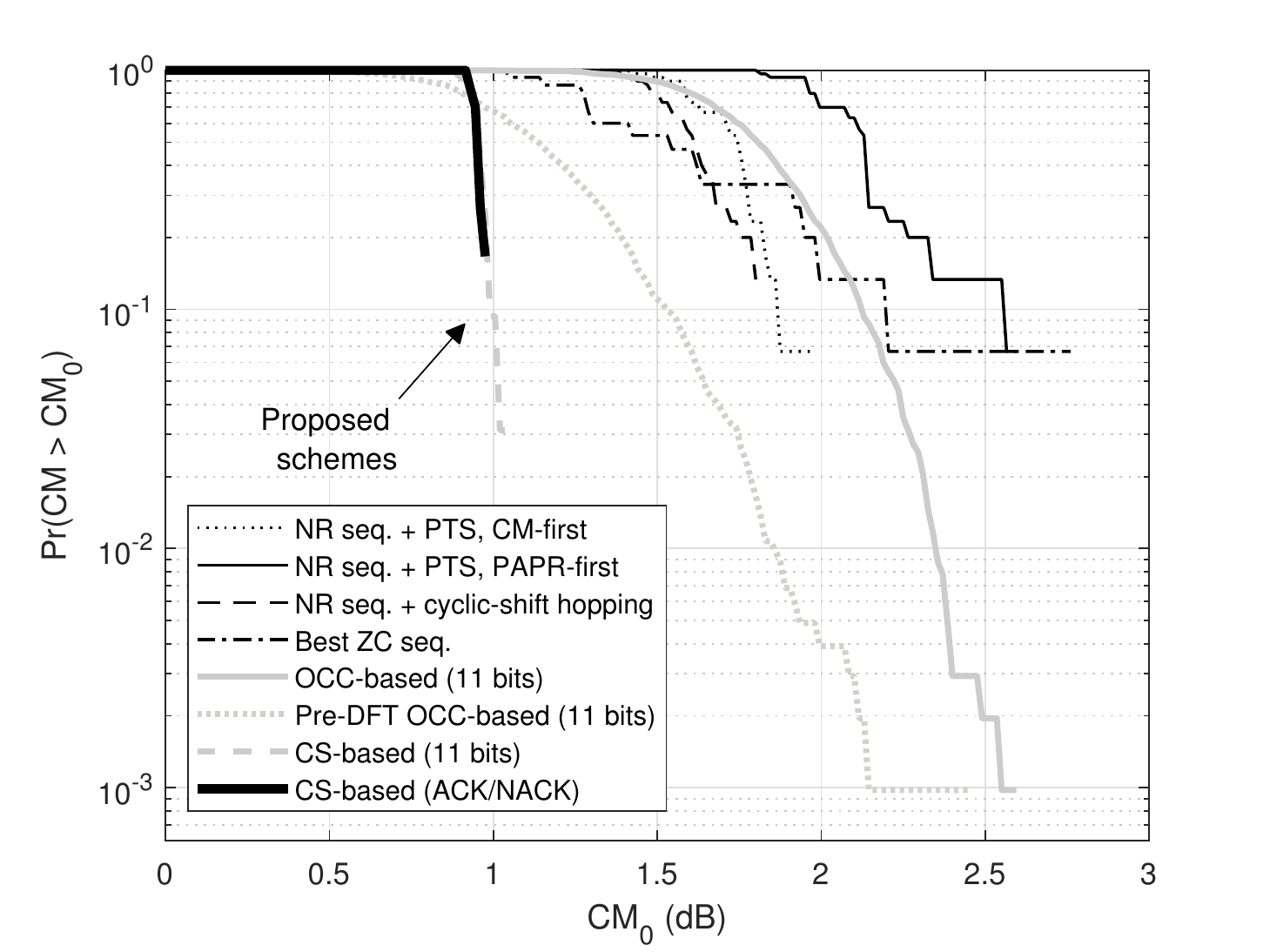}
	}
	\vspace{-10pt}
	\caption{CM distribution.}
	\label{fig:cm}
\end{figure}

\subsection{Peak Cross-correlation Distribution}
We evaluate the peak cross-correlation distribution of \acp{CS} designed for the first scheme by calculating
$    \peakCrossCorrelatation[] = {\operationMAX[|{\operationIDFT[\seqGx_\indexSequence\odot\seqGx_j^*][\dftSize]}|]} / {\RBsize},
$
where $\seqGx_\indexSequence$ is the $\indexSequence$th sequence in the set, $\indexSequence\neq j$ and $\operationIDFT[\cdot][\dftSize]$ is the \ac{DFT} operation of size $\dftSize$ \cite{eval_NR}. To achieve a large oversampling in time, we choose $\dftSize=4096$. In \figurename~\ref{fig:correlation}, we provide the distribution of $\peakCrossCorrelatation[]$ for different schemes. The \ac{ZC} sequences fail as the maximum peak cross-correlation reaches up to $0.95$ although 50 percentile performance is better than the other methods. The set of \ac{NR} sequences results in a maximum of $0.8$. On the other hand,
they are $0.715$  for the both sets $\setC$ and $\setD$  as we set $\thershold=0.715$. Hence, the proposed set is superior to the sequences adopted in \ac{NR} in terms of the maximum peak cross-correlation.
\begin{figure}[t]
	\centering
	{\includegraphics[width =3.4in]{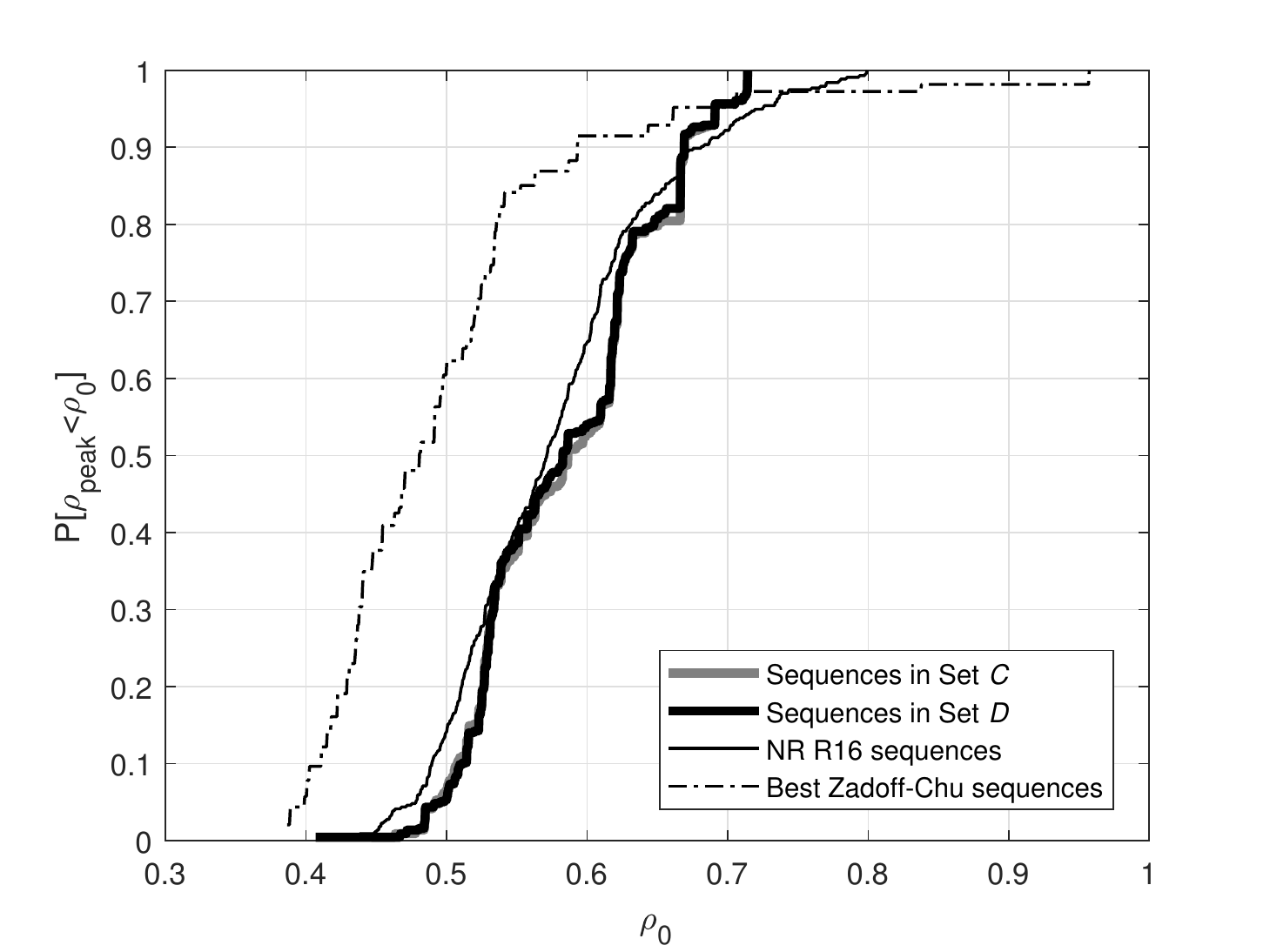}
	}
	\vspace{-10pt}
	\caption{Peak cross-correlation distribution.}
	\label{fig:correlation}
\end{figure}

\subsection{False Alarm and Miss-Detection Performance}
In this analysis, we demonstrate the impact of interlacing on the \ac{ACK}-to-\ac{NACK} rate and the \ac{ACK} miss-detection rate for a given \ac{DTX}-to-\ac{ACK} probability. The \ac{DTX}-to-\ac{ACK} and \ac{NACK}-to-\ac{ACK} rates correspond to the probability of  \ac{ACK} detection when there is no signal or a \ac{NACK} is being transmitted, respectively. The \ac{ACK} miss-detection rate is the probability of not detecting \ac{ACK} when \ac{ACK} is actually being transmitted. For the single-\ac{PRB} approach, we consider \ac{NR} \ac{PUCCH} Format~0 with interlaced transmission. 
To show the limits, we consider two extreme channel conditions where the occupied \acp{PRB} in an interlace experience the same  fading coefficients, i.e., flat fading, or \ac{i.i.d.} Rayleigh fading  to model selective fading. In practice, there is always correlation between channel coefficients. However, the correlation can decrease significantly for a large spacing between the occupied \acp{PRB} in an interlace. 

In the simulation, we set \ac{DTX}-to-\ac{ACK} probability to be $1\%$ at the detector based on Neyman-Pearson criterion and consider  \ac{MRC} of signals from 2 receive antennas. We provide curves based on \ac{SNR} per subcarrier as it reveals the benefit of interlace under the \ac{PSD} requirement in the unlicensed band as compared to single \ac{PRB} transmission. The results in \figurename~\ref{fig:errorSequencebased} show that the interlaced transmission improves the performance as compared to the single-RB approach due to the increased signal power under the \ac{PSD} requirement. When the channel is frequency-selective, the slopes of the \ac{NACK}-to-\ac{ACK} and \ac{ACK} miss-detection rates are much larger as the non-contiguous resource allocation exploits the  diversity due the frequency selectivity. Note that all the schemes for ACK/NACK indication provide the same error-rate performance. {\color{\reviewColor}However, the proposed scheme achieves it with low-\ac{PAPR} and \ac{CM}, which potentially increases the reliability in unlicensed channels for cell-edge users.}
\begin{figure}[t]
	\centering
	{\includegraphics[width =3.4in]{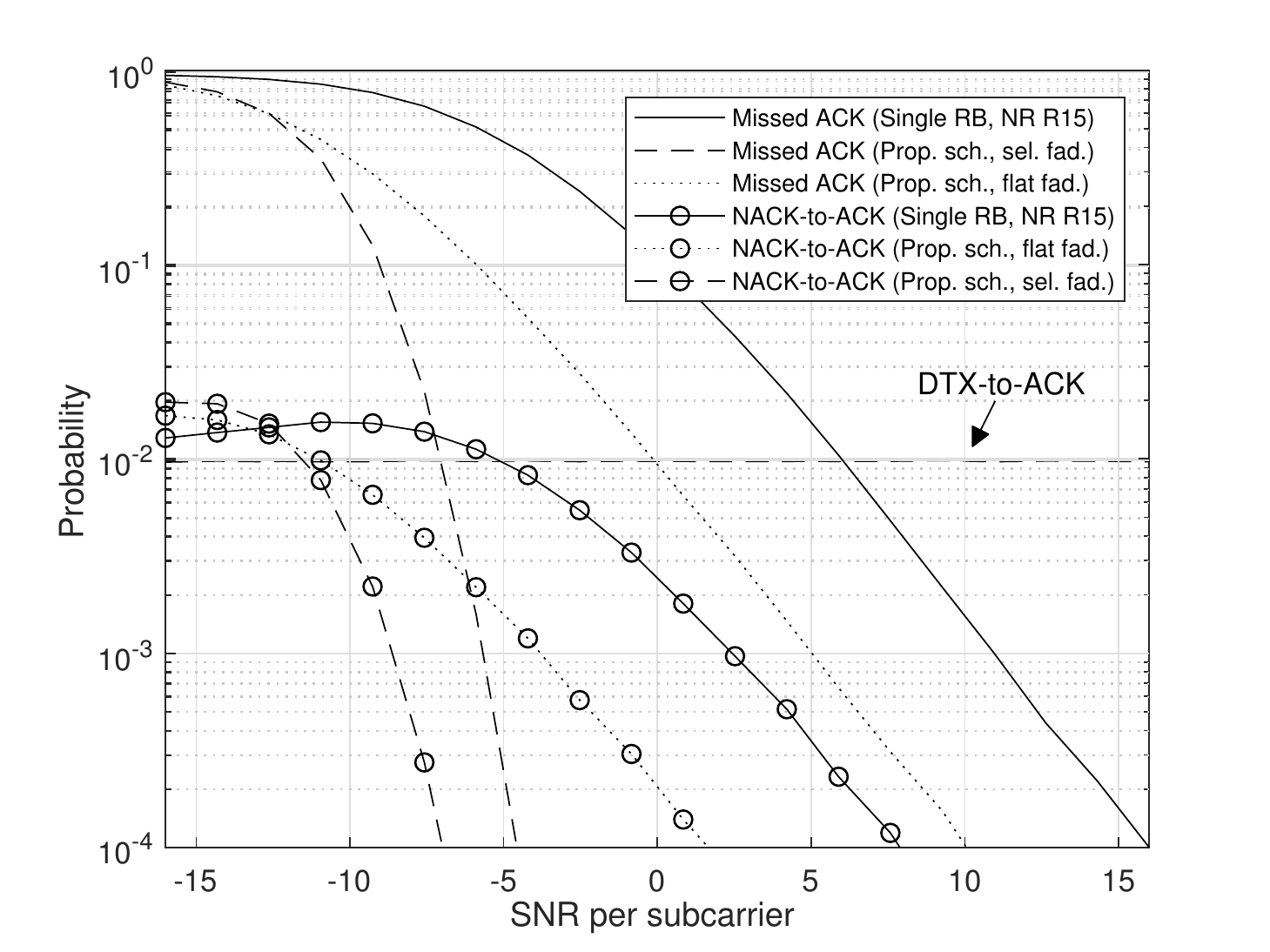}
	}
	\vspace{-10pt}
	\caption{The receiver performance for ACK/NACK transmission.}
	\label{fig:errorSequencebased}
\end{figure} 

{\color{\reviewColor}
\subsection{Error-rate Comparison}
In \figurename~\ref{fig:blerber}\subref{subfig:ber} and \figurename~\ref{fig:blerber}\subref{subfig:bler}, we compare the \ac{BER} and \ac{BLER} performance  of the second  proposed modulation scheme  and the methods in  \cite{qualcommNRUproposal} and \cite{ericssonNRUproposal} for moderate \ac{UCI} payload. While the minimum Euclidean distance between sequences are 11.3137 and 9.798 for the pre-DFT OCC-based and the OCC-based approaches, respectively, it is 10.9545 for the proposed scheme. Hence, for \ac{AWGN} channel, the pre-DFT \ac{OCC}-based scheme \cite{qualcommNRUproposal} offers approximately 0.3~dB and 1~dB gain as compared to proposed scheme and \ac{OCC}-based scheme, respectively. 
For flat-fading, the difference between the schemes is negligible at 0.01 \ac{BLER}.
On the other hand,  for selective channels,  the proposed approach is $2$~dB better than the \ac{OCC}-based scheme and similar to the one in \cite{qualcommNRUproposal}, respectively. The performance difference between the proposed method and \ac{OCC}-based scheme is because the receiver for the proposed scheme coherently combine the symbols on different \acp{PRB} with \ac{MRC}. The receiver for \ac{OCC}-based scheme cannot exploit the frequency selectivity as the data symbols are not spread to different \acp{PRB}. However, the pre-DFT \ac{OCC} spreads the information to different \acp{PRB} through \ac{DFT} operation and harnesses the selectivity better. However, it does not maintain the flatness in the frequency and \ac{FDE} slightly deteriorates its performance. 
}


\begin{figure}[t]
	\centering
	\subfloat[{\color{\reviewColor}{\ac{BER}}.}]{\includegraphics[width =3.4in]{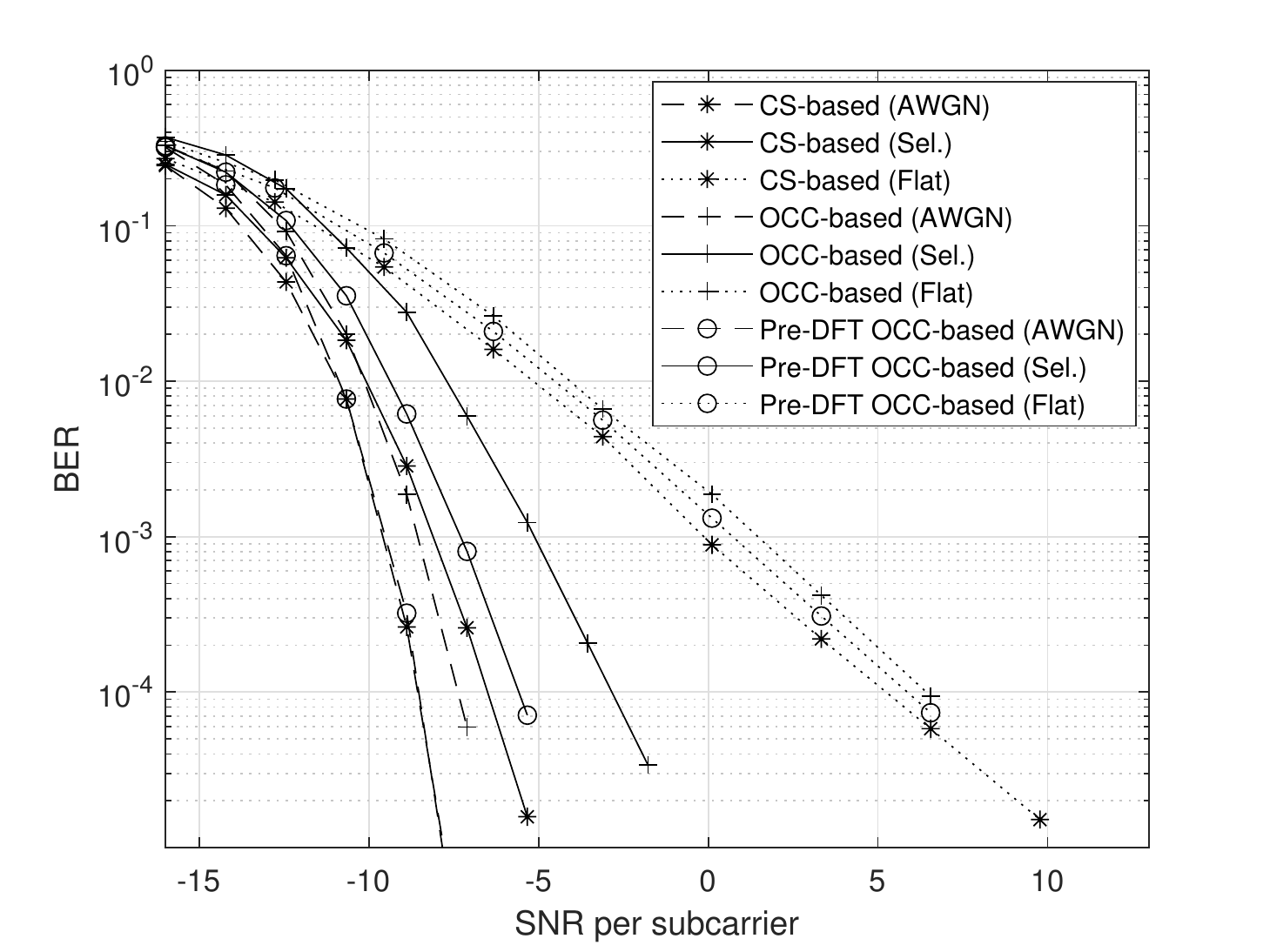}
		\label{subfig:ber}}
		\if\IEEEsubmission0
	  	\\
		\fi
	\subfloat[\ac{BLER}.]{\includegraphics[width =3.4in]{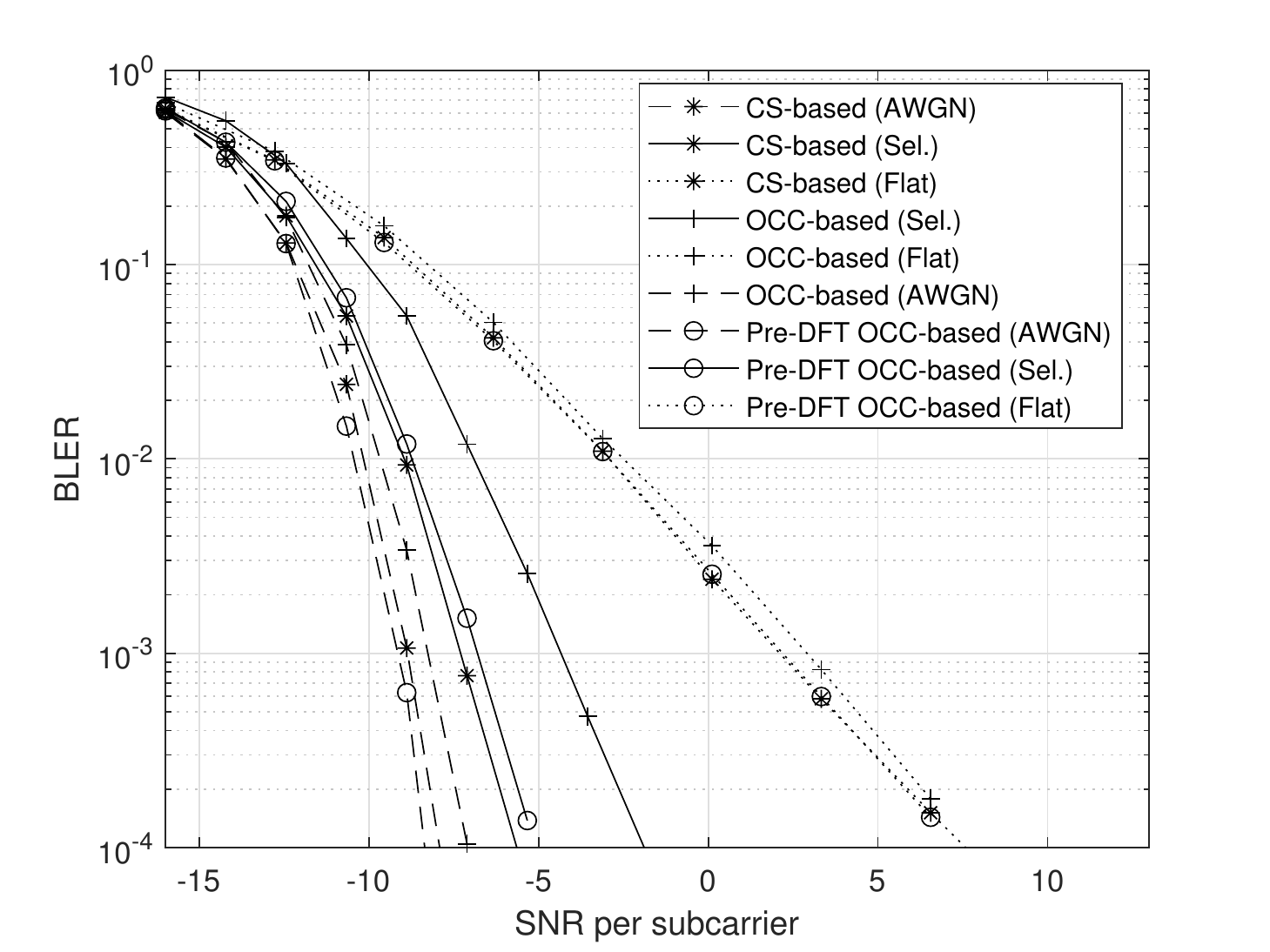}
		\label{subfig:bler}}
	\caption{{\color{\reviewColor}Error rate comparison of the proposed \ac{CS}-based scheme, \ac{OCC}-based scheme  \cite{ericssonNRUproposal}, and pre-\ac{DFT} \ac{OCC}-based schemes \cite{qualcommNRUproposal} for 11 \ac{UCI} bits.}}
	\label{fig:blerber}
\end{figure}

\section{Concluding Remarks}\label{sec:conclusion}
In this study, we propose two modulation schemes for \ac{UL} control channels  which consist of non-contiguous \acp{PRB} in the frequency by exploiting \acp{GCP} and introduce Theorem~\ref{th:golayIterative} and Theorem~\ref{th:reduced}. 
The main benefit of the proposed approaches is that they address the \ac{PAPR} problem of \ac{OFDM} signals while allowing a flexible non-contiguous resource allocation. For example, the number of null symbols between the \acp{PRB}  or the number of \acp{PRB} in an interlace can be chosen flexibly with minor modifications in both proposed schemes. In all cases, the \ac{PAPR} of the corresponding signal is less than or equal to $3$~dB. With comprehensive numerical analysis, we show that the \ac{PAPR} gains are in the range of $2.7$-$3$ dB and $4.3$-$5.1$~dB for the first scheme and the second scheme, respectively,  as compared other schemes considered in this study.

The first modulation scheme is similar to the \ac{NR} \ac{PUCCH} Format 0 for 1 or 2 \ac{UCI} bits. It separates the \ac{PAPR}  and  \ac{CCI} minimization problems by utilizing the properties of \acp{CS}. While the first challenge is solved by choosing the sequences for \acp{PRB} as a \ac{GCP} in the light of Theorem~\ref{th:golayIterative}, the second problem is addressed by design a set of \acp{GCP} with low peak cross-correlation with a search algorithm. {\color{\reviewColor}While our algorithm generate a set of \acp{GCP} better than the sequences in \ac{NR} in terms of maximum peak cross-correlation, a systematic solution for a \ac{GCP} set with low peak cross-correlation  is still an open problem.
For the second scheme, we develop a new theorem, i.e., Theorem~\ref{th:reduced}, which is capable of generating a wide-variety of \acp{CS} though multiple seed sequences. It can generate up to $\numberSequencesDifferentThanLEngthOne!\frac{(\numberOfIterations!)^2}{(\numberOfIterations-\numberSequencesNotColinear+1)!}\numberOfPointsForPSK^{\numberOfIterations+1}$ distinct \acp{CS}, which is a function of the  multiplicity and lengths of the seed \acp{GCP}.
We show that this joint coding-and-modulation scheme allows 3 users to transmit 11 bits on the same interlace while providing  4.3~dB \ac{PAPR} gain and similar \ac{BLER} performance as compared to the approach used in \ac{NR} \ac{PUCCH} Format 3. Hence, the proposed schemes can be beneficial for cell-edge users and complement the existing approaches in wireless standards.}

\appendices

\section{Proof of Theorem~\ref{th:reduced}}
\label{app:GCPrecursion}
\begin{proof}
	
	By using Theorem~\ref{th:golayIterative}, a recursion which generates a \ac{GCP} $(\seqGaIt[\numberOfIterations],\seqGbIt[\numberOfIterations])$ for $\numberOfIterations\ge 1$ can be given by
	\begin{align}
	\polySeq[{\seqGaIt[\indexIteration]}][\polyVariable] =& \polySeq[{\seqGcRecursion[{\permutationOrderEle[\indexIteration]}]}][\polyVariable]  \polySeq[{\seqGaIt[\indexIteration-1]}][\polyVariable]  
	+ \angleGolay[\indexIteration]\polySeq[{\seqGdRecursion[{\permutationOrderEle[\indexIteration]}]}][\polyVariable]
	\polySeq[{\seqGbIt[\indexIteration-1]}][\polyVariable]  \polyVariable^{\separationGolay[\indexIteration]}\polyVariableIt^{2^\permutationShift[{\indexIteration}]}~, \nonumber\\
	\polySeq[{\seqGbIt[\indexIteration]}][\polyVariable] =& \polySeq[{\seqGdTildeRecursion[{\permutationOrderEle[\indexIteration]}]}][\polyVariable]\polySeq[{\seqGaIt[\indexIteration-1]}][\polyVariable]  
	-\angleGolay[\indexIteration]\polySeq[{\seqGcTildeRecursion[{\permutationOrderEle[\indexIteration]}]}][\polyVariable]\polySeq[{\seqGbIt[\indexIteration-1]}][\polyVariable]  \polyVariable^{\separationGolay[\indexIteration]}\polyVariableIt^{2^\permutationShift[{\indexIteration}]}~, 
	\label{eq:iterationGolay}
	\end{align}
	where $\seqGaIt[0]=\seqGbIt[0]=1$, $(\seqGcRecursion[{\permutationOrderEle[\indexIteration]}],\seqGdRecursion[{\permutationOrderEle[\indexIteration]}])$ is the ${\permutationOrderEle[\indexIteration]}$th \ac{GCP} of length $\lengthGcGdIterative[{\permutationOrderEle[\indexIteration]}]\in\integersPositive$,  $\angleGolay[\indexIteration],\polyVariableIt\in \{u:u\in\complexNumbers, |u|=1\}$ are arbitrary complex numbers of unit magnitude, $\separationGolay[\indexIteration]\in  \integers$  for $\indexIteration=1,2,\mydots,\numberOfIterations$, and $\permutationShift[{\indexIteration}]$ is the $\indexIteration$th element of the sequence $\seqPermutationShift=(\permutationShift[{\indexIteration}])_{i=1}^{\numberOfIterations}$ defined by the permutation of $\{0,1,\dots,\numberOfIterations-1\}$. 
	The recursion in	\eqref{eq:iterationGolay} can be re-expressed as
	\if \IEEEsubmission0
		\begin{align}
		\functionf[\indexIteration]  =& 
		\operatorOrderC[1][{\indexIteration}]
		\operatorOrderD[0][{\indexIteration}]
		\operatorOrderCtilde[0][{\indexIteration}]
		\operatorOrderDtilde[0][{\indexIteration}]
		\operatorSeparation[0][{\indexIteration}]
		\operatorAngleScaleB[0][{\indexIteration}]
		\operatorSign[0][{\indexIteration}]
		(\functionf[\indexIteration-1] )\nonumber\\
		&+ 
		\operatorOrderC[0][{\indexIteration}]
		\operatorOrderD[1][{\indexIteration}]
		\operatorOrderCtilde[0][{\indexIteration}]
		\operatorOrderDtilde[0][{\indexIteration}]
		\operatorSeparation[1][{\indexIteration}]
		\operatorAngleScaleB[1][{\indexIteration}]
		\operatorSign[0][{\indexIteration}]
		(\functiong[\indexIteration-1] )  
		\polyVariableIt^{2^\permutationShift[{\indexIteration}]}~, \nonumber
		\\
		\functiong[\indexIteration]  =& 
		\operatorOrderC[0][{\indexIteration}]
		\operatorOrderD[0][{\indexIteration}]
		\operatorOrderCtilde[0][{\indexIteration}]
		\operatorOrderDtilde[1][{\indexIteration}]
		\operatorSeparation[0][{\indexIteration}]
		\operatorAngleScaleB[0][{\indexIteration}]
		\operatorSign[0][{\indexIteration}]
		(\functionf[\indexIteration-1] )\nonumber\\
		&+ 
		\operatorOrderC[0][{\indexIteration}]
		\operatorOrderD[0][{\indexIteration}]
		\operatorOrderCtilde[1][{\indexIteration}]
		\operatorOrderDtilde[0][{\indexIteration}]
		\operatorSeparation[1][{\indexIteration}]
		\operatorAngleScaleB[1][{\indexIteration}]
		\operatorSign[1][{\indexIteration}]
		(\functiong[\indexIteration-1] )  
		\polyVariableIt^{2^\permutationShift[{\indexIteration}]}~, 
		\label{eq:iterationGolayOperators}
		\end{align}
	\else
		\begin{align}
		\functionf[\indexIteration]  =& 
		\operatorOrderC[1][{\indexIteration}]
		\operatorOrderD[0][{\indexIteration}]
		\operatorOrderCtilde[0][{\indexIteration}]
		\operatorOrderDtilde[0][{\indexIteration}]
		\operatorSeparation[0][{\indexIteration}]
		\operatorAngleScaleB[0][{\indexIteration}]
		\operatorSign[0][{\indexIteration}]
		(\functionf[\indexIteration-1] )+ 
		\operatorOrderC[0][{\indexIteration}]
		\operatorOrderD[1][{\indexIteration}]
		\operatorOrderCtilde[0][{\indexIteration}]
		\operatorOrderDtilde[0][{\indexIteration}]
		\operatorSeparation[1][{\indexIteration}]
		\operatorAngleScaleB[1][{\indexIteration}]
		\operatorSign[0][{\indexIteration}]
		(\functiong[\indexIteration-1] )  
		\polyVariableIt^{2^\permutationShift[{\indexIteration}]}~, \nonumber
		\\
		\functiong[\indexIteration]  =& 
		\operatorOrderC[0][{\indexIteration}]
		\operatorOrderD[0][{\indexIteration}]
		\operatorOrderCtilde[0][{\indexIteration}]
		\operatorOrderDtilde[1][{\indexIteration}]
		\operatorSeparation[0][{\indexIteration}]
		\operatorAngleScaleB[0][{\indexIteration}]
		\operatorSign[0][{\indexIteration}]
		(\functionf[\indexIteration-1] )+ 
		\operatorOrderC[0][{\indexIteration}]
		\operatorOrderD[0][{\indexIteration}]
		\operatorOrderCtilde[1][{\indexIteration}]
		\operatorOrderDtilde[0][{\indexIteration}]
		\operatorSeparation[1][{\indexIteration}]
		\operatorAngleScaleB[1][{\indexIteration}]
		\operatorSign[1][{\indexIteration}]
		(\functiong[\indexIteration-1] )  
		\polyVariableIt^{2^\permutationShift[{\indexIteration}]}~, 
		\label{eq:iterationGolayOperators}
		\end{align}
	\fi
	where $\functionf[0]=\functiong[0]=1$, the operators $\operatorOrderC[0][{\indexIteration}](\functionh)$, $\operatorOrderD[0][{\indexIteration}](\functionh)$, $\operatorOrderCtilde[0][{\indexIteration}](\functionh)$, $\operatorOrderDtilde[0][{\indexIteration}](\functionh)$,
	$\operatorSeparation[0][{\indexIteration}](\functionh)$
	$\operatorSign[0][{\indexIteration}](\functionh)$, $\operatorAngleScaleB[0][{\indexIteration}](\functionh) $ are equal to $\functionh$, and
	the operators
	$\operatorOrderC[1][{\indexIteration}](\functionh)$,
	$\operatorOrderD[1][{\indexIteration}](\functionh)$, 
	$\operatorOrderCtilde[1][{\indexIteration}](\functionh)$, 
	$\operatorOrderDtilde[1][{\indexIteration}](\functionh)$, 
	$\operatorSeparation[1][{\indexIteration}](\functionh)$,
	$\operatorAngleScaleB[0][{\indexIteration}](\functionh) $, 
	$\operatorSign[0][{\indexIteration}](\functionh)$ 	
	are set to  
	$\polySeq[{\seqGcRecursion[{\permutationOrderEle[\indexIteration]}]}][\polyVariable] \functionh $,
	$\polySeq[{\seqGdRecursion[{\permutationOrderEle[\indexIteration]}]}][\polyVariable] \functionh $,
	$\polySeq[{\seqGcTildeRecursion[{\permutationOrderEle[\indexIteration]}]}][\polyVariable] \functionh $,
	$\polySeq[{\seqGdTildeRecursion[{\permutationOrderEle[\indexIteration]}]}][\polyVariable] \functionh $, and
	$\polyVariable^{\separationGolay[\indexIteration]}  \functionh $,
	$\exponentialBase^{\constantj\angleexp[\indexIteration]} \functionh $,  and $\exponentialBase^ {\constantj\frac{\numberOfPointsForPSK}{2}}\functionh $,
	 respectively. 	
	 
	  By utilizing an approach that represents the outcome of a recursion concisely \cite{Sahin_2018} (summarized in Appendix~\ref{app:Algebraic} for the sake of completeness) and investigating the position of the operators in \eqref{eq:iterationGolayOperators}, we obtain 	 
	 the configuration vectors, i.e.,  $\vecArrangement[\indexIteration]^{\rm T}$  for $\indexIteration=1,2,\mydots,\numberOfIterations$, for $\operatorOrderC[0,1][{\indexIteration}](\functionh)$,
	$\operatorOrderD[0,1][{\indexIteration}](\functionh)$, 
	$\operatorOrderCtilde[0,1][{\indexIteration}](\functionh)$, and
	$\operatorOrderDtilde[0,1][{\indexIteration}](\functionh)$  as  $[1~0~0~0]$, $[0~1~0~0]$, $[0~0~0~1]$, and $[0~0~1~0]$, respectively,  . Therefore, by plugging the configuration vectors into \eqref{eq:closedformANFsF} and \eqref{eq:closedformANFsG}, the Boolean functions associated with the construction sequences (or indication sequences) for $\operatorOrderC[0,1][{\indexIteration}](\functionh)$ are obtained as
	\begin{align}
	\funcGfForC[\indexIteration](\seqx) &=  
	\begin{cases}
	\displaystyle
	(1 - \monomial[{\permutationMono[{\indexIteration}]}])
	& \indexIteration=\numberOfIterations \nonumber\\
	\displaystyle
	(1- \monomial[{\permutationMono[{\indexIteration}]}])(1 - \monomial[{\permutationMono[{\indexIteration+1}]}]) 
	& \indexIteration<\numberOfIterations \nonumber\\
	\end{cases}~,
	\end{align}
	\begin{align}
	\funcGgForC[\indexIteration](\seqx) &= 
	\begin{cases}
	\displaystyle
	0
	& \indexIteration=\numberOfIterations \nonumber\\
	\displaystyle
	(1- \monomial[{\permutationMono[{\indexIteration}]}])(1 - \monomial[{\permutationMono[{\indexIteration+1}]}])
	& \indexIteration<\numberOfIterations \nonumber\\
	\end{cases}~,
	\end{align}
	respectively. Similarly, for $\operatorOrderD[0,1][{\indexIteration}](\functionh)$,
	\begin{align}
	\funcGfForD[\indexIteration](\seqx) &=  
	\begin{cases}
	\displaystyle
	\monomial[{\permutationMono[{\indexIteration}]}]
	& \indexIteration=\numberOfIterations \nonumber\\
	\displaystyle
	{\monomial[{\permutationMono[{\indexIteration}]}]}(1 - \monomial[{\permutationMono[{\indexIteration+1}]}])  & \indexIteration<\numberOfIterations \nonumber\\
	\end{cases}~,
	\end{align}
	\begin{align}
	\funcGgForD[\indexIteration](\seqx) &= 
	\begin{cases}
	\displaystyle
	0 & \indexIteration=\numberOfIterations \nonumber\\
	\displaystyle
	{\monomial[{\permutationMono[{\indexIteration}]}]}(1 - \monomial[{\permutationMono[{\indexIteration+1}]}]) & \indexIteration<\numberOfIterations \nonumber\\
	\end{cases}~.
	\end{align}
	For $\operatorOrderDtilde[0,1][{\indexIteration}](\functionh)$,
	\begin{align}
	\funcGfForDtilde[\indexIteration](\seqx) &=  
	\begin{cases}
	\displaystyle
	0
	& \indexIteration=\numberOfIterations \nonumber\\
	\displaystyle
	(1 - \monomial[{\permutationMono[{\indexIteration}]}])\monomial[{\permutationMono[{\indexIteration+1}]}] & \indexIteration<\numberOfIterations \nonumber\\
	\end{cases}~,
	\end{align}
	\begin{align}
	\funcGgForDtilde[\indexIteration](\seqx) &= 
	\begin{cases}
	\displaystyle
	(1 - \monomial[{\permutationMono[{\indexIteration}]}]) 
	& \indexIteration=\numberOfIterations \nonumber\\
	\displaystyle
	(1 - \monomial[{\permutationMono[{\indexIteration}]}])\monomial[{\permutationMono[{\indexIteration+1}]}] 
	& \indexIteration<\numberOfIterations \nonumber\\
	\end{cases}~.
	\end{align}
	For $\operatorOrderCtilde[0,1][{\indexIteration}](\functionh)$,
	\begin{align}
	\funcGfForCtilde[\indexIteration](\seqx) &=  
	\begin{cases}
	\displaystyle
	0
	& \indexIteration=\numberOfIterations \nonumber\\
	\displaystyle
	{\monomial[{\permutationMono[{\indexIteration}]}]}\monomial[{\permutationMono[{\indexIteration+1}]}]
	& \indexIteration<\numberOfIterations \nonumber\\
	\end{cases}~,
	\end{align}
	\begin{align}
	\funcGgForCtilde[\indexIteration](\seqx) &= 
	\begin{cases}
	\displaystyle
	\monomial[{\permutationMono[{\indexIteration}]}]
	& \indexIteration=\numberOfIterations \nonumber\\
	\displaystyle
	{\monomial[{\permutationMono[{\indexIteration}]}]}\monomial[{\permutationMono[{\indexIteration+1}]}]
	& \indexIteration<\numberOfIterations \nonumber\\
	\end{cases}~.
	\end{align}
	Therefore, the combined effects of the operators $\operatorOrderC[0,1][{\indexIteration}](\functionh)$,
	$\operatorOrderD[0,1][{\indexIteration}](\functionh)$, 
	$\operatorOrderCtilde[0,1][{\indexIteration}](\functionh)$, and
	$\operatorOrderDtilde[0,1][{\indexIteration}](\functionh)$ on the coefficients of $\polyVariableIt^\varMonomial$ of $\functionf[\indexIteration]$ and $\functiong[\indexIteration]$ can calculated as
\if \IEEEsubmission0
\begin{align}
\compositeOperatorFcd[\varMonomial][\functionh]=&\functionh\prod_{\indexIteration=1}^{\numberOfIterations}
\polySeq[{\seqGcRecursion[{\permutationOrderEle[\indexIteration]}]}][\polyVariable] ^{\funcGfForC[\indexIteration](\seqx)}
\polySeq[{\seqGcTildeRecursion[{\permutationOrderEle[\indexIteration]}]}][\polyVariable] ^{\funcGfForCtilde[\indexIteration](\seqx)}\nonumber
\\&~~~~~~~~~~~
\times\polySeq[{\seqGdRecursion[{\permutationOrderEle[\indexIteration]}]}][\polyVariable] ^{\funcGfForD[\indexIteration](\seqx)}
\polySeq[{\seqGdTildeRecursion[{\permutationOrderEle[\indexIteration]}]}][\polyVariable] ^{\funcGfForDtilde[\indexIteration](\seqx)}\nonumber
\\\stackrel{(a)}{=}&
\functionh\prod_{\indexIteration=1}^{\numberOfIterations}
\polySeq[{\seqGcRecursion[{\permutationOrderEle[\indexIteration]}]}][\polyVariable] {\funcGfForC[\indexIteration](\seqx)}+
\polySeq[{\seqGcTildeRecursion[{\permutationOrderEle[\indexIteration]}]}][\polyVariable] {\funcGfForCtilde[\indexIteration](\seqx)}
\nonumber
\\&~~~~~~~~~~~
+\polySeq[{\seqGdRecursion[{\permutationOrderEle[\indexIteration]}]}][\polyVariable] {\funcGfForD[\indexIteration](\seqx)}+
\polySeq[{\seqGdTildeRecursion[{\permutationOrderEle[\indexIteration]}]}][\polyVariable] {\funcGfForDtilde[\indexIteration](\seqx)}\nonumber
\\=&
\functionh\times\funcfForCommonOrder(\seqx,\polyVariable)\nonumber~.
\end{align}
\else
\begin{align}
\compositeOperatorFcd[\varMonomial][\functionh]=&\functionh\prod_{\indexIteration=1}^{\numberOfIterations}
\polySeq[{\seqGcRecursion[{\permutationOrderEle[\indexIteration]}]}][\polyVariable] ^{\funcGfForC[\indexIteration](\seqx)}
\polySeq[{\seqGcTildeRecursion[{\permutationOrderEle[\indexIteration]}]}][\polyVariable] ^{\funcGfForCtilde[\indexIteration](\seqx)}\polySeq[{\seqGdRecursion[{\permutationOrderEle[\indexIteration]}]}][\polyVariable] ^{\funcGfForD[\indexIteration](\seqx)}
\polySeq[{\seqGdTildeRecursion[{\permutationOrderEle[\indexIteration]}]}][\polyVariable] ^{\funcGfForDtilde[\indexIteration](\seqx)}\nonumber
\\\stackrel{(a)}{=}&
\functionh\prod_{\indexIteration=1}^{\numberOfIterations}
\polySeq[{\seqGcRecursion[{\permutationOrderEle[\indexIteration]}]}][\polyVariable] {\funcGfForC[\indexIteration](\seqx)}+
\polySeq[{\seqGcTildeRecursion[{\permutationOrderEle[\indexIteration]}]}][\polyVariable] {\funcGfForCtilde[\indexIteration](\seqx)}
\nonumber
+\polySeq[{\seqGdRecursion[{\permutationOrderEle[\indexIteration]}]}][\polyVariable] {\funcGfForD[\indexIteration](\seqx)}+
\polySeq[{\seqGdTildeRecursion[{\permutationOrderEle[\indexIteration]}]}][\polyVariable] {\funcGfForDtilde[\indexIteration](\seqx)}\nonumber
\\=&
\functionh\times\funcfForCommonOrder(\seqx,\polyVariable)\nonumber~.
\end{align}
\fi
and
\if \IEEEsubmission0
\begin{align}
\compositeOperatorGcd[\varMonomial][\functionh]=&\functionh\prod_{\indexIteration=1}^{\numberOfIterations} \polySeq[{\seqGcRecursion[{\permutationOrderEle[\indexIteration]}]}][\polyVariable] ^{\funcGgForC[\indexIteration](\seqx)} \polySeq[{\seqGcTildeRecursion[{\permutationOrderEle[\indexIteration]}]}][\polyVariable] ^{\funcGgForCtilde[\indexIteration](\seqx)} \nonumber \\
												&~~~~~~~~~~~\times \polySeq[{\seqGdRecursion[{\permutationOrderEle[\indexIteration]}]}][\polyVariable] ^{\funcGgForD[\indexIteration](\seqx)} \polySeq[{\seqGdTildeRecursion[{\permutationOrderEle[\indexIteration]}]}][\polyVariable] ^{\funcGgForDtilde[\indexIteration](\seqx)}\nonumber \\
							   \stackrel{(b)}{=}&\functionh\prod_{\indexIteration=1}^{\numberOfIterations} \polySeq[{\seqGcRecursion[{\permutationOrderEle[\indexIteration]}]}][\polyVariable] {\funcGgForC[\indexIteration](\seqx)}+\polySeq[{\seqGcTildeRecursion[{\permutationOrderEle[\indexIteration]}]}][\polyVariable] {\funcGgForCtilde[\indexIteration](\seqx)}\nonumber\\
                                                &~~~~~~~~~~~ \polySeq[{\seqGdRecursion[{\permutationOrderEle[\indexIteration]}]}][\polyVariable] {\funcGgForD[\indexIteration](\seqx)}+\polySeq[{\seqGdTildeRecursion[{\permutationOrderEle[\indexIteration]}]}][\polyVariable] {\funcGgForDtilde[\indexIteration](\seqx)}\nonumber\\
                                               =&\functionh\times\funcgForCommonOrder(\seqx,\polyVariable)\nonumber~,
\end{align}
\else
\begin{align}
\compositeOperatorGcd[\varMonomial][\functionh]=&\functionh\prod_{\indexIteration=1}^{\numberOfIterations} \polySeq[{\seqGcRecursion[{\permutationOrderEle[\indexIteration]}]}][\polyVariable] ^{\funcGgForC[\indexIteration](\seqx)} \polySeq[{\seqGcTildeRecursion[{\permutationOrderEle[\indexIteration]}]}][\polyVariable] ^{\funcGgForCtilde[\indexIteration](\seqx)} 
												\polySeq[{\seqGdRecursion[{\permutationOrderEle[\indexIteration]}]}][\polyVariable] ^{\funcGgForD[\indexIteration](\seqx)} \polySeq[{\seqGdTildeRecursion[{\permutationOrderEle[\indexIteration]}]}][\polyVariable] ^{\funcGgForDtilde[\indexIteration](\seqx)}\nonumber \\
						       \stackrel{(b)}{=}&\functionh\prod_{\indexIteration=1}^{\numberOfIterations} \polySeq[{\seqGcRecursion[{\permutationOrderEle[\indexIteration]}]}][\polyVariable] {\funcGgForC[\indexIteration](\seqx)}+\polySeq[{\seqGcTildeRecursion[{\permutationOrderEle[\indexIteration]}]}][\polyVariable] {\funcGgForCtilde[\indexIteration](\seqx)} + \polySeq[{\seqGdRecursion[{\permutationOrderEle[\indexIteration]}]}][\polyVariable] {\funcGgForD[\indexIteration](\seqx)}+\polySeq[{\seqGdTildeRecursion[{\permutationOrderEle[\indexIteration]}]}][\polyVariable] {\funcGgForDtilde[\indexIteration](\seqx)}\nonumber\\
                                               =&\functionh\funcgForCommonOrder(\seqx,\polyVariable)\nonumber~,
\end{align}
\fi
%
respectively, where (a) ((b)) is because only one of the functions among
${\funcGfForC[\indexIteration](\seqx)}$,
${\funcGfForCtilde[\indexIteration](\seqx)}$,
${\funcGfForD[\indexIteration](\seqx)}$, and
${\funcGfForDtilde[\indexIteration](\seqx)}$
 (${\funcGgForC[\indexIteration](\seqx)}$,
${\funcGgForCtilde[\indexIteration](\seqx)}$,
${\funcGgForD[\indexIteration](\seqx)}$, and
${\funcGgForDtilde[\indexIteration](\seqx)}$) is 1 while the others are equal to 0.

By defining $\angleGolay[\indexIteration]\triangleq\exponentialBase^{\constantj\angleexp[\indexIteration]}$ and using the identity $\exponentialBase^{\constantj\frac{\numberOfPointsForPSK}{2}}=-1$, the coefficients of $\polyVariableIt^\varMonomial$ of $\functionf[\indexIteration]$ and $\functiong[\indexIteration]$ due to the operators $\operatorSign[0,1][{\indexIteration}](\functionh)$, $\operatorAngleScaleB[0,1][{\indexIteration}](\functionh)$, and $\operatorSeparation[0,1][{\indexIteration}](\functionh)$ are  obtained in \cite{Sahin_2018}. Their compositions can be expressed as
$
\compositeOperatorFangle[\varMonomial][\functionh]=\compositeOperatorGangle[\varMonomial][\functionh] =\functionh\exponentialBase^{\constantj	\funcfForCommonPhase(\seqx)}\times \polyVariable^{\funcfForCommonShift(\seqx)}.
$
Finally, by composing $\compositeOperatorFcd[\varMonomial][\functionh]$ and  $\compositeOperatorFangle[\varMonomial][\functionh]$, and
$\compositeOperatorGcd[\varMonomial][\functionh]$
and $\compositeOperatorGangle[\varMonomial][\functionh]$, $\polySeq[{\seqGaIt[\numberOfIterations]}][\polyVariable]  $ and $\polySeq[{\seqGbIt[\numberOfIterations]}][\polyVariable]$ can be calculated as
\begin{align}
\polySeq[{\seqGaIt[\numberOfIterations]}][\polyVariable] =&\functionf[\numberOfIterations] \nonumber \if\IEEEsubmission0
\\
\fi
 =& \sum_{\varMonomial=0}^{2^\numberOfIterations-1} \compositeOperatorFcd[\varMonomial][{\compositeOperatorFangle[\varMonomial][{\functionh}]}]\polyVariableIt^\varMonomial\big|_{\functionh=\functionf[0]=\functiong[0]=1}\nonumber \if\IEEEsubmission0
 \\
 \fi
=&
\sum_{\varMonomial=0}^{2^\numberOfIterations-1} \funcfForCommonOrder(\seqx,\polyVariable)\times 
\exponentialBase^{ \constantj \funcfForCommonPhase(\seqx)}\times \polyVariable^{\funcfForCommonShift(\seqx)}\times \polyVariableIt^\varMonomial \nonumber
\end{align}
and
\begin{align}
\polySeq[{\seqGbIt[\numberOfIterations]}][\polyVariable] =&\functiong[\numberOfIterations] \nonumber  \if\IEEEsubmission0
\\
\fi
 =& \sum_{\varMonomial=0}^{2^\numberOfIterations-1} \compositeOperatorGcd[\varMonomial][{\compositeOperatorGangle[\varMonomial][{\functionh}]}]\polyVariableIt^\varMonomial\big|_{\functionh=\functionf[0]=\functiong[0]=1}\nonumber
 \if\IEEEsubmission0
 \\
 \fi
=&
\sum_{\varMonomial=0}^{2^\numberOfIterations-1} \funcgForCommonOrder(\seqx,\polyVariable)\times 
\exponentialBase^{ \constantj \funcfForCommonPhase(\seqx)}\times \polyVariable^{\funcfForCommonShift(\seqx)}\times \polyVariableIt^\varMonomial~, \nonumber
\end{align}
respectively, where $\polyVariableIt$ can be  chosen arbitrarily as $\polyVariableIt=\polyVariable^\varUpsample$.

The sequences $\seqGaIt[\numberOfIterations]$ and $\seqGbIt[\numberOfIterations]$ construct a \ac{GCP} based on \eqref{eq:iterationGolay}. Since the phase rotation does not change the \ac{APAC} of a sequence,  $\seqGaIt[\numberOfIterations]\times\exponentialBase^{\constantj\arbitraryPhaseA}$ and $\seqGbIt[\numberOfIterations]\times\exponentialBase^{\constantj\arbitraryPhaseB}$  also construct a \ac{GCP}.	
	
\end{proof}

\section{Representation of a Recursion}
\label{app:Algebraic}
For $\indexIteration=1,2,\dots,\numberOfIterations$, consider a recursion given by
\begin{align}
\functionf[\indexIteration] = \operator[11][{\indexIteration}]({\functionf[\indexIteration-1]}) + \operator[12][{\indexIteration}]({\functiong[\indexIteration-1]})\polyVariableIt^{2^\permutationShift[{\indexIteration}]}~, \nonumber\\
\functiong[\indexIteration] = \operator[21][{\indexIteration}]({\functionf[\indexIteration-1]}) + \operator[22][{\indexIteration}]({\functiong[\indexIteration-1]})\polyVariableIt^{2^\permutationShift[{\indexIteration}]}~,
\label{eq:iterationBasic}
\end{align}
where $\polyVariableIt$ is an arbitrary complex number, $\permutationShift[{\indexIteration}]$ is the $\indexIteration$th element of the sequence $\seqPermutationShift\triangleq(\permutationShift[{\indexIteration}])_{i=1}^{\numberOfIterations}$ defined by the permutation of $\{0,1,\dots,\numberOfIterations-1\}$,  $\operator[ij][{\indexIteration}]\in\{\operatorBinary[0][\indexIteration],\operatorBinary[1][\indexIteration]\}$  is a linear operator which transforms one function to another function in $\functionSpace$, and $\functionf[0]=\functiong[0]=\functionh$. In  \cite{Sahin_2018}, it was shown that $\functionf[\numberOfIterations]$ and $\functiong[\numberOfIterations]$ can be obtained as 
\begin{align}
\functionf[\numberOfIterations] = \sum_{\varMonomial=0}^{2^\numberOfIterations-1} \overbrace{\operatorBinary[{\funcGfForANF[\numberOfIterations](\seqx)}][\numberOfIterations]\dots\operatorBinary[\funcGfForANF[\indexIteration](\seqx)][\indexIteration]\dots\operatorBinary[\funcGfForANF[2](\seqx)][2]\operatorBinary[\funcGfForANF[1](\seqx)][1](\functionh)}^{\compositeOperatorF[\varMonomial][\functionh] }\polyVariableIt^\varMonomial ~,
\label{eq:finalSequenceFBasic}
\end{align}
and
\begin{align}
\functiong[\numberOfIterations] = \sum_{\varMonomial=0}^{2^\numberOfIterations-1} \overbrace{ \operatorBinary[\funcGgForANF[\numberOfIterations](\seqx)][\numberOfIterations]\dots\operatorBinary[\funcGgForANF[\indexIteration](\seqx)][\indexIteration]\dots\operatorBinary[\funcGgForANF[2](\seqx)][2]\operatorBinary[\funcGgForANF[1](\seqx)][1](\functionh)}^{\compositeOperatorG[\varMonomial][\functionh]  }\polyVariableIt^\varMonomial~,
\end{align}
where  
\begin{align}
\funcGfForANF[\indexIteration](\seqx) &=  
\begin{cases}
\displaystyle
\binaryAsignment[11][{\indexIteration}] (1 - \monomial[{\permutationMono[{\indexIteration}]}]) + 
\binaryAsignment[12][{\indexIteration}]\monomial[{\permutationMono[{\indexIteration}]}]
& \indexIteration=\numberOfIterations \\
\displaystyle
\binaryAsignment[11][{\indexIteration}](1- \monomial[{\permutationMono[{\indexIteration}]}])(1 - \monomial[{\permutationMono[{\indexIteration+1}]}]) 
\if\IEEEsubmission0
\\~~~~+ 
\else
+
\fi
\binaryAsignment[12][{\indexIteration}]{\monomial[{\permutationMono[{\indexIteration}]}]}(1 - \monomial[{\permutationMono[{\indexIteration+1}]}]) \\~~~~+
\binaryAsignment[21][{\indexIteration}](1 - \monomial[{\permutationMono[{\indexIteration}]}])\monomial[{\permutationMono[{\indexIteration+1}]}] 
\if\IEEEsubmission0
\\~~~~+ 
\else
+
\fi
\binaryAsignment[22][{\indexIteration}]{\monomial[{\permutationMono[{\indexIteration}]}]}\monomial[{\permutationMono[{\indexIteration+1}]}]
& \indexIteration<\numberOfIterations \\
\end{cases}~,	\label{eq:closedformANFsF}
\end{align}
\begin{align}
\funcGgForANF[\indexIteration](\seqx) &= 
\begin{cases}
\displaystyle
\binaryAsignment[21][{\indexIteration}] (1 - \monomial[{\permutationMono[{\indexIteration}]}]) + 
\binaryAsignment[22][{\indexIteration}]\monomial[{\permutationMono[{\indexIteration}]}]
& \indexIteration=\numberOfIterations \\
\displaystyle
\binaryAsignment[11][{\indexIteration}](1- \monomial[{\permutationMono[{\indexIteration}]}])(1 - \monomial[{\permutationMono[{\indexIteration+1}]}])
\if\IEEEsubmission0
\\~~~~+ 
\else
+
\fi
\binaryAsignment[12][{\indexIteration}]{\monomial[{\permutationMono[{\indexIteration}]}]}(1 - \monomial[{\permutationMono[{\indexIteration+1}]}]) \\~~~~+
\binaryAsignment[21][{\indexIteration}](1 - \monomial[{\permutationMono[{\indexIteration}]}])\monomial[{\permutationMono[{\indexIteration+1}]}] \if\IEEEsubmission0
\\~~~~+ 
\else
+
\fi
\binaryAsignment[22][{\indexIteration}]{\monomial[{\permutationMono[{\indexIteration}]}]}\monomial[{\permutationMono[{\indexIteration+1}]}]
& \indexIteration<\numberOfIterations \\
\end{cases}~,
\label{eq:closedformANFsG}
\end{align}
for $\indexIteration=1,2,\dots,\numberOfIterations$, where $\permutationMono[\indexIteration]=\numberOfIterations - \permutationShift[\indexIteration]$ is the $\indexIteration$th element of the sequence $\seqPermutationCompShift\triangleq(\permutationMono[\indexIteration])_{\indexIteration=1}^{\numberOfIterations}$,   $\binaryAsignment[ij][{\indexIteration}]=0$ if  $\operator[ij][{\indexIteration}] = \operatorBinary[0][{\indexIteration}]$ and  $\binaryAsignment[ij][{\indexIteration}]=1$ if  $\operator[ij][{\indexIteration}] = \operatorBinary[1][{\indexIteration}]$. The vector $ \vecArrangement[\indexIteration]=
[\binaryAsignment[11][{\indexIteration}]~\binaryAsignment[12][{\indexIteration}]~
\binaryAsignment[21][{\indexIteration}]~\binaryAsignment[22][{\indexIteration}] 
]$ is denoted as the configuration vector.

The functions  $\funcGfForANF[\indexIteration](\seqx)$ and $\funcGgForANF[\indexIteration](\seqx)$  show which of the two operators, i.e., $\operatorBinary[0][\indexIteration]$ and $\operatorBinary[1][\indexIteration]$, are involved in the construction of  $\compositeOperatorF[\varMonomial][\functionh]$ and $\compositeOperatorG[\varMonomial][\functionh]$ by setting the indices as $\operatorBinary[{\funcGfForANF[\indexIteration](\seqx)}][\indexIteration]$ and $\operatorBinary[{\funcGgForANF[\indexIteration](\seqx)}][\indexIteration]$, respectively. The binary sequences associated with  $\funcGfForANF[\indexIteration](\seqx)$ and $\funcGgForANF[\indexIteration](\seqx)$ are referred to as the $\indexIteration$th construction sequences (or indication sequences) of   $\functionf[\numberOfIterations]$ and $\functiong[\numberOfIterations]$ for $\indexIteration=1,2,\dots,\numberOfIterations$.

\bibliographystyle{IEEEtran}
\bibliography{reliableControl}

\begin{thebibliography}{10}
\providecommand{\url}[1]{#1}
\csname url@samestyle\endcsname
\providecommand{\newblock}{\relax}
\providecommand{\bibinfo}[2]{#2}
\providecommand{\BIBentrySTDinterwordspacing}{\spaceskip=0pt\relax}
\providecommand{\BIBentryALTinterwordstretchfactor}{4}
\providecommand{\BIBentryALTinterwordspacing}{\spaceskip=\fontdimen2\font plus
\BIBentryALTinterwordstretchfactor\fontdimen3\font minus
  \fontdimen4\font\relax}
\providecommand{\BIBforeignlanguage}[2]{{%
\expandafter\ifx\csname l@#1\endcsname\relax
\typeout{** WARNING: IEEEtran.bst: No hyphenation pattern has been}%
\typeout{** loaded for the language `#1'. Using the pattern for}%
\typeout{** the default language instead.}%
\else
\language=\csname l@#1\endcsname
\fi
#2}}
\providecommand{\BIBdecl}{\relax}
\BIBdecl

\bibitem{sahin_2019icc}
A.~Sahin and R.~Yang, ``A reliable uplink control channel design with
  complementary sequences,'' in \emph{Proc. IEEE International Conference on
  Communications (ICC)}, May 2019.

\bibitem{Lagen_2020}
S.~{Lagen}, L.~{Giupponi}, S.~{Goyal}, N.~{Patriciello}, B.~{Bojović},
  A.~{Demir}, and M.~{Beluri}, ``New radio beam-based access to unlicensed
  spectrum: Design challenges and solutions,'' \emph{IEEE Commun. Surveys
  Tuts.}, vol.~22, no.~1, pp. 8--37, 2020.

\bibitem{Labib_2017}
M.~{Labib}, V.~{Marojevic}, J.~H. {Reed}, and A.~I. {Zaghloul}, ``Extending
  {LTE} into the unlicensed spectrum: Technical analysis of the proposed
  variants,'' \emph{IEEE Communications Standards Magazine}, vol.~1, no.~4, pp.
  31--39, Dec. 2017.

\bibitem{Aijaz_2013}
A.~{Aijaz}, H.~{Aghvami}, and M.~{Amani}, ``A survey on mobile data offloading:
  technical and business perspectives,'' \emph{IEEE Wireless Commun.}, vol.~20,
  no.~2, pp. 104--112, Apr. 2013.

\bibitem{etsi_2017}
ETSI, ``{5 GHz RLAN}; {H}armonized {EN} covering the essential requirements of
  article 3.2 of the {Directive 2014/53/EU}”,'' {EN 301 893}, May 2017.

\bibitem{nr_phy_2020}
3GPP, ``Physical channels and modulation ({Release} 16),'' {TS 38.211 V16.0.0},
  Dec. 2019.

\bibitem{nr_phy_2017}
------, ``Physical channels and modulation ({Release} 15),'' {TS 38.211
  V15.0.0}, Mar. 2017.

\bibitem{Kundu_2018}
L.~{Kundu}, G.~{Xiong}, and J.~{Cho}, ``Physical uplink control channel design
  for {5G} {New Radio},'' in \emph{Proc. IEEE 5G World Forum (5GWF)}, Jul.
  2018, pp. 233--238.

\bibitem{Benedetto_2009}
J.~J. Benedetto, I.~Konstantinidis, and M.~Rangaswamy, ``Phase-coded waveforms
  and their design,'' \emph{IEEE Signal Processing Magazine}, vol.~26, no.~1,
  pp. 22--31, Jan. 2009.

\bibitem{qualcommNRUproposal}
Qualcomm, ``{UL} signals and channels for {NR-U},'' {R1-183412}, Nov. 2018.

\bibitem{nr_coding_2020}
3GPP, ``Multiplexing and channel coding ({Release} 16),'' {TS 38.212 V16.0.0},
  Dec. 2019.

\bibitem{ericssonNRUneedPf0}
Ericsson, ``{UL} signals and channels for {NR-U},'' {R1-1907453}, May 2019.

\bibitem{ericssonNRUproposal}
------, ``{UL} signals and channels for {NR-U},'' {R1-1900997}, Jan. 2019.

\bibitem{featureLead}
------, ``Feature lead summary for {UL} signals and channels,'' {R1-1912715},
  Aug. 2019.

\bibitem{Rahmatallah_2013}
Y.~Rahmatallah and S.~Mohan, ``Peak-to-average power ratio reduction in {OFDM}
  systems: A survey and taxonomy,'' \emph{IEEE Commun. Surveys Tut.}, vol.~15,
  no.~4, pp. 1567--1592, Fourth 2013.

\bibitem{Wunder_2013}
G.~{Wunder}, R.~F.~H. {Fischer}, H.~{Boche}, S.~{Litsyn}, and J.~{No}, ``The
  {PAPR} problem in {OFDM} transmission: New directions for a long-lasting
  problem,'' \emph{IEEE Signal Processing Magazine}, vol.~30, no.~6, pp.
  130--144, Nov. 2013.

\bibitem{Myung_2006}
H.~G. {Myung}, J.~{Lim}, and D.~J. {Goodman}, ``Single carrier {FDMA} for
  uplink wireless transmission,'' \emph{IEEE Vehicular Technology Magazine},
  vol.~1, no.~3, pp. 30--38, Sep. 2006.

\bibitem{davis_1999}
J.~A. Davis and J.~Jedwab, ``Peak-to-mean power control in {OFDM}, {Golay}
  complementary sequences, and {Reed-Muller} codes,'' \emph{IEEE Trans. Inf.
  Theory}, vol.~45, no.~7, pp. 2397--2417, Nov. 1999.

\bibitem{Golay_1961}
M.~Golay, ``Complementary series,'' \emph{IRE Trans. Inf. Theory}, vol.~7,
  no.~2, pp. 82--87, Apr. 1961.

\bibitem{Hori_2018}
Y.~{Hori} and H.~{Ochiai}, ``A new uplink multiple access based on ofdm with
  low papr, low latency, and high reliability,'' \emph{IEEE Transactions on
  Communications}, vol.~66, no.~5, pp. 1996--2008, Jun 2018.

\bibitem{Sahin_2018}
A.~Sahin and R.~Yang, ``A generic complementary sequence encoder,''
  \emph{CoRR}, vol. abs/1810.02383v2, Aug. 2019.

\bibitem{parker_2003}
M.~G. Parker, K.~G. Paterson, and C.~Tellambura, ``Golay complementary
  sequences,'' in \emph{Wiley Encyclopedia of Telecommunications}, 2003.

\bibitem{Sesia2009}
S.~Sesia, I.~Toufik, and M.~Baker, \emph{{LTE}, {T}he {UMTS} {L}ong {T}erm
  {E}volution: {F}rom {T}heory to {P}ractice}.\hskip 1em plus 0.5em minus
  0.4em\relax Wiley Publishing, 2009.

\bibitem{Turyn_1974}
R.~Turyn, ``Hadamard matrices, {Baumert-Hall} units, four-symbol sequences,
  pulse compression, and surface wave encodings,'' \emph{Journal of
  Combinatorial Theory, Series A}, vol.~16, no.~3, pp. 313--333, 1974.

\bibitem{Garcia_2010_ml}
E.~Garcia, J.~J. Garcia, J.~U. A, M.~C. Perez, and A.~Hernandez, ``Generation
  algorithm for multilevel {LS} codes,'' \emph{Electronics Letters}, vol.~46,
  no.~21, pp. 1465--1467, Oct. 2010.

\bibitem{paterson_2000}
K.~G. Paterson, ``Generalized {Reed-Muller} codes and power control in {OFDM}
  modulation,'' \emph{IEEE Trans. Inf. Theory}, vol.~46, no.~1, pp. 104--120,
  Jan. 2000.

\bibitem{holzmann_1991}
W.~H. Holzmann and H.~Kharaghani, ``A computer search for complex {Golay}
  sequences,'' \emph{Aust. Journ. Comb.}, vol.~10, pp. 251--258, Apr. 1994.

\bibitem{interDigitalGCSeval}
InterDigital, ``Evaluation of {CGS} candidates for {PUCCH},'' {R1-1718490},
  Oct. 2017.

\bibitem{Schmidt2005}
K.~{Schmidt} and A.~{Finger}, ``Simple maximum-likelihood decoding of
  generalized first-order {Reed-Muller} codes,'' \emph{IEEE Communications
  Letters}, vol.~9, no.~10, pp. 912--914, Oct. 2005.

\bibitem{eval_NR}
Ericsson, Nokia, Lenovo, {LG}, and {ZTE}, ``{WF} on {DMRS} multi-tone
  evaluation methods,'' {R1-163437}, Apr. 2016.

\end{thebibliography}

\end{document}